\documentclass[11pt]{article}
\usepackage{epsf}
\usepackage{amsmath}
\usepackage{amssymb}
\usepackage{amsfonts}
\usepackage{graphicx}
\usepackage{rotating}
\usepackage{bm}
\usepackage[usenames]{color}
\usepackage[scale={0.8,0.85}]{geometry}
\usepackage{framed}
\definecolor{shadecolor}{rgb}{0.9,0.9,0.9}

\newtheorem{theorem}{Theorem}[section]
\newtheorem{proposition}[theorem]{Proposition}
\newtheorem{corollary}[theorem]{Corollary}
\newtheorem{lemma}[theorem]{Lemma}

\newtheorem{definition}{Definition}[section]
\newtheorem{observation}{Observation}[section]

\newcommand{\qedsymb}{\hfill{\rule{2mm}{2mm}}}
\newenvironment{proof}{\begin{trivlist}
\item[\hspace{\labelsep}{\bf\noindent Proof: }]
}{\qedsymb\end{trivlist}}

\newcommand{\remove}[1]{}
\begin{document}

\title{The Complexity of Equilibria for Risk-Modeling Valuations\thanks{Partially supported
                                                 by the German Research Foundation (DFG)
                                                 within the Collaborative Research Centre
                                                 ``On-the-Fly-Computing'' (SFB 901),
                                                 and by funds for
                                                 the promotion of research
                                                 at University of Cyprus.
                                                }
      }

\author{{\sl Marios Mavronicolas}\thanks{Department of Computer Science,
                                                                   University of Cyprus,
                                                                   Nicosia CY-1678,
                                                                   Cyprus.
                                                                   Email
                                                                   {\tt mavronic@ucy.ac.cy}
                                                                  }
              \and
             {\sl Burkhard Monien}\thanks{Faculty of Electrical Engineering,
                                                             Computer Science
                                                             and Mathematics,
                                                             University of Paderborn,
                                                             33102 Paderborn,
                                                             Germany.
                                                             Email {\tt bm@upb.de}
                                                            }
             }

\date{{\sc (\today)}}

\maketitle

\pagenumbering{arabic}

\begin{abstract}
We study the complexity of deciding
the existence of {\it mixed equilibria}
for minimization games where players
use {\it valuations} other than expectation
to evaluate their {\it costs}.
We consider
{\it risk-averse} players
seeking to minimize
the sum
${\mathsf{V}}
 =
 {\mathsf{E}} +
 {\mathsf{R}}$
of {\it expectation} ${\mathsf{E}}$
and a
{\it risk valuation}
${\mathsf{R}}$
of their costs;
${\mathsf{R}}$ is non-negative
and vanishes exactly when
the cost incurred to a player
is constant over all
choices of strategies
by the other players.
In a {\it ${\mathsf{V}}$-equilibrium,}
no player could
unilaterally reduce her cost.

Say that
${\mathsf{V}}$ has
the {\it Weak-Equilibrium-for-Expectation} property
if all strategies supported
in a player's {\it best-response} mixed strategy
incur the same conditional expectation
of her cost.
We introduce {\it ${\mathsf{E}}$-strict concavity}
and observe that
every ${\mathsf{E}}$-strictly concave valuation
has the
{\it Weak-Equilibrium-for-Expectation} property.
We focus on a broad class of valuations
shown to have
the {\it Weak-Equilibrium-for-Expectation} property,
which we exploit to prove
two main complexity results,
the first of their kind,
for the two simplest cases of the problem:
\begin{itemize}

\item
\underline{Two strategies:}
Deciding the existence of a ${\mathsf{V}}$-equilibrium
is strongly ${\mathcal{NP}}$-hard
for the restricted class of
{\it player-specific scheduling games
      on two ordered links}~\cite{MM12},
when choosing ${\mathsf{R}}$ as
{\sf (1)}
${\mathsf{Var}}$ ({\it variance}),
or
{\sf (2)}
${\mathsf{SD}}$ ({\it standard deviation}),
or
{\sf (3)}
a concave linear sum
of even {\it moments} of small order.

\item
\underline{Two players:}
Deciding the existence of a ${\mathsf{V}}$-equilibrium
is strongly ${\mathcal{NP}}$-hard when
choosing ${\mathsf{R}}$ as
{\sf (1)}
$\gamma \cdot {\mathsf{Var}}$,
or
{\sf (2)}
$\gamma \cdot {\mathsf{SD}}$,
where ${\mathsf{\gamma}} > 0$
is the {\it risk-coefficient,}
or choosing ${\mathsf{V}}$
as {\sf (3)}
a convex combination
of ${\mathsf{E}} + {\mathsf{\gamma}} \cdot {\mathsf{Var}}$
and
the concave {\it ${\mathsf{\nu}}$-valuation}
${\mathsf{\nu}}^{-1}({\mathsf{E}}({\mathsf{\nu}}(\cdot)))$,
where ${\mathsf{\nu}}(x) = x^{r}$,
with $r \geq 2$.
This is a concrete consequence
of a general strong
${\mathcal{NP}}$-hardness result
that only needs the
{\it Weak-Equilibrium-for-Expectation} property
and a few additional properties
for ${\mathsf{V}}$;
its proof involves a reduction
with a single parameter,
which can be chosen efficiently
so that each 
valuation
satisfies the additional properties.

\end{itemize}
\end{abstract}


\section{Introduction}
\label{introduction}

\subsection{The Pros and Cons of Expectation}

In a {\it game,}
each {\it player} is using a {\it mixed strategy,}
a probability distribution
over her {\it strategies};
her {\it cost}
depends on the choices
of all players
and is evaluated by
a {\it valuation}:
a function from probability distributions
to reals.
The most prominent valuation
in {\it Non-Cooperative Game Theory}
is {\it expectation};
each player
minimizes
her expected cost.

A drawback of expectation
is that it may not accomodate {\it risk}
and its impact on strategic decision;
this inadequacy of expectation
was addressed as early as 1738
by Bernoulli~\cite{B1738}.
Indeed,
{\it risk-averse} players~\cite{A71}
are willing to accept a larger amount of payment
rather than gambling
and taking the risk
of a larger cost;
according to~\cite{FS48},
``a risk-averse player is willing
to pay something for certainty''.
So,
valuations other than expectation
have been sought
(cf.~\cite{A71,EGBG07,M52,S63}).
{\it Concave} valuations,
such as {\it variance}
and {\it standard deviation,}
are well-suited
to model risk-averse minimizing players.
Already in 1906,
Fisher~\cite{F06}
proposed attaching
standard deviation
to expectation
as an additive measure of risk.

In his seminal paper~\cite{M52},
Markowitz introduced the {\it Mean-Variance} approach
to portfolio maximization,
advocating the minimization of variance
constrained on some lower bound
on the expected return.
This way,
instead of a single optimal solution,
a class of ``efficient'' solutions,
termed as the {\it Efficient frontier}~\cite{EGBG07},
is defined,
incurring the lowest risk
for a given level of expected return.
Popular valuations
for determining a single maximizing solution
from the {\it Efficient frontier}
are 
{\it (i)}
${\mathsf{E}} - 
     {\mathsf{\gamma}}
     \cdot
     {\mathsf{Var}}$,
where ${\mathsf{E}}$
and ${\mathsf{Var}}$
stand for {\it expectation}
and {\it variance,}
respectively,
and ${\mathsf{\gamma}} > 0$
describes the risk tolerance
(see~\cite{EGBG07}),
and {\it (ii)}
the {\it Sharpe Ratio}
${\mathsf{SR}}
 =
 {\mathsf{E}}/{\mathsf{SD}}$~\cite{S63},
where ${\mathsf{SD}}$
stands for {\it standard deviation}.
The {\it Mean-Variance} paradigm~\cite{M52}
created Modern Portfolio Theory~\cite{EGBG07}
as a new field
and initiated a tremendous amount of research ---
see the surveys~\cite{KZU11,S01}
for an overview.
However,
in the {\it Mean-Variance} paradigm~\cite{M52},
only expectation and variance were used
for evaluating the return;
this choice is justified only if
the return is normally distributed~\cite{HL69}.
Subsequently
this inadequacy led
to risk models
involving higher {\it moments}
so as to accomodate
returns with a more general distribution~\cite{KL76}.

We now switch back
to the minimization setting.
A significant advantage
of expectation is that
it guarantees the existence
of a {\it Nash equilibrium}~\cite{N50,N51},
where each player is playing
a {\it best-response} mixed strategy
and could not unilaterally reduce
her expected cost.
Existence of equilibria
(for minimization games)
extends to
{\it convex} valuations~\cite{D52,F52},
but may fail
for non-convex
and even for {\it concave} ones.
{\it Crawford's game}~\cite[Section 4]{C90}
was the {\em first} counterexample game
with no equilibrium
for a certain valuation;
for more 
counterexamples,
see~\cite{DSS91,MM12}.
The view that {\it mixed} equilibria
get ``endangered''
in games where players are not expectation-optimizers
has been put forward in~\cite{C90}
and adopted further
in~\cite{FP10,MM12}.

Fiat and Papadimitriou~\cite{FP10}
introduced the equilibrium computation problem
in games where risk-averse players
use valuations
more general
than expectation,
and addressed the complexity
of deciding the existence
of such equilibria.\footnote{The work in~\cite{FP10}
                             considered the dual setting
                             where risk-averse players {\it maximize}
                             non-concave valuations;
                             they focused on the convex valuation
                             {\it expectation minus variance}.}
Subsequently,
Mavronicolas and Monien~\cite{MM12}
focused on the concave valuation
{\it expectation plus variance,}
for which
they established structural 
and complexity results
for their introduced class of
{\it player-specific scheduling games}~\cite[Section 3]{MM12};
their results 
provided a solid basis
for the study of more general concave valuations.

\subsection{Valuations More General than Expectation}
\label{sub framework}

In this work,
we shall consider minimization games.
We model
the valuation of each player
as the sum
${\mathsf{V}}
 =
 {\mathsf{E}}
 +
 {\mathsf{R}}$,
where ${\mathsf{E}}$
and ${\mathsf{R}}$ are
the {\it expectation} and {\it risk valuation,}
respectively.
The formulation of
{\it $({\mathsf{E}}
      +
      {\mathsf{R}})$-valuations}
draws motivation from
the {\it Mean-Variance} paradigm~\cite{M52},
and from
the {\it Variance Principle}
and
the {\it Standard Deviation Principle,}
two standard {\it premium principles}
in Actuarial Risk Theory
(cf.~\cite[Section 5.3]{KGDD08}),
by which
${\mathsf{V}}
 =
 {\mathsf{E}}
 +
 {\mathsf{\gamma}}\,
 \cdot
 {\mathsf{Var}}$
(resp.,
${\mathsf{V}}
 =
 {\mathsf{E}}
 +
 {\mathsf{\gamma}}\,
 \cdot
 {\mathsf{SD}}$),
where
${\mathsf{\gamma}} > 0$
is the {\it risk-coefficient}.
We focus on
the associated decision problem,
denoted as
{\sf $\exists {\mathsf{V}}$-EQUILIBRIUM},
asking, given a game ${\mathsf{G}}$,
whether ${\mathsf{G}}$ has
a {\it ${\mathsf{V}}$-equilibrium,}
where no player could unilaterally reduce her cost
(as evaluated by ${\mathsf{V}}$).
{\it What is the impact
     of properties of ${\mathsf{V}}$
     on the complexity of
     {\sf $\exists {\mathsf{V}}$-EQUILIBRIUM}?}\footnote{Strictly speaking, 
                                                   some of the considered properties of ${\mathsf{V}}$ 
                                                   are rather properties of the equilibria of some game ${\mathsf{G}}$
                                                   whose players minimize ${\mathsf{V}}$.
                                                   For ease of presentation,
                                                   we shall omit reference to ${\mathsf{G}}$
                                                   since ${\mathsf{G}}$ will be either fixed
                                                   or clear from context in the settings we shall consider,
                                                   so that the property depends only on ${\mathsf{V}}$.}

We stipulate a very basic
property for ${\mathsf{R}}$,
called {\it Risk-Positivity}:
the value of ${\mathsf{R}}$
is non-negative,
en par with
the {\it Non-negative Loading} property
of the premium principles
in~\cite[Section 5.3.1]{KGDD08};
it is $0$,
yielding no risk,
exactly when
the cost incurred to a player
is constant
over all choices of strategies
by the other players.

We shall focus on concave valuations.
A key property of a concave valuation,
called {\it Optimal-Value},
we prove and exploit
is that
it maintains the same optimal value
over all convex combinations
of strategies supported
in a given best-response mixed strategy
(Proposition~\ref{constant value}).
Unfortunately,
unlike ${\mathsf{Var}}$,
moments
of order higher than $2$
are not concave.
But on the positive side,
all even moments
have the {\it Risk-Positivity} property.
Besides,
recent work in Portfolio Theory~\cite{KL76}
motivates using higher moments
to model risk.

To obtain an enhanced class 
of interesting concave valuations,
we introduce into the context
of equilibrium computation
valuations prominent
for evaluating risk
in 
Actuarial Risk Theory,
which are transferred from the
{\it Mean Value Principle}
(see~\cite[Section 5.3]{KGDD08}).
Specifically,
we shall consider
{\it ${\mathsf{\nu}}$-valuations}
of the form
${\mathsf{V}}^{{\mathsf{\nu}}}
 =
 {\mathsf{\nu}}^{-1}({\mathsf{E}}({\mathsf{\nu}}(\cdot)))$,
for any increasing and
strictly convex function
${\mathsf{\nu}}$,
so that
${\mathsf{\nu}}^{-1}$
is concave.\footnote{For the special case
                                           ${\mathsf{\nu}}(x)
                                            =
                                            e^{x}$,
                                           ${\mathsf{V}}^{{\mathsf{\nu}}}$
                                           corresponds to the {\it moment generating function}
                                           (cf.~\cite[Section 2.4]{KGDD08})
                                           and has gained special attention
                                           as the {\it Exponential Principle} premium
                                           in Actuarial Risk Theory~\cite[Section 5.3]{KGDD08}.}
It is good news that
for 
${\mathsf{\nu}}$ increasing
and strictly convex,
${\mathsf{R}}^{{\mathsf{\nu}}}
 :=
 {\mathsf{V}}^{{\mathsf{\nu}}} - {\mathsf{E}}$
has the {\it Risk-Positivity} property
(Lemma~\ref{paderborn});
so,
${\mathsf{V}}^{{\mathsf{\nu}}}$
is an
$({\mathsf{E}}
  +
  {\mathsf{R}})$-valuation.

\subsection{Weak-Equilibrium-for-Expectation and ${\mathsf{E}}$-Strict Concavity}

As our main tool,
we shall exploit
the {\it Weak-Equilibrium-for-Expectation} property~\cite[Section 2.6]{MM12}:
all strategies supported 
in a player's best-response mixed strategy
induce the same conditional expectation,
taken over all random choices
of the other players,
for her cost;
thus,
the player is holding
a {\em unique} expectation
for her cost
no matter which of her supported strategies
she ends up choosing.
This property
holds vacuously for Nash equilibria~\cite{N50,N51}.
The {\it Weak-Equilibrium-for-Expectation} property
formalizes the most natural intuition
for the players' expectations;
hence,
it is a very natural property
to seek and employ 
in the setting of risk.
We aim at 
an enhanced class of valuations
with the
{\it Weak-Equilibrium-for-Expectation} property.

We introduce an {\it ${\mathsf{E}}$-strictly concave valuation}
as a concave valuation which,
viewed as a function of a single mixed strategy,
fulfills the definition
of a strictly concave function
for any pair of mixed strategies
inducing different expectations
(Definition~\ref{e strict concavity}).
We observe that
a convex combination
of an ${\mathsf{E}}$-strictly concave valuation
and a concave valuation
is ${\mathsf{E}}$-strictly concave
(Corollary~\ref{hilton frappe});
furthermore,
${\mathsf{Var}}$ and ${\mathsf{SD}}$
are ${\mathsf{E}}$-strictly concave
(Lemma~\ref{paderborn 1}).
Hence,
a wide class of concrete instances
of ${\mathsf{E}}$-strictly concave valuations
results by plugging in
the convex combination
{\it (i)}
${\mathsf{E}} + {\mathsf{\gamma}} \cdot {\mathsf{Var}}$
for an ${\mathsf{E}}$-strictly concave valuation,
with ${\mathsf{\gamma}} > 0$,
and
{\it (ii)}
a ${\mathsf{\nu}}$-valuation,
with an increasing and strictly convex function ${\mathsf{\nu}}$,
for a concave valuation
(Corollary~\ref{standart deviation}).
We establish the key fact
that every ${\mathsf{E}}$-strictly concave valuation
has the {\it Weak-Equilibrium-for-Expectation} property
(Proposition~\ref{hilton}).
${\mathsf{E}}$-strictly concave valuations
make the largest subclass of concave valuations
we know of
with the {\it Weak-Equilibrium-for-Expectation} property.

An obstacle to extending
the {\it Weak-Equilibrium-for-Expectation} property
to moments of order higher than $2$
is their non-concavity.
Instead we consider
concave linear sums of even moments.
(Even order is needed
to guarantee the
{\it Risk-Positivity} property.)
We use the {\it Optimal-Value} property
(Proposition~\ref{constant value})
to establish that
such concave valuations
have the
{\it Weak-Equilibrium-for-Expectation} property
(Corollary~\ref{concavity implies weak equilibrium});
this property renders
such concave linear sums of even moments
sufficiently interesting to consider.

\subsection{Complexity Results for More General Valuations}

By exploiting 
the {\it Weak-Equilibrium-for-Expectation} property
for an $({\mathsf{E}} + {\mathsf{R}})$-valuation ${\mathsf{V}}$,
we shall show 
the strong ${\mathcal{NP}}$-hardness of
{\sf $\exists {\mathsf{V}}$-EQUILIBRIUM}
for the two simplest cases:
games with two strategies
and games with two players;
these are the {\em first} complexity results
for deciding the existence of equilibria
in the context of risk-modeling valuations
(cf.~Section~\ref{related work and comparison}).

\subsubsection{Two Strategies}

We discover that
the complexity of
{\sf $\exists {\mathsf{V}}$-EQUILIBRIUM}
is captured by
{\it player-specific scheduling games on two ordered links}~\cite[Section 5.2.2]{MM12},
where the cost of player $i$ on a {\it link} $\ell$
she chooses
is a sum of {\it weights} $\omega (i, i', \ell)$,
each corresponding to player $i$, link $\ell$
and a player $i'$
choosing the same link.\footnote{In the well-known model
                            of {\it weighted congestion games with player-specific latency functions}~\cite{M96},
                            each weighted player may use a different, {\it player-specific} cost function
                            of the total weight
                            on her selected link;
                            in this model,
                            it is the weights that are {\it player-specific,}
                            while each cost function
                            is the identity one.}
{\it Two ordered links}~\cite[Section 5.2.2]{MM12} means that
link $1$ incurs less weight
than link $2$
to a player
due to another player choosing the same link
unless both weights are $0$.
We show that
for an $({\mathsf{E}} + {\mathsf{R}})$-valuation ${\mathsf{V}}$,
{\sf $\exists {\mathsf{V}}$-EQUILIBRIUM}
is strongly ${\mathcal{NP}}$-hard
for the class of player-specific scheduling games
on two ordered links
when
${\mathsf{R}}$ is
{\sf (1)}
${\mathsf{Var}}$,
or
{\sf (2)}
${\mathsf{SD}}$,
or
{\sf (3)}
a concave linear sum
of even moments
of order $k \in \{ 2, 4, 6, 8 \}$
(Theorem~\ref{concrete two ordered links}).

Instrumental to the proof
of Theorem~\ref{concrete two ordered links}
is the key property we show
that a concave linear sum
of even moments
of order $k \in \{ 2, 4, 6, 8 \}$
enjoys the {\it Mixed-Player-Has-Pure-Neighbors} property:
each player $i$ either is pure
or has all of her neighbors
(that is, players $i'$
with $\omega (i, i', 1) \neq 0$)
pure
(Proposition~\ref{the real almost pure property}).
This property
is a quantitative expression
of the view
that mixed equilibria
get ``endangered''
when players are not expectation-optimizers
(cf.~\cite{C90,FP10,MM12}).
The class of
concave linear sum
of even moments
of order $k \in \{ 2, 4, 6, 8 \}$
is the largest class of valuations
we were able to identify with the
{\it Mixed-Player-Has-Pure-Neighbors} property.
We conjecture that
every concave linear sum of even moments
enjoys the property.

\subsubsection{Two Players}

We show that
{\sf $\exists {\mathsf{V}}$-EQUILIBRIUM}
is strongly ${\mathcal{NP}}$-hard
when ${\mathsf{V}}$
is an $({\mathsf{E}} + {\mathsf{R}})$-valuation ${\mathsf{V}}$
with the {\it Weak-Equilibrium-for-Expectation} property
provided that
there is a
polynomial time computable ${\mathsf{\delta}}$,
$0 < {\mathsf{\delta}}
   \leq \frac{\textstyle 1}
                  {\textstyle 4}$,
such that three additional conditions hold
(Theorem~\ref{two players complexity});
the requirement that ${\mathsf{\delta}}$
be polynomial time computable
stems from the fact that
${\mathsf{\delta}}$ enters the reduction
as a parameter.
The first two such conditions
({\sf (2/a)} and {\sf (2/b)})
stipulate a particular inequality
and a particular 
monotonicity property,
respectively.
The third condition
{\sf (2/c)}
refers to
the {\it Crawford game} ${\mathsf{G}}_{C}({\mathsf{\delta}})$,
a generalization of 
a bimatrix game from~\cite[Section 4]{C90},
whose bimatrix involves ${\mathsf{\delta}}$;
it is required that
${\mathsf{G}}_{C}({\mathsf{\delta}})$
has no ${\mathsf{V}}$-equilibrium.
The game ${\mathsf{G}}_{C}({\mathsf{\delta}})$
is used as a ``gadget''
in the reduction.
The proof of Theorem~\ref{two players complexity},
involving a reduction with a single parameter ${\mathsf{\delta}}$,
is very general
since it refers to no particular valuation
but to a class of valuations
enjoying two natural properties,
{\it Risk-Positivity} and
{\it Weak-Equilibrium-for-Expectation}.

Concrete strong ${\mathcal{NP}}$-hardness results
follow as instantiations
of Theorem~\ref{two players complexity}
for three particular valuations
{\sf (1)}
${\mathsf{V}}
 =
 {\mathsf{E}}
 +
 {\mathsf{\gamma}}
 \cdot
 {\mathsf{Var}}$,
{\sf (2)}
${\mathsf{V}}
 =
 {\mathsf{E}}
 +
 {\mathsf{\gamma}}
 \cdot
 {\mathsf{SD}}$,
and
{\sf (3)}
${\mathsf{V}}
 =
 \lambda
 \cdot
 \left( {\mathsf{E}}
        +
        {\mathsf{\gamma}}
        \cdot
        {\mathsf{Var}}
 \right)
 +
 (1 - \lambda)
 \cdot
 {\mathsf{V}}^{{\mathsf{\nu}}}$,
with $0 < \lambda \leq 1$,
where ${\mathsf{\nu}}$
is the increasing and strictly convex function
${\mathsf{\nu}}(x) = x^{r}$,
with $r \geq 2$,
and with
${\mathsf{\gamma}} > 0$
(Theorem~\ref{call from paderborn}).
For all three valuations
in Theorem~\ref{call from paderborn},
we prove that
the three additional conditions
in Theorem~\ref{two players complexity}
hold.
In particular,
for Condition {\sf (2/c)},
we prove 
that ${\mathsf{G}}_{C}({\mathsf{\delta}})$
has no ${\mathsf{V}}$-equilibrium
for {\it any} value
$0 < {\mathsf{\delta}} < 1$;
for the valuation {\sf (3)},
this holds in the more general case 
${\mathsf{\nu}}$
is any increasing
and strictly convex function
(Lemma~\ref{no v equilibrium for crawford}).

\subsection{Summary, Significance and Related Work}
\label{related work and comparison}

For a concave valuation ${\mathsf{V}}$,
there may or may not exist a ${\mathsf{V}}$-equilibrium;
we have identified  
${\mathsf{V}}^{{\mathsf{\nu}}}$,
with an increasing and strictly convex function
${\mathsf{\nu}}$,
as an example of a concave valuation ${\mathsf{V}}$
for which every game has
a ${\mathsf{V}}$-equilibrium
(cf. Section~\ref{two players}).
However,
restricting to a {\it strictly concave} valuation 
excludes the existence
of a mixed ${\mathsf{V}}$-equilibrium.
Restricting instead to an
${\mathsf{E}}$-strictly concave valuation ${\mathsf{V}}$,
a mixed ${\mathsf{V}}$-equilibrium
may or may not exist;
{\it
what this work is revealing
is that it then becomes
strongly ${\mathcal{NP}}$-hard to decide
if there is one
already for the two simplest cases,
games with two strategies or two players
}
(Theorems~\ref{concrete two ordered links},~\ref{two players complexity}
and~\ref{call from paderborn}).

While bringing 
concave valuations from
Actuarial Risk Theory~\cite{KGDD08}
into play,
our framework encompasses
general classes
of 
$({\mathsf{E}}
  +
  {\mathsf{R}})$-valuations,
assuming the
{\it Risk-Positivity} property,
that also enjoy the 
{\it Weak-Equilibrium-for-Expectation} property,
and a few additional properties.
The $({\mathsf{E}} + {\mathsf{Var}})$-equilibria
on which~\cite{FP10,MM12} focused
are a special case of our general framework.

Fiat and Papadimitriou~\cite[Theorem 5]{FP10}
presented a proof sketch to claim
that it is ${\mathcal{NP}}$-hard
to decide the existence of
an $({\mathsf{E}} + {\mathsf{Var}})$-equilibrium
for games with two players;
unfortunately,
their proof had been flawed,
containing several gaps and errors.
In personal communication
with Fiat and Papadimitriou~\cite{FP15},
they state:
{\it
``The proof, as is,
has gaps and errors,
which we believe can be fixed
to yield a proof with the same architecture,
but we have not done it yet.''}
In lack of a proof,
the complexity of deciding the existence
of an $({\mathsf{E}} + {\mathsf{Var}})$-equilibrium
for games with two players
has remained open,
and Theorems~\ref{two players complexity}
and~\ref{call from paderborn}
represent new results,
with Theorem~\ref{call from paderborn}
encompassing
$({\mathsf{E}} +
  {\mathsf{Var}})$-equilibria
as a special case. 
Our reduction
adapts techniques
originally used by Conitzer and Sandholm~\cite[Section 3]{CS08}
to show that deciding the existence of Nash equilibria
with certain properties
for games with two players is ${\mathcal{NP}}$-complete.

Fiat and Papadimitriou~\cite[Section 2]{FP10}
coined a notion termed as {\it strict concavity,}
denoted here as {\it ${\mathsf{FP}}$-strict concavity,}
which is similar to 
but different than
${\mathsf{E}}$-strict concavity.
It turns out that
their difference is essential
since
${\mathsf{E}} + {\mathsf{Var}}$
is {\em not} ${\mathsf{FP}}$-strictly concave
while it is ${\mathsf{E}}$-strictly concave.\footnote{See the Appendix
                                                                                     for the definition of ${\mathsf{FP}}$-strict concavity
                                                                                     and a proof
                                                                                     that ${\mathsf{E}} + {\mathsf{Var}}$
                                                                                     is not ${\mathsf{FP}}$-strictly concave.}
In fact,
no concrete example
of an ${\mathsf{FP}}$-strictly concave valuation
was given in~\cite{FP10}.
Fiat and Papadimitriou~\cite[Theorem 3 \& Observation 4]{FP10}
proved the sparsity
of {\it mixed} ${\mathsf{V}}$-equilibria,
when ${\mathsf{V}}$ is
${\mathsf{FP}}$-strictly concave:
games with a mixed ${\mathsf{V}}$-equilibrium
have measure $0$.
The sparsity of mixed $({\mathsf{E}} + {\mathsf{Var}})$-equilibria
follows from~\cite[Theorem 1]{BKS10}
where it is established that  
$({\mathsf{E}} + {\mathsf{Var}})$
is a {\it Mean-Variance Preference Function}~\cite[Claim 1]{BKS10}.\footnote{See
Section~\ref{wee}
for the definition of Mean-Variance Preference Functions
and their relation to this work.}
Contrary to their sparsity,
this work establishes that
deciding the existence
of mixed $({\mathsf{E}} + {\mathsf{Var}})$-equilibria
is strongly ${\mathcal{NP}}$-hard
(Theorems~\ref{concrete two ordered links}
and~\ref{call from paderborn}).

It was known that
${\mathsf{E}} + {\mathsf{Var}}$,
as well as
${\mathsf{E}} +
 {\mathsf{SD}}$,
have the {\it Weak-Equilibrium-for-Expectation} property~\cite[Theorem 3.5]{MM12};
they were also known
to have the {\it Mixed-Player-Has-Pure-Neighbors} property
for player-specific scheduling games
on two ordered links~\cite[Theorem 5.13]{MM12}.
It was established in~\cite[Theorem 4]{BKS10}
that there is always a {\it correlated equilibrium}~\cite{A87}
for ${\mathsf{E}} + {\mathsf{Var}}$.

\subsection{Paper Organization}
\label{road map}

Section~\ref{framework}
presents the game-theoretic framework
and introduces $({\mathsf{E}} + {\mathsf{R}})$-valuations
and ${\mathsf{E}}$-strict concavity.
Equilibria and their properties
are articulated
in Section~\ref{equilibria}.
The ${\cal NP}$-hardness results
for two strategies and two players
are presented in
Sections~\ref{two ordered links} and~\ref{two players},
respectively.
We conclude,
in Section~\ref{epilogue},
with a discussion of the results
and some open problems.

\section{Games, $({\mathsf{E}} + {\mathsf{R}})$-Valuations and ${\mathsf{E}}$-Strict Concavity}
\label{framework}

\subsection{Games}

For an integer $n \geq 2$,
an {\it $n$-players game}
${\mathsf{G}}$,
or {\it game,}
consists of
{\it (i)}
$n$ finite sets
$\left\{ S_{i}
\right\}_{i \in [n]}$
of {\it strategies,}
and
{\it (ii)}
$n$ {\it cost functions}
$\left\{ {\mathsf{\mu}}_{i}
\right\}_{i \in [n]}$,
each mapping
${\cal S} = \times_{i \in [n]}
              S_{i}$
to the reals.
So,
${\mathsf{\mu}}_{i}({\bf s})$
is the {\it cost}
of player $i \in [n]$
on the {\it profile}
${\bf s}
=
\langle s_{1},
    \ldots,
    s_{n}
\rangle$
of strategies,
one per player.

A {\it mixed strategy}
for player $i \in [n]$
is a probability distribution $p_{i}$
on $S_{i}$;
the {\it support}
of player $i$
in $p_{i}$
is the set
$\sigma (p_{i})
=
\left\{ \ell \in S_{i}\
    \mid\
    p_{i}(\ell) > 0
\right\}$.
Denote as
${\mathsf{\Delta}}_{i}
=
{\mathsf{\Delta}}(S_{i})$
the set of
mixed strategies
for player $i$.
Player $i$
is {\it pure}
if for each strategy $s_{i} \in S_{i}$,
$p_{i}(s_{i}) \in \{ 0, 1 \}$;
else she is {\it non-pure}.
Denote as $p_{i}^{\ell}$
the {\it pure strategy}
of player $i$
choosing the strategy $\ell$
with probability $1$.

A {\it mixed profile}
is a tuple
${\bf p} = \langle p_{1},
         \ldots,
         p_{n}
     \rangle$
of $n$ mixed strategies,
one per player;
denote as
${\mathsf{\Delta}}
=
{\mathsf{\Delta}}(S)
=
\times_{i \in [n]}
 {\mathsf \Delta}_{i}$
the set of mixed profiles.
The mixed profile ${\bf p}$
induces
probabilities
${\bf p}({\bf s})$
for each profile
${\bf s} \in {\cal S}$
with
${\bf p}({\bf s})
=
\prod_{i' \in [n]}
  p_{i'} (s_{i'})$.
For a player $i \in [n]$,
the {\it partial profile}
${\bf s}_{-i}$
(resp.,
{\it partial mixed profile}
${\bf p}_{-i}$)
results by eliminating
the strategy $s_{i}$
(resp.,
the mixed strategy $p_{i}$)
from ${\bf s}$
(resp.,
${\bf p}$).
The partial mixed profile ${\bf p}_{-i}$
induces
probabilities
${\bf p}({\bf s}_{-i})$
for each partial profile
${\bf s}_{-i}
 \in
 {\cal S}_{-i}
 :=
 \times_{i' \in [n] \setminus \{ i \}}
  S_{i'}$
with
${\bf p}({\bf s}_{-i})
=
\prod_{i' \in [n] \setminus \{ i \}}
  p_{i'} (s_{i'})$.

A function
${\mathsf{V}} : {\mathsf{T}}
                \rightarrow
                {\mathbb{R}}$
on a convex set ${\mathsf{T}}$
is {\it concave}
(resp., {\it strictly concave})
if for any two points
$t_{1}, t_{2}
\in
{\mathsf{T}}$
and any number $\delta \in [0, 1]$
(resp., $\delta \in (0, 1)$),
${\mathsf{V}}(\delta\,t_{1} + (1-\delta)\, t_{2})
\geq
\delta\, {\mathsf{V}}(t_{1})
+
(1-\delta)\, {\mathsf{V}}(t_{2})$
(resp.,
${\mathsf{V}}(\delta\,t_{1} + (1-\delta)\, t_{2})
 >
 \delta\, {\mathsf{V}}(t_{1})
 +
 (1-\delta)\, {\mathsf{V}}(t_{2})$).
A function
${\mathsf{V}} : {\mathsf{T}}
                \rightarrow
                {\mathbb{R}}$
on a convex set ${\mathsf{T}}$
is {\it convex}
(resp., {\it strictly convex})
if $-{\mathsf{V}}$ is concave
(resp., strictly concave).

\subsection{$({\mathsf{E}} + {\mathsf{R}})$-Valuations}

For a player $i \in [n]$,
a {\it valuation function,}
or {\it valuation} for short,
${\mathsf V}_{i}$
is a real-valued function
on $\mathsf{\Delta}({\cal S})$,
yielding a value
${\mathsf V}_{i}({\bf p})$
to each mixed profile
${\bf p}$,
so that
in the special case where ${\bf p}$ 
is a profile ${\bf s}$,
${\mathsf V}_{i}({\bf s})
=
{\mathsf{\mu}}_{i}({\bf s})$.
A {\it valuation}
${\mathsf V}
=
\langle {\mathsf V}_{1},
    \ldots,
    {\mathsf V}_{n}
\rangle$
is a tuple of valuations,
one per player;
${\mathsf{G}}^{{\mathsf{V}}}$
denotes ${\mathsf{G}}$
together with ${\mathsf{V}}$.\footnote{We shall mostly treat
                                       a valuation function ${\mathsf{V}}_{i}$
                                       and
                                       a valuation ${\mathsf{V}}$
                                       interchangeably
                                       for an easier notation;
                                       we shall use ${\mathsf{V}}_{i}$
                                       only when $p_{i}$
                                       has some special property.}
We now introduce
a special class of valuations.

\begin{center}
\fbox{
\begin{minipage}{5.9in}
\begin{definition}[$({\mathsf{E}} + {\mathsf{R}}$)-Valuation]
\label{most basic}
An $({\mathsf{E}} + {\mathsf{R}}$)-valuation
is a valuation
of the form
${\mathsf{V}}
 =
 {\mathsf{E}}
 +
 {\mathsf{R}}$,
where ${\mathsf{E}}$
is the {\it expectation valuation}
with
${\mathsf{E}}_{i}({\bf p})
 = 
 \sum_{{\bf s}
       \in
       S}
   {\bf p}({\bf s})\,
   {\mathsf{\mu}}_{i}({\bf s})$
for $i \in [n]$,
and
${\mathsf{R}}$
is the {\it risk valuation,}
a continuous valuation
with the
{\it Risk-Positivity} property:
For each player $i \in [n]$
and mixed profile ${\bf p}$,
{\sf (C.1)}
${\mathsf{R}}_{i}({\bf p})
 \geq
 0$
and
{\sf (C.2)}
${\mathsf{R}}_{i}({\bf p})
 =
 0$
if and only if
for each profile
${\bf s}
 \in
 {\cal S}$
with
${\bf p} ({\bf s}) > 0$,
${\mathsf{\mu}}_{i}({\bf s})$
remains constant
over all choices of strategies
by the other players;
in such case,
${\mathsf{V}}_{i}({\bf p})
 =
 {\mathsf{E}}_{i}({\bf p})
 =
 {\mathsf{\mu}}_{i}({\bf s})$
for any profile
${\bf s}
 \in
 {\cal S}$
with
${\bf p} ({\bf s}) > 0$.
\end{definition}
\end{minipage}
}
\end{center}

\noindent
For each integer $k \geq 0$,
the {\it $k$-moment} valuation
is given by
{
\small
\begin{eqnarray*}
{\mathsf{kM}}_{i}({\bf p})
& = &
   \sum_{{\bf s}
       \in
       S}
   {\bf p}({\bf s})\,
   \left( {\mathsf{\mu}}_{i}({\bf s})
          -
          {\mathsf{E}}_{i}({\bf p})
   \right)^{k}\, ,
\end{eqnarray*}
}
for each player $i \in [n]$;
so,
${\mathsf{1M}} = 0$.
Furthermore,
${\mathsf{2M}}$,
known as {\it variance}
and denoted as ${\mathsf{Var}}$,
is concave;
hence,
also is
the square root of variance,
known as {\it standard deviation}
and denoted as ${\mathsf{SD}}$.
However,
$k$-moments
of order higher than $2$
are not concave.

We shall consider {\it ${\mathsf{\nu}}$-valuations}
${\mathsf{V}}^{{\mathsf{\nu}}}
 =
 {\mathsf{\nu}}^{-1}({\mathsf{E}}({\mathsf{\nu}}(.)))$,
for an increasing and strictly convex function
${\mathsf{\nu}}$;
so,
${\mathsf{\nu}}^{-1}$,
and hence
${\mathsf{V}}^{{\mathsf{\nu}}}$,
is concave.
So,
for a player $i$,
{
\small
\begin{eqnarray*}
            {\mathsf{V}}^{{\mathsf{\nu}}}_{i}\left( {\bf p}
                                       \right)
& = &  {\mathsf{\nu}}^{-1} \left( \sum_{{\bf s}
                                       \in
                                       {\cal S}}
                                   {\bf p} \left( {\bf s}
                                           \right)
                                 \cdot
                                 {\mathsf{\nu}}\left( {\mathsf{\mu}}_{i}({\bf s})
                                               \right)
                          \right)\, ;
\end{eqnarray*}
}
set also
${\mathsf{R}}^{{\mathsf{\nu}}}
 :=
 {\mathsf{V}}^{{\mathsf{\nu}}}
 -
 {\mathsf{E}}$.
We observe:

\begin{lemma}
\label{paderborn}
For an increasing and strictly convex function
${\mathsf{\nu}}$,
the risk valuation
${\mathsf{R}}^{{\mathsf{\nu}}}$
has the {\it Risk-Positivity} property.
\end{lemma}

\begin{proof}
Fix a player $i \in [n]$
and a profile
${\bf p}$.
Then,
${\mathsf{R}}^{{\mathsf{\nu}}}_{i}\left( {\bf p}
                                  \right)
 \geq
 0$
if and only if
${\mathsf{V}}^{{\mathsf{\nu}}}_{i}\left( {\bf p}
                                  \right)
 \geq
 {\mathsf{E}}_{i}\left( {\bf p}
                 \right)$
if and only if
{
\small
\begin{eqnarray*}
         {\mathsf{\nu}}^{-1} \left( \sum_{{\bf s}
                                        \in
                                        {\cal S}}
                                      {\bf p} \left( {\bf s}
                                              \right)
                                      \cdot
                                      {\mathsf{\nu}}\left( {\mathsf{\mu}}_{i}({\bf s})
                                                    \right)
                             \right)
& \geq  & \sum_{{\bf s}
               \in
               {\cal S}}
           {\bf p} \left( {\bf s}
                   \right)
           \cdot
           {\mathsf{\mu}}_{i}({\bf s})
\end{eqnarray*}
}
if and only if (since ${\mathsf{\nu}}$ is increasing)
{
\small
\begin{eqnarray*}
         \sum_{{\bf s}
               \in
               {\cal S}}
           {\bf p} \left( {\bf s}
                   \right)
           \cdot
           {\mathsf{\nu}}\left( {\mathsf{\mu}}_{i}({\bf s})
                         \right)
& \geq & {\mathsf{\nu}} \left( \sum_{{\bf s}
                                     \in
                                     {\cal S}
                                    }
                                p \left( {\bf s}
                                  \right)
                                \cdot
                                {\mathsf{\mu}}_{i}({\bf s})
                       \right)\, .
\end{eqnarray*}
}
Now {\sf (C.1)} follows since
${\mathsf{\nu}}$ is convex;
{\sf (C.2)} follows since
${\mathsf{\nu}}$ is strictly convex.
\end{proof}

\noindent
By Lemma~\ref{paderborn},
for an increasing and strictly convex function
${\mathsf{\nu}}$,
${\mathsf{V}}^{{\mathsf{\nu}}}$
is an $({\mathsf{E}} + {\mathsf{R}})$-valuation.

\noindent
We shall deal 
with cases where
for a player $i \in [n]$
and a mixed profile ${\bf p}$,
$\{ {\mathsf{\mu}}_{i}({\bf s})
    \mid
    {\bf p} ({\bf s}) > 0
 \}
 =
 \{ a, b
 \}$
with $a < b$,
so that
${\mathsf{R}}_{i}({\bf p})$
depends on the three parameters
$a$, $b$ and $q$,
where
$q 
 :=
 \sum_{{\bf s} \in {\cal S}\,
       \mid\,
       {\mathsf{\mu}}_{i}({\bf s})
       =
       b}
   {\bf p} ({\bf s})$.
Then, 
denote
{
\small
\begin{eqnarray*}
            \widehat{{\mathsf{V}}}_{i}(a,b,q)
& := & a + q (b-a)
           +
            \widehat{{\mathsf{R}}}_{i}(a,b,q)\, .
\end{eqnarray*}
}

\subsection{\bf ${\mathsf{E}}$-Strict Concavity}

We introduce
a refinement of concavity:

\begin{center}
\fbox{
\begin{minipage}{5.9in}
\begin{definition}[${\mathsf{E}}$-Strict Concavity]
\label{e strict concavity}
Fix a player $i$. 
The $({\mathsf{E}} +
      {\mathsf{R}})$-valuation
${\mathsf{V}}_{i}$
is {\it ${\mathsf{E}}$-strictly concave}
if for every game ${\mathsf{G}}$,
the following conditions hold
for a fixed
partial mixed profile
${\bf p}_{-i}$:

\begin{enumerate}

\item[{\sf (1)}]
${\mathsf{V}}_{i}$
is concave in the mixed strategy $p_{i}$.

\item[{\sf (2)}]
For a pair of mixed strategies
$p_{i}^{\prime},
 p_{i}^{\prime\prime}
 \in
 {\mathsf{\Delta}} (S_{i})$,
if
${\mathsf{E}}_{i} \left( p_{i}^{\prime},
                         {\bf p}_{-i}
                  \right)
 \neq
 {\mathsf{E}}_{i} \left( p_{i}^{\prime\prime},
                         {\bf p}_{-i}
                  \right)$,
then for any $\lambda$ with
$0 < \lambda < 1$,
{
\small
\begin{eqnarray*}
      {\mathsf{V}}_{i} \left( \lambda  p_{i}^{\prime}
                              +
                              (1 - \lambda) p_{i}^{\prime\prime},
                              {\bf p}_{-i}
                       \right)
& > & \lambda
      {\mathsf{V}}_{i} \left( p_{i}^{\prime},
                             {\bf p}_{-i}
                       \right)
      +
      (1 - \lambda)
      {\mathsf{V}}_{i} \left( p_{i}^{\prime\prime},
                              {\bf p}_{-i}
                       \right)\, .
\end{eqnarray*}
}

\end{enumerate}
\end{definition}
\end{minipage}
}
\end{center}

\noindent
Note that
${\mathsf{E}}$-strict concavity
is different from the
{\it strict concavity}
formulated in~\cite[Section 2]{FP10},
and denoted here as
${\mathsf{FP}}$-strict concavity,
by using the
{\it payoff distribution}
${\mathsf{P}}_{i}$,
which is the probability distribution
on the costs
induced by a mixed strategy $p_{i}$
and a partial mixed profile
${\bf p}_{-i}$.\footnote{In Appendix~\ref{fp strict concavity},
                                we provide a counterexample
                                to demonstrate
                                that the valuation
                                ${\mathsf{E}}
                                 +
                                 {\mathsf{Var}}$
                                is {\em not}
                                ${\mathsf{FP}}$-strictly concave.}
A closure property
of ${\mathsf{E}}$-strict concavity 
follows.

\begin{corollary}
\label{hilton frappe}
Fix an arbitrary pair of
an ${\mathsf{E}}$-strictly concave
valuation
${\mathsf{V}}^{(1)}$
and a concave valuation
${\mathsf{V}}^{(2)}$.
Then,
for any $\lambda$,
with $0 < \lambda \leq 1$,
the valuation
${\mathsf{V}}
 =
 \lambda\,
 {\mathsf{V}}^{(1)}
 +
 (1- \lambda)\,
 {\mathsf{V}}^{(2)}$
is
${\mathsf{E}}$-strictly concave.
\end{corollary}

\noindent
We observe:

\begin{lemma}
\label{paderborn 1}
The valuations
${\mathsf{E}} + {\mathsf{\gamma}} \cdot {\mathsf{Var}}$
and
${\mathsf{E}} + {\mathsf{\gamma}} \cdot {\mathsf{SD}}$,
with ${\mathsf{\gamma}} > 0$,
are
${\mathsf{E}}$-strictly concave.
\end{lemma}

\begin{proof}
Fix a game ${\mathsf{G}}$
and a player $i \in [n]$.
Recall that
${\mathsf{Var}}$
and
${\mathsf{SD}}$
are concave in the
mixed strategy $p_{i}$.
Also,
{
\small
\begin{eqnarray*}
{\mathsf{Var}}_{i}\left( p_{i},
                          {\bf p}_{-i}
                   \right)
& = &
 \sum_{{\bf s}
       \in
       {\cal S}}
   {\bf p} ({\bf s})\,
   {\mathsf{\mu}}_{i}^{2}({\bf s})
 -
 \left( \sum_{\ell \in \sigma (p_{i})}
          p_{i}(\ell)\,
          {\mathsf{E}}_{i}\left( p_{i}^{\ell},
                                 {\bf p}_{-i}
                          \right)
 \right)^{2}\, ;
\end{eqnarray*}
}
so,
the non-linear term in
the mixed strategy $p_{i}$
is a function in the variables
${\mathsf{E}}_{i}\left( p_{i}^{\ell},
                              {\bf p}_{-i}
                 \right)$,
with $\ell \in \sigma (p_{i})$.
Assume that
there are strategies
$r, t \in \sigma (p_{i})$
such that
${\mathsf{E}}_{i}\left( p_{i}^{r},
                        {\bf p}_{-i}
                 \right)
 \neq
 {\mathsf{E}}_{i}\left( p_{i}^{t},
                        {\bf p}_{-i}
                 \right)$.
Since the function
${\mathsf{\widehat{\nu}}}(x) = - x^{2}$
is strictly concave,
we get that
{
\small
\begin{eqnarray*}
&   & {\mathsf{Var}}_{i}\left( p_{i},
                                {\bf p}_{-i}
                        \right)                                          \\
& > & \sum_{{\bf s}
            \in
            {\cal S}}
        {\bf p} ({\bf s})
        {\mathsf{\mu}}_{i}^{2}({\bf s})
        -
        \sum_{\ell \in \sigma (p_{i})}
          p_{i}(\ell)\,
          {\mathsf{E}}_{i}^{2}\left( p_{i}^{\ell},
                                     {\bf p}_{-i}
                              \right)                                   \\
& = & \sum_{\ell \in \sigma (p_{i})}
        p_{i} (\ell)
        \sum_{{\bf s}_{-i} \in {\cal S}_{-i}}
          {\bf p}\left( {\bf s}_{-i}
                 \right)
          \cdot
          {\mathsf{\mu}}_{i}^{2}\left( p_{i}^{\ell},
                                       {\bf s}_{-i}
                                \right)
      -
      \sum_{\ell \in \sigma (p_{i})}
          p_{i}(\ell)\,
          {\mathsf{E}}_{i}^{2}\left( p_{i}^{\ell},
                                     {\bf p}_{-i}
                              \right)                                  \\
& = & \sum_{\ell \in \sigma (p_{i})}
        p_{i}(\ell)\,
        \left( \sum_{{\bf s}_{-i} \in {\cal S}_{-i}}
                 {\bf p}\left( {\bf s}_{-i}
                        \right)
                 \cdot
                 {\mathsf{\mu}}_{i}^{2}\left( p_{i}^{\ell},
                                              {\bf s}_{-i}
                                       \right)
                 -
                 {\mathsf{E}}_{i}^{2}\left( p_{i}^{\ell},
                                            {\bf p}_{-i}
                                     \right)
        \right)                                                        \\
& = & \sum_{\ell \in \sigma (p_{i})}
        p_{i} (\ell)
        {\mathsf{Var}}_{i} \left( p_{i}^{\ell},
                                  {\bf p}_{-i}
                           \right)\, ,
\end{eqnarray*}
}
as needed.
Now,
${\mathsf{SD}}$
is ${\mathsf{E}}$-strictly concave
as the square root
of ${\mathsf{Var}}$.
\end{proof}

\noindent
By Corollary~\ref{hilton frappe}
and Lemma~\ref{paderborn 1},
it follows:

\begin{corollary}
\label{standart deviation}
The valuation
${\mathsf{V}}
 =
 \lambda ({\mathsf{E}}
          +
          {\mathsf{\gamma}}
          \cdot
          {\mathsf{Var}})
 +
 (1-\lambda)
 {\mathsf{V}}^{{\mathsf{\nu}}}$,
with $0 < \lambda \leq 1$,
where
${\mathsf{\nu}}$
is increasing and strictly convex,
and with ${\mathsf{\gamma}} > 0$,
is ${\mathsf{E}}$-strictly concave.
\end{corollary}

\section{Equilibria and Their Properties}
\label{equilibria}

\subsection{${\mathsf{V}}$-Equilibrium}

Fix a player $i \in [n]$.
The pure strategy $p_{i}^{\ell}$
is a {\it ${\mathsf{V}}_{i}$-best pure response}
to a partial mixed profile ${\bf p}_{-i}$
if
{
\small
\begin{eqnarray*}
{\mathsf{V}}_{i}\left( p_{i}^{\ell},
                       {\bf p}_{-i}
                \right)
& = &
  \min
      \left\{ {\mathsf{V}}_{i}\left( p_{i}^{\ell^{\prime}},
                                     {\bf p}_{-i}
                              \right)
                 \mid
                 \ell^{\prime}
                 \in
                 S_{i}
    \right\}\, ;
\end{eqnarray*}
}
so,
the pure strategy $p_{i}^{\ell}$
minimizes the valuation
${\mathsf{V}}_{i}\left( .,
                        {\bf p}_{-i}
                 \right)$
of player $i$
over her pure strategies.
The mixed strategy $p_{i}$
is a {\it ${\mathsf{V}}_{i}$-best response}
to 
${\bf p}_{-i}$
if
{
\small
\begin{eqnarray*}
{\mathsf{V}}_{i}\left( p_{i},
              {\bf p}_{-i}
           \right)
& = &
 \min
        \left\{ {\mathsf{V}}_{i}\left( p_{i}^{\prime},
                                                      {\bf p}_{-i}
                                             \right)
                   \mid
                   p_{i}^{\prime}
                   \in
                   \mathsf{\Delta}(S_{i})
       \right\}\, ;
\end{eqnarray*}
}
so,
the mixed strategy $p_{i}$
minimizes the valuation
${\mathsf V}_{i}\left( .,
                       {\bf p}_{-i}
                \right)$
of player $i$
over her mixed strategies.
The mixed profile ${\bf p}$
is a {\it ${\mathsf{V}}$-equilibrium}
if for each player $i$, 
the mixed strategy $p_{i}$
is a ${\mathsf{V}}_{i}$-best response
to ${\bf p}_{-i}$;
so,
no player could unilaterally deviate
to another mixed strategy
to reduce her cost.
Denote as
{\sf $\exists {\mathsf{V}}$-EQUILIBRIUM}
the algorithmic problem
of deciding, 
given a game ${\mathsf{G}}$,
the existence of a ${\mathsf{V}}$-equilibrium
for ${\mathsf{G}}^{{\mathsf{V}}}$.

\subsection{The Optimal-Value Property}

We show:


\begin{proposition}[The Optimal-Value Property]
\label{constant value}
Fix a game ${\mathsf{G}}$,
a player $i \in [n]$
and a partial mixed profile
${\bf p}_{-i}$.
Assume that
{\sf (A.1)}
the valuation
${\mathsf{V}}_{i} \left( p_{i},
                         {\bf p}_{-i}
                  \right)$
is concave in $p_{i}$,
and
{\sf (A.2)}
$\widehat{p}_{i}$
is a ${\mathsf{V}}_{i}$-best response
to ${\bf p}_{-i}$.
Then,
for any mixed strategy
$q_{i}$
with
$\sigma \left( q_{i}
        \right)
 \subseteq
 \sigma \left( \widehat{p}_{i}
       \right)$,
${\mathsf{V}}_{i}\left( q_{i},
                        {\bf p}_{-i}
                 \right)
 =
 {\mathsf{V}}_{i}\left( \widehat{p}_{i},
                       {\bf p}_{-i}
                \right)$.
\end{proposition}

\begin{proof}
Set
$A
 :=
 {\mathsf{V}}_{i}\left( \widehat{p}_{i},
                                    {\bf p}_{-i}
                            \right)$.
Since $\widehat{p}_{i}$
is a ${\mathsf{V}}_{i}$-best response to ${\bf p}_{i}$,
it follows that
${\mathsf{V}}_{i}\left( q_{i},
                        {\bf p}_{-i}
                 \right)
 \geq
 A$
for each mixed strategy
$q_{i} \in {\mathsf{\Delta}}(S_{i})$
with
$\sigma \left( q_{i}
        \right)
 \subseteq
 \sigma \left( \widehat{p}_{i}
        \right)$.
Assume,
by way of contradiction,
that there is a mixed strategy
$q_{i} \in {\mathsf{\Delta}}(S_{i})$
with
$\sigma \left( q_{i}
        \right)
 \subseteq
 \sigma \left( p_{i}
        \right)$
such that
${\mathsf{V}}_{i}\left( q_{i},
                        {\bf p}_{-i}
                 \right)
 >
 A$.

Denote as
${\mathsf{\Delta}}\left( \sigma \left( \widehat{p}_{i}
                                \right)
                  \right)$
the set of mixed strategies for player $i$
with supports
contained in
$\sigma \left( \widehat{p}_{i}
        \right)$;
so,
${\mathsf{\Delta}}\left( \sigma \left( \widehat{p}_{i}
                                \right)
                  \right)$
is a subspace in
$[0, 1]^{|\sigma \left( \widehat{p}_{i}
                 \right)|}$.

For each $\lambda \in [0, 1]$,
the strategies
$\lambda \widehat{p}_{i}
 +
 (1 - \lambda) q_{i}$
form a line segment in
${\mathsf{\Delta}}(\sigma \left( \widehat{p}_{i}
                          \right))$.
Extend this line segment
to some strategy
$\widehat{q}_{i}
 \in
 {\mathsf{\Delta}}\left( \sigma \left( \widehat{p}_{i}
                                \right)
                  \right)$
with $\widehat{q}_{i}
      \neq
      \widehat{p}_{i}$
so that
$\widehat{p}_{i}$
is an interior point
on the line segment connecting
$q_{i}$ and $\widehat{q}_{i}$.
The extension is possible since
$\widehat{p}_{i}(j)
 >
 0$
for each strategy $j \in \sigma \left( \widehat{p}_{i}
                                \right)$,
so that
$\widehat{p}_{i}$
is an interior point in
${\mathsf{\Delta}}\left( \sigma \left( \widehat{p}_{i}
                                \right)
                  \right)$.
Since
$\widehat{p}_{i}$
is a ${\mathsf{V}}_{i}$-best response
to ${\bf p}_{-i}$,
it follows that
${\mathsf{V}}_{i}\left( \widehat{q}_{i},
                        {\bf p}_{-i}
                 \right)
 \geq
 A$.
So,
there are points
$\widehat{p}_{i}$,
$q_{i}$,
$\widehat{q}_{i}$
so that:
\begin{itemize}

\item
$\widehat{p}_{i}$
is an interior point on
the line segment connecting
$q_{i}$ and
$\widehat{q}_{i}$.

\item
${\mathsf{V}}_{i}\left( \widehat{p}_{i},
                        {\bf p}_{-i}
                 \right)
 =
 A$,
${\mathsf{V}}_{i}\left( q_{i},
                        {\bf p}_{-i}
                 \right)
 >
 A$
and
${\mathsf{V}}_{i}\left( \widehat{q}_{i},
                        {\bf p}_{-i}
                 \right)
 \geq
 A$.

\end{itemize}
A contradiction to the concavity of
${\mathsf{V}}_{i} \left( .,
                         {\bf p}_{-i}
                  \right)$.
\end{proof}

\subsection{The Strong Equilibrium and Weak Equilibrium Properties}

\noindent
The mixed profile ${\bf p}$
has the {\it Strong Equilibrium} property~\cite[Section 2.6]{MM12}
for player $i$ 
in the game ${\mathsf{G}}^{\mathsf{V}}$
if
for each strategy
$\ell \in \sigma (p_{i})$,
{
\small
\begin{eqnarray*}
      {\mathsf{V}}_{i}\left( p_{i}^{\ell},
                             {\bf p}_{-i}
                      \right)
& = &
   \min \left\{ {\mathsf{V}}_{i}\left( p_{i}^{\ell'},
                                          {\bf p}_{-i}
                                   \right)
                   \mid
                   \ell' \in S_{i}
     \right\}\, ;
\end{eqnarray*}
}
so,
each strategy in the support
of player $i$
is a ${\mathsf{V}}_{i}$-best pure response
to the partial mixed profile
${\bf p}_{-i}$.
Clearly,
the {\it Optimal-Value} property for player $i$
implies the {\it Strong Equilibrium} property for player $i$;
Proposition~\ref{constant value}
extends~\cite[Proposition 2.1]{MM12},
establishing the {\it Strong Equilibrium} property
with the same assumptions.

\noindent
The mixed profile ${\bf p}$
has the {\it Weak Equilibrium} property~\cite[Section 2.6]{MM12}
for player $i \in [n]$
in the game ${\mathsf{G}}^{\mathsf{V}}$
if for each pair of strategies
$\ell, \ell^{\prime} \in \sigma (p_{i})$,
${\mathsf{V}}_{i}\left( p_{i}^{\ell},
              {\bf p}_{-i}
           \right)
= {\mathsf{V}}_{i}\left( p_{i}^{\ell^{\prime}},
                             {\bf p}_{-i}
                  \right)$. 
The mixed profile ${\bf p}$
has the {\it Weak Equilibrium} property~\cite[Section 2.6]{MM12}
in 
${\mathsf G}^{\mathsf V}$
if it has the
{\it Weak Equilibrium} property
for each player $i \in [n]$
in ${\mathsf G}^{\mathsf V}$.
(So,
{\it Strong Equilibrium}
implies {\it Weak Equilibrium}.)

\subsection{The Weak-Equilibrium-for-Expectation Property}
\label{wee}

\noindent
We introduce:

\begin{center}
\fbox{
\begin{minipage}{5.9in}
\begin{definition}[The Weak-Equilibrium-for-Expectation Property]
The valuation ${\mathsf{V}}$
has the
\emph{{\em Weak-Equilibrium-for-Expectation}} property
if the following condition holds
for every game ${\mathsf{G}}$:
For each player $i \in [n]$,
if $p_{i}$ is a
${\mathsf{V}}_{i}$-best-response
to ${\bf p}_{-i}$,
then
${\bf p}$
has the {\it Weak Equilibrium} property
for player $i$
in the game ${\mathsf{G}}^{{\mathsf{E}}}$:
for each pair of strategies
$\ell, \ell^{\prime} \in \sigma (p_{i})$,
$      {\mathsf{E}}_{i}\left( p_{i}^{\ell},
                             {\bf p}_{-i}
                      \right)
= {\mathsf{E}}_{i}\left( p_{i}^{\ell^{\prime}},
                             {\bf p}_{-i}
                      \right)$.
\end{definition}
\end{minipage}
}
\end{center}
\noindent
We now prove that
${\mathsf{E}}$-strict concavity
implies the
{\it Weak-Equilibrium-for-Expectation}
property:

\begin{proposition} 
\label{hilton}
Take a player $i \in [n]$
where
${\mathsf{V}}_{i}$
is {\it ${\mathsf{E}}$-strictly concave}.
Then,
${\mathsf{V}}$
has the
{\it Weak-Equilibrium-for-Expectation} property
for player $i$.
\end{proposition}

\begin{proof}
Assume,
by way of contradiction,
that
${\mathsf{V}}$
does {\em not} have the
{\it Weak-Equilibrium-for-Expectation} property
for player $i$.
Then,
there is a game ${\mathsf{G}}$,
a partial mixed profile
${\bf p}_{-i}$
and a mixed strategy $p_{i}$
which is a
${\mathsf{V}}_{i}$-best-response to ${\bf p}_{-i}$
such that
for some strategies
$r, t \in \sigma (p_{i})$,
${\mathsf{E}}_{i} \left( p_{i}^{r},
                         {\bf p}_{-i}
                  \right)
 \neq
 {\mathsf{E}}_{i} \left( p_{i}^{t},
                         {\bf p}_{-i}
                  \right)$.
Since ${\mathsf{V}}_{i}$
is ${\mathsf{E}}$-strictly concave,
this implies that
{
\small
\begin{eqnarray*}
{\mathsf{V}}_{i} \left( \frac{\textstyle 1}
                              {\textstyle 2}
                         p_{i}^{r}
                         +
                         \frac{\textstyle 1}
                              {\textstyle 2}
                         p_{i}^{t},
                         {\bf p}_{-i}
                  \right)
& > &
 \frac{\textstyle 1}
      {\textstyle 2}
 {\mathsf{V}}_{i} \left( p_{i}^{r},
                         {\bf p}_{-i}
                  \right)
 +
 \frac{\textstyle 1}
      {\textstyle 2}
 {\mathsf{V}}_{i} \left( p_{i}^{t},
                         {\bf p}_{-i}
                  \right)\, .
\end{eqnarray*}
}
Since ${\mathsf{V}}_{i}$ is concave,
the {\it Optimal-Value} property (Proposition~\ref{constant value})
implies that
{
\small
\begin{eqnarray*}
{\mathsf{V}}_{i} \left( \frac{\textstyle 1}
                              {\textstyle 2}
                         p_{i}^{r}
                         +
                         \frac{\textstyle 1}
                              {\textstyle 2}
                         p_{i}^{t},
                         {\bf p}_{-i}
                  \right)
& = &
 {\mathsf{V}}_{i} \left( p_{i}^{r},
                         {\bf p}_{-i}
                  \right)\ \
=\ \
 {\mathsf{V}}_{i} \left( p_{i}^{t},
                         {\bf p}_{-i}
                  \right)\, .
\end{eqnarray*}
}
A contradiction.
\end{proof}

\noindent
By Lemma~\ref{paderborn 1} and Corollary~\ref{standart deviation},
Proposition~\ref{hilton}
immediately implies:

\begin{corollary}
\label{brand new}
Fix an
$({\mathsf{E}}
  +
  {\mathsf{R}})$-valuation
${\mathsf{V}}$,
where
{\sf (1)}
${\mathsf{R}}
 =
 {\mathsf{\gamma}}
 \cdot
 {\mathsf{Var}}$,
or
{\sf (2)}
${\mathsf{R}}
 =
 {\mathsf{\gamma}}
 \cdot
 {\mathsf{SD}}$,
or
{\sf (3)}
${\mathsf{V}}
 =
 \lambda ({\mathsf{E}}
          +
          {\mathsf{\gamma}}
          \cdot
          {\mathsf{Var}})
 +
 (1-\lambda)
 {\mathsf{V}}^{{\mathsf{\nu}}}$,
with $0 < \lambda \leq 1$,
where
${\mathsf{\nu}}$
is increasing and strictly convex,
and with
${\mathsf{\gamma}} > 0$.
Then,
${\mathsf{V}}$
has the {\it Weak-Equilibrium-for-Expectation} property.
\end{corollary}

\noindent
We now turn to a particular
concave valuation
and exploit the {\it Optimal-Value} property
(Proposition~\ref{constant value})
to prove:

\begin{proposition}
\label{concavity implies weak equilibrium}
Fix a player $i \in [n]$,
and consider the concave valuation
{
\small
\begin{eqnarray*}
{\mathsf{V}}_{i}
& = & 
  \alpha_{0}
      \cdot
      {\mathsf{E}}_{i}
      +
      \sum_{2 \leq k \leq \ell \mid \mbox{$k$ is even}}
        \alpha_{k}
        \cdot
        {\mathsf{kM}}_{i}\, ,
\end{eqnarray*}
}
for some constants
$\alpha_{k} \geq 0$,
$0 \leq k \leq \ell$.
Then,
${\mathsf{V}}_{i}$
has the
{\it Weak-Equilibrium-for-Expectation} property
for player $i$.
\end{proposition}

\begin{proof}
Fix a game ${\mathsf{G}}$.
Clearly,
for a mixed profile ${\bf p}$,
{
\small
\begin{eqnarray*}
&   & {\mathsf{2kM}}_{i}({\bf p})                                            \\
& = & \sum_{{\bf s}
            \in
            S}
        \left( {\mathsf{\mu}}_{i}({\bf s})
               -
               {\mathsf{E}}_{i}({\bf p})
        \right)^{2k}
        p({\bf s})                                                          \\
& = & \sum_{{\bf s}
            \in
            S}
        \left( \sum_{t=0}^{2k}
                 (-1)^{t}\,
                 {2k \choose t}\,
                 \left( {\mathsf{\mu}}_{i}({\bf s})
                 \right)^{t}\,
                 \left( {\mathsf{E}}_{i}({\bf p})
                 \right)^{2k-t}
        \right)
        p({\bf s})                                                         \\
& = & \sum_{t=0}^{2k}
        (-1)^{t}\,
        {2k \choose t}\,
        \left( \sum_{{\bf s}
                     \in
                     S}
                 \left( {\mathsf{\mu}}_{i}({\bf s})
                 \right)^{t}
                 p({\bf s})
        \right)
        \left( {\mathsf{E}}_{i}({\bf p})
        \right)^{2k-t}                                                     \\
& = & \underbrace{\sum_{t=2}^{2k}
        (-1)^{t}\,
        {2k \choose t}\,
        \left( \sum_{{\bf s}
                     \in
                     S}
                 \left( {\mathsf{\mu}}_{i}({\bf s})
                 \right)^{t}
                 p({\bf s})
        \right)
        \left( {\mathsf{E}}_{i}({\bf p})
        \right)^{2k-t}}_{\mbox{polynomial of degree bounded by $2k-1$ in ${\bf p}$}}
      -
      \underbrace{(2k-1)\,
                  \left( {\mathsf{E}}_{i}({\bf p})
                  \right)^{2k}}_{\mbox{polynomial of degree $2k$}}\, .
\end{eqnarray*}
}
Since
\begin{eqnarray*}
{\mathsf{V}}_{i}({\bf p})
& = &
 \alpha_{0}
 \cdot
 {\mathsf{E}}_{i}({\bf p})
 +
 \sum_{2 \leq k \leq \ell}
   \alpha_{2l}
   \cdot
   {\mathsf{2lM}}\, ,
\end{eqnarray*}   
so that
${\mathsf{V}}_{i}({\bf p})$
is the sum of
{\it (i)}
the highest-degree term
$-
(2\ell -1)\,
 \left( {\mathsf{E}}_{i}({\bf p})
 \right)^{2 \ell}$,
which is a polynomial of degree $2 \ell$ in ${\bf p}$,
and
{\it (ii)}
a polynomial of degree
bounded by $2 \ell -1$ in ${\bf p}$.
Since ${\mathsf{V}}_{i}({\bf p})$
is a concave polynomial in ${\bf p}$,
the {\it Optimal-Value} property
(Proposition~\ref{constant value})
implies that
${\mathsf{V}}_{i}({\bf p})$
is a {\em constant} polynomial in ${\bf p}$.
Hence,
it follows that
$-
 (2\, \ell -1)
 \cdot
 \left( {\mathsf{E}}_{i}({\bf p})
 \right)^{2\, \ell}$
is a constant polynomial in ${\bf p}$;
thus,
so is
${\mathsf{E}}_{i}({\bf p})$.
The {\it Weak-Equilibrium-for-Expectation} property
for player $i$ follows.
\end{proof}

\noindent
Brautbar {\em et al.}~\cite[Section 3.1]{BKS10}
study a class of valuations,
coming from the 
{\it Mean-Variance} paradigm
of Markowitz~\cite{M52}
and termed as
{\it Mean-Variance Preference Functions};
we rephrase their definition~\cite[Definition 1]{BKS10}
to fit into the adopted setting of minimization games:
\begin{definition}[\cite{BKS10}]
\label{kearns}
Fix a player $i \in [n]$.
A {\it Mean-Variance Preference Function}
is a valuation
${\mathsf{V}}_{i}\left( p_{i},
                        {\bf p}_{-i}
                 \right)
 :=
 {\mathsf{G}}_{i}\left( {\mathsf{E}}_{i}\left( p_{i},
                                               {\bf p}_{-i}
                                        \right),
                        {\mathsf{Var}}_{i}\left( p_{i},
                                                 {\bf p}_{-i}
                                          \right)
                 \right)$                                                                
which satisfies:
\begin{enumerate}

\item[$({\mathsf{1}})$]
${\mathsf{V}}_{i}\left( p_{i},
                        {\bf p}_{-i}
                 \right)$
is concave in $p_{i}$.

\item[$({\mathsf{2}})$]
${\mathsf{G}}_{i}$
is non-decreasing
in its first argument
(${\mathsf{E}}_{i}\left( p_{i},
                                       {\bf p}_{-i}
                              \right)$).

\item[$({\mathsf{3}})$]
Fix a partial mixed profile
${\bf p}_{-i}$
and a nonempty convex subset
${\mathsf{\Delta}}
 \subseteq
 {\mathsf{\Delta}}_{i}
 =
 {\mathsf{\Delta}}(S_{i})$
such that
${\mathsf{V}}_{i}\left( p_{i},
                        {\bf p}_{-i}
                 \right)$ 
is constant
on ${\mathsf{\Delta}}$.
Then,
both
${\mathsf{E}}_{i}\left( p_{i},
                        {\bf p}_{-i}
                 \right)$
and
${\mathsf{Var}}_{i}\left( p_{i},
                        {\bf p}_{-i}
                 \right)$
are constant
on ${\mathsf{\Delta}}$.

\end{enumerate}
\end{definition}

\noindent
So,
a Mean-Variance Preference Function
simultaneously generalizes
and restricts the
$({\mathsf{E}}
 +
 {\mathsf{R}})$-valuations;
it generalizes sum
to ${\mathsf{G}}$
but restricts ${\mathsf{R}}$
to ${\mathsf{Var}}$.
Note that Condition ${\mathsf{(3)}}$
in Definition~\ref{kearns}
may be seen as a generalization
of the {\it Weak-Equilibrium-for-Expectation} property
conditioned on the assumption that
${\mathsf{V}}_{i}\left( p_{i},
                        {\bf p}_{-i}
                 \right)$ 
is constant
on a nonempty convex set
${\mathsf{\Delta}}
 \subseteq 
 {\mathsf{\Delta}}_{i}$.   
It is proved in~\cite[Claim 1]{BKS10}
that ${\mathsf{E}}
 +
 {\mathsf{Var}}$
is a Mean-Variance Preference Function.
Since
${\mathsf{E}}
 +
 {\mathsf{Var}}$
is ${\mathsf{E}}$-strictly concave
(Lemma~\ref{paderborn 1})
and has the {\it Optimal-Value} property
(Proposition~\ref{constant value}),
their established Condition ${\mathsf{(3)}}$
is a special case
of our general result
that every ${\mathsf{E}}$-strictly concave valuation
has the
{\it Weak-Equilibrium-for-Expectation} property
(Proposition~\ref{hilton}).

\section{Two Strategies}
\label{two ordered links}

\subsection{Player-Specific Scheduling Games}
\label{player specific scheduling games}

\noindent
A {\it player-specific scheduling game}~\cite[Section 3]{MM12}
is equipped with
an integer {\it weight}
$\omega (i, i', \ell)$
for each triple of
a player $i \in [n]$,
a player $i' \in [n]$
and a strategy
$\ell \in S_{i}$,
also called {\it link,}
with $S_{1} =
   \ldots =
   S_{n}
   =
   [m]$;
$\omega (i, i', \ell)$
represents the load
due to player $i'$
incurred to player $i$
on link $\ell$.
Given the collection of weights
$\left\{ \omega (i, i', \ell)
\right\}_{i, i' \in [n],
     \ell \in [m]}$,
the cost function $\mathsf{\mu}_{i}$
is defined
by
${\mathsf{\mu_{i}}}\left( {\bf s}
                        \right)
 =  \sum_{i' \in [n]\
            \mid\
            s_{i'} = s_{i}
           }
    \omega \left( i, i', s_{i}
           \right)$.
In a
{\it player-specific scheduling game on
     two ordered links $1$ and $2$}~\cite[Section 5.2.2]{MM12},
for each pair of players $i, i' \in [n]$,
either
$\omega (i, i', 1)
 =
 \omega (i, i', 2)
 =
 0$
or
$\omega (i, i', 1)
 <
 \omega (i, i', 2)$.

We derive a combinatorial formula
for the {\it $k$-moment} valuation
of the cost of a player choosing
a link $\ell$
in a player-specific scheduling game.
The formula takes the form
of a {\it partition polynomial}~\cite{B27}:
a multivariable polynomial
defined by a sum
over partitions of the integer $k$.
The formula
uses the function
${\mathsf{f}}:
[0, 1]
\times
{\mathbb{N}}_{0}
\rightarrow
{\mathbb{R}}$
with
{
\small
\begin{eqnarray*}
      {\mathsf{f}}(x, j)
& := &
(-x)^{j}(1-x) + (1-x)^{j} x\, .
\end{eqnarray*}
}
\noindent
Note that
{
\small
\begin{eqnarray*}
      {\mathsf{f}}(x, j)
& = & \left\{ \begin{array}{ll}
                 x(1-x)\left( x^{j-1} + (1-x)^{j-1}
                       \right)\, ,
                 &
                 \mbox{for an even integer $j \geq 2$}           \\
                 x(1-x)\left( - x^{j-1} + (1-x)^{j-1}
                       \right)\, ,
                 &
                 \mbox{for an odd integer $j \geq 3$}
             \end{array}
       \right.\, .
\end{eqnarray*}
}
The following simple claim follows
by inspection.

\begin{lemma}
\label{by inspection}
The function
${\mathsf{f}}$
has the following properties:
\begin{enumerate}

\item
${\mathsf{f}}(x, 0) = 1$ for all $x \in [0, 1]$. \\
${\mathsf{f}}(x, 1) = 0$ for all $x \in [0, 1]$. \\
${\mathsf{f}}(0, j) = {\mathsf{f}}(1, j) = 0$
for all integers $j \geq 0$.

\item
For an even integer $j \geq 2$:                 \\
${\mathsf{f}}(x, j) > 0$
for all $0 < x < 1$.                            \\
${\mathsf{f}}(x, j)
 =
 {\mathsf{f}}(1-x, j)$
for all $x \in [0, 1]$.

\item
For an odd integer $j \geq 3$: \\
${\mathsf{f}}\left( \frac{\textstyle 1}
                         {\textstyle 2},
                    j
             \right)
 =
 0$.                                                    \\
${\mathsf{f}}(x, j)
 >
 0$
for all
$0 < x <
 \frac{\textstyle 1}
      {\textstyle 2}$.                                  \\
${\mathsf{f}}(x, j)
 <
 0$
for all
$\frac{\textstyle 1}
      {\textstyle 2}
 < x < 1$.                                              \\
${\mathsf{f}}(x, j)
 =
 -
 {\mathsf{f}}(1-x, j)$
for all $x \in [0, 1]$.

\end{enumerate}
\end{lemma}

\noindent
We show:

\begin{proposition}
\label{formula for k moment}
Consider a player-specific
scheduling game
with $n$ players.
Then,
for each player $i \in [n]$,
for a link $\ell \in [m]$
and a mixed profile ${\bf p}$,
{
\small
\begin{eqnarray*}
      {\mathsf{kM}}_{i}\left( p_{i}^{\ell},
                              {\bf p}_{-i}
                       \right)
& = & \sum_{\begin{array}{c}
              \langle r_{1}, \ldots, r_{n}
              \rangle
              \in
              {\mathbb{N}}_{0}^{n}
              \mid
              \\
              \sum_{j \in [n]}
                r_{j}
              =
              k,
              r_{i} = 0
              \\
              \&\,
              \forall j \in [n]: r_{j} \neq 1
            \end{array}
           }
       \frac{k!}
            {r_{1}! \ldots r_{n}!}\,
      \prod_{j \in [n] \setminus \{ i \}}
                {\mathsf{f}} \left( p_{j}(\ell),
                                    r_{j}
                             \right)
                \left( \omega (i, j, \ell)
                \right)^{r_{j}}\, .
\end{eqnarray*}
}
\end{proposition}

\begin{proof}
We shall first consider
the special cases
$k=0$ and $k=1$.

\noindent
\underline{The case $k = 0$:}\\
Clearly,
   $   {\mathsf{0M}}_{i}\left( p_{i}^{\ell}, {\bf p}_{-i}
                       \right)
=
\sum_{{\bf s} \in S}
        \left( {\mathsf{\mu}}_{i}({\bf s}) - {\mathsf{E}}_{i}({\bf p})
        \right)^{0}\,
        p({\bf s})                   
=
 \sum_{{\bf s} \in S}
        p({\bf s})                      
= 1$.
The formula has value 
{
\small
\begin{eqnarray*}
      \sum_{\begin{array}{c}
              r_{1}, \ldots, r_{n}
              \mid
              \\
              \sum_{j \in [n]}
                r_{j}
                =
                0,
                r_{i} = 0
              \\
              \&\,
              \forall j \in [n]: r_{j} \neq 1
            \end{array}
           }
        \frac{\textstyle 0!}
             {\textstyle r_{1}! \ldots r_{n}!}\,
        \prod_{j \in [n] \setminus \{ i \}}
           {\mathsf{f}} \left( p_{j}(\ell),
                               r_{j}
                        \right)\,
           \left( \omega (i, j, \ell)
           \right)^{r_{j}}
& = & 1\, .
\end{eqnarray*}
}

\noindent
\underline{The case $k = 1$:}
Clearly,
{
\small
\begin{eqnarray*}
      {\mathsf{1M}}_{i}\left( p_{i}^{\ell}, {\bf p}_{-i}
                       \right)
& = & \sum_{{\bf s} \in S}
        \left( {\mathsf{\mu}}_{i}({\bf s}) - {\mathsf{E}}_{i}({\bf p})
        \right)^{1}\,
        p({\bf s})\
       \
      =\
      \                     
      0\, .
\end{eqnarray*}
}
The value of the formula is
{
\small
\begin{eqnarray*}
      \sum_{\begin{array}{c}
              r_{1}, \ldots, r_{n}
              \mid
              \\
              \sum_{j \in [n]}
                r_{j}
                =
              1,
            r_{i} = 0
            \\
            \&\,
            \forall j \in [n]: r_{j} \neq 1
            \end{array}
           }
        \frac{\textstyle 1!}
             {\textstyle r_{1}! \ldots r_{n}!}\,
        \prod_{j \in [n] \setminus \{ i \}}
           {\mathsf{f}} \left( p_{j}(\ell),
                               r_{j}
                        \right)\,
           \left( \omega (i, j, \ell)
           \right)^{r_{j}}
& = & 0\, ,
\end{eqnarray*}
}
since there is no term to add.

Assume now that $k \geq 2$.
The proof is
by induction on
the number of players $n$.
For the basis case
where $n = 1$,
${\mathsf{kM}}_{1}(p_{1}^{\ell}, .)
=
0$,
and $0$ is also the value
given by the formula
(since $r_{1} = 0$
implies
$\sum_{j} r_{j}
\neq
k$,
and there is no term to add).

Assume inductively that
the formula holds
for $n-1$ players.
For the induction step,
we shall establish
the formula for
$n$ players.
Without loss of generality,
fix $i := 1$.
For simplicity,
write $p_{j}$
and $\omega_{j}$
for $p_{j}(\ell)$
and $\omega (i, j, \ell)$,
respectively.
For any integer
$t \leq k$,
denote as
${\bf p} \mid t$
the restriction of ${\bf p}$
to the players
$1, \ldots, t$;
so,
${\bf p} \mid n
=
{\bf p}$.
Set
${\mathsf{kM}}_{1}(t)
:=
{\mathsf{kM}}_{1}\left( p_{1}^{\ell},
            ({\bf p} \mid t)_{-1}
         \right)$;
so,
${\mathsf{kM}}_{1}(n)
 =
 {\mathsf{kM}}_{1}\left( p_{1}^{\ell},
                        {\bf p}_{-1}
                  \right)$.
Clearly,
{
\small
\begin{eqnarray*}
&    & {\mathsf{kM}}_{1}(n)                                                                 \\
& =  & \sum_{{\bf s} \in S}
         p ({\bf s})\,
         \left( {\mathsf{\mu}}_{1}({\bf s})
                -
                {\mathbb{E}}_{1}({\bf p})
         \right)^{k}                                                                        \\
& =  & \sum_{{\bf s}_{-n} \in S_{-n}}
         p ({\bf s}_{-n})
         \cdot
         \left( \left( {\mathsf{\mu}}_{1}({\bf s}_{-n})
                       + w_{n}
                       -
                       {\mathbb{E}}_{1}({\bf p})
                 \right)^{k}\,
                 p_{n}
                 +
                 \left( {\mathsf{\mu}}_{1}({\bf s}_{-n})
                        -
                        {\mathbb{E}}_{1}({\bf p})
                 \right)^{k}\,
                 (1 - p_{n})
         \right)                                                                           \\
& = & \sum_{{\bf s}_{-n} \in S_{-n}}
        p ({\bf s}_{-n})
        \cdot                                                                              \\
&   & \left( \left( {\mathsf{\mu}}_{1}({\bf s}_{-n})
                    + w_{n}
                    -
                    {\mathbb{E}}_{1}({\bf p} \mid (n-1))
                    -
                    p_{n} w_{n}
             \right)^{k}\,
             p_{n}
             +
             \left( {\mathsf{\mu}}_{1}({\bf s}_{-n})
                    -
                    {\mathbb{E}}_{1}({\bf p} \mid (n-1))
                    -
                    p_{n} w_{n}
             \right)^{k}\,
             (1 - p_{n})
      \right)                                                                              \\
& = & \sum_{{\bf s}_{-n} \in S_{-n}}
        p ({\bf s}_{-n})
        \cdot                                                                              \\
&   &  \left( \left( {\mathsf{\mu}}_{1}({\bf s}_{-n})
                     -
                     {\mathbb{E}}_{1}({\bf p} \mid (n-1))
                     +
                     w_{n}
                     (1 - p_{n})
             \right)^{k}\,
             p_{n}
             +
             \left( {\mathsf{\mu}}_{1}({\bf s}_{-n})
                    -
                    {\mathbb{E}}_{1}({\bf p} \mid (n-1))
                    -
                    p_{n} w_{n}
             \right)^{k}\,
             (1 - p_{n})
       \right)\, .
\end{eqnarray*}
}
Set
${\mathsf{B}}
 =
 {\mathsf{B}}({\bf s}_{-n})
 :=
 {\mathsf{\mu}}_{1}({\bf s}_{-n})
 -
 {\mathbb{E}}_{1}({\bf p} \mid (n-1))$.
Then,
{
\small
\begin{eqnarray*}
&    & {\mathsf{kM}}_{1}(n)                                                                            \\
& = & \sum_{{\bf s}_{-n} \in S_{-n}}
        p ({\bf s}_{-n})
        \left( \left( {\mathsf{B}}
                      +
                      (1 - p_{n})\,
                      w_{n}
              \right)^{k}\,
              p_{n}
              +
              \left( {\mathsf{B}}
                     -
                     p_{n} w_{n}
              \right)^{k}\,
              (1 - p_{n})
       \right)                                                                                       \\
& = & \sum_{{\bf s}_{-n} \in S_{-n}}
        p ({\bf s}_{-n})
        \left( \left( \sum_{t=0}^{k}
                        {k \choose t}
                        {\mathsf{B}}^{t}
                        (1 - p_{n})^{k-t}
                        \left( w_{n}
                        \right)^{k-t}
               \right)
               p_{n}
               +
              \left( \sum_{t=0}^{k}
                       {k \choose t}
                       {\mathsf{B}}^{t}
                       (- p_{n})^{k-t}
                       \left( w_{n}
                       \right)^{k-t}
              \right)
              (1 - p_{n})
       \right)                                                                                      \\
& = & \sum_{{\bf s}_{-n} \in S_{-n}}
        p ({\bf s}_{-n})
        \sum_{t=0}^{k}
          {k \choose t}
          {\mathsf{B}}^{t}
          \left( p_{n} (1 - p_{n})^{k-t}
                 +
                 (1 - p_{n}) (- p_{n})^{k-t}
          \right)
          \left( w_{n}
          \right)^{k-t}                                                                                \\
& = & \sum_{{\bf s}_{-n} \in S_{-n}}
        p ({\bf s}_{-n})
        \sum_{t=0}^{k}
          {k \choose t}
          {\mathsf{B}}^{t}
          {\mathsf{f}}\left( p_{n},
                             k-t
                      \right)
          \left( w_{n}
          \right)^{k-t}                                                                               \\
& = & \sum_{t=0}^{k}
        {k \choose t}\,
        \left( \sum_{{\bf s}_{-n} \in S_{-n}}
                 p ({\bf s}_{-n})
                 \cdot
                 \left( {\mathsf{B}}({\bf s}_{-n})
                 \right)^{t}
        \right)\,
        {\mathsf{f}}\left( p_{n},
                           k-t
                    \right)\,
        \left( w_{n}
        \right)^{k-t}                                                                                 \\
& = & \sum_{t=0}^{k}
        {k \choose t}\,
        {\mathsf{tM}}_{1}(n-1)\,
        {\mathsf{f}}\left( p_{n},
                           k-t
                    \right)\,
        \left( w_{n}
        \right)^{k-t}\, .
\end{eqnarray*}
}
By induction hypothesis,
it follows that
{
\small
\begin{eqnarray*}
&    & {\mathsf{kM}}_{1}(n)                                                                            \\
& = & \sum_{t=0}^{k}
        {k \choose t}\,
        \left( \sum_{\begin{array}{c}
                       r_{1}, \ldots, r_{n-1}
                       \mid
                       \\
                       \sum_{j \in [n-1]}
                         r_{j}
                       =
                       t,
                       r_{1} = 0
                       \\
                       \&\,
                       \forall j \in [n-1]: r_{j} \neq 1
                      \end{array}
                    }
                 \frac{t!}
                      {r_{1}! \ldots r_{n-1}!}\,
                 \prod_{j \in [n-1] \setminus \{ 1 \}}
                    {\mathsf{f}} \left( p_{j},
                                        r_{j}
                                 \right)
                 \left( \omega_{j}
                 \right)^{r_{j}}
        \right)
        {\mathsf{f}}\left( p_{n},
                           k-t
                    \right)\,
        \left( w_{n}
        \right)^{k-t}                                                                                   \\
& = & \sum_{\begin{array}{c}
              r_{1}, \ldots, r_{n-1}
              \mid
              \\
              \sum_{j \in [n-1]}
                r_{j}
              =
              t,
              r_{1} = 0
              \\
              \&\,
              \forall j \in [n-1]: r_{j} \neq 1 
            \end{array}
           }
        \sum_{t=0}^{k}
          {k \choose t}\,
          \frac{t!}
               {r_{1}! \ldots r_{n-1}!}\,
          \cdot
          \prod_{j \in [n-1] \setminus 1}
             {\mathsf{f}} \left( p_{j},
                                 r_{j}
                          \right)\,
             {\mathsf{f}}\left( p_{n},
                                k-t
                         \right)
          \cdot
          \left( \omega_{j}
          \right)^{r_{j}}\,
          \left( w_{n}
          \right)^{k-t}                                                                               \\
& = & \sum_{\begin{array}{c}
              r_{1}, \ldots, r_{n-1}
              \mid
              \\
              \sum_{j \in [n-1]}
                r_{j}
              =
              t,
              r_{1} = 0
              \\
              \&\,
              \forall j \in [n-1]: r_{j} \neq 1
            \end{array}
           }
        \sum_{t=0}^{k}
          \frac{k!}
               {r_{1}! \ldots r_{n-1}! (k-t)!}\,
          \cdot
          \prod_{j \in [n-1] \setminus \{ 1 \}}
             {\mathsf{f}} \left( p_{j},
                                 r_{j}
                          \right)\,
             {\mathsf{f}}\left( p_{n},
                                k-t
                         \right)
          \cdot
          \left( \omega_{j}
          \right)^{r_{j}}\,
          \left( w_{n}
          \right)^{k-t}                                                                            \\
& = & \sum_{\begin{array}{c}
              r_{1}, \ldots, r_{n-1}, r_{n}
              \mid
              \\
              \sum_{j \in [n]}
                r_{j}
              =
              k,
              r_{1} = 0
              \\
              \&\,
              \forall j \in [n-1]: r_{j} \neq 1
            \end{array}
           }
          \frac{k!}
               {r_{1}! \ldots r_{n-1}! r_{n}!}\,
          \cdot
          \prod_{j \in [n] \setminus \{ 1 \}}
             {\mathsf{f}} \left( p_{j},
                                 r_{j}
                          \right)\,
          \cdot
          \left( \omega_{j}
          \right)^{r_{j}}\, .
\end{eqnarray*}
}
Since
${\mathsf{f}}(p_{n}, r_{n})
 =
 0$
when $r_{n} = 1$,
it follows that
\begin{eqnarray*}
      {\mathsf{kM}}_{1}(n)
& = & \sum_{\begin{array}{c}
              r_{1}, \ldots, r_{n-1}, r_{n}
              \mid
              \\
              \sum_{j \in [n]}
                r_{j}
              =
              k,
              r_{1} = 0
              \\
              \&\,
              \forall j \in [n]: r_{j} \neq 1
            \end{array}
           }
          \frac{k!}
               {r_{1}! \ldots r_{n}!}\,
          \cdot
          \prod_{j \in [n] \setminus \{ 1 \}}
             {\mathsf{f}} \left( p_{j},
                                 r_{j}
                          \right)\,
          \cdot
          \left( \omega_{j}
          \right)^{r_{j}}\, ,
\end{eqnarray*}
as needed.
By Lemma~\ref{by inspection},
{
\small
\begin{eqnarray*}
      {\mathsf{f}} \left( p_{j}(1),
                     r_{j}
              \right)
& = & \left\{ \begin{array}{ll}
                {\mathsf{f}} \left( 1 - p_{j}(1),
                                    r_{j}
                             \right)\, ,                    &
                \mbox{for even $r_{j}$}                                 \\
                -
                {\mathsf{f}} \left( 1 - p_{j}(1),
                                    r_{j}
                             \right)                        &
                \mbox{if $r_{j}$ is odd}
              \end{array}
      \right.          
\end{eqnarray*}
}
Since
$\sum_{j \in [n] \setminus \{ i \}}
   r_{j}
 =
 k$
and $k$ is even,
the number of odd $r_{j}$'s is even,
and this implies that
$\alpha_{i}^{1}
 =
 \alpha_{i}^{2}$.
\end{proof}

\subsection{The Mixed-Player-Has-Pure-Neighbors Property}
\label{the mixed player has pure neighbors property}

We show:

\begin{proposition}[The Mixed-Player-Has-Pure-Neighbors Property]
\label{the real almost pure property}
Fix a player-specific scheduling game
on two ordered links $1$ and $2$.
Fix an
$({\mathsf{E}}
  +
  {\mathsf{R}})$-valuation
${\mathsf{V}}$,
where
${\mathsf{R}}$ is either
{\sf (1)}
${\mathsf{Var}}$,
or
{\sf (2)}
${\mathsf{SD}}$,
or 
{\sf (3)}
a concave linear sum
$\sum_{k \in \{ 2, 4, 6, 8 \}}
   \alpha_{k}
   \cdot
   {\mathsf{kM}}$,
with $\alpha_{k}
      \geq
      0$
for $k \in \{ 2, 4, 6, 8 \}$.
Fix a ${\mathsf{V}}$-equilibrium ${\bf p}$
and a player $i \in [n]$.
Then,
either
{\sf (C.1)}
player $i$ is pure,
or
{\sf (C.2)}
all players $i' \in [n] \setminus \{ i \}$
with
$\omega (i, i', 1)
 \neq
 0$
are pure.
\end{proposition}

\begin{proof}
We first prove a key property
of moment valuations
${\mathsf{kM}}$,
where
$k \in \{ 2, 4, 6, 8 \}$.

\begin{lemma}
\label{even moments}
Fix a player-specific scheduling game
on two ordered links $1$ and $2$.
Fix the valuation ${\mathsf{kM}}$,
for an even integer
$k \in \{ 2, 4, 6, 8 \}$.
Fix a  mixed profile ${\bf p}$
and a player $i \in [n]$.
Then,
one of the conditions
{\sf (C.1)},
{\sf (C.2)}
and
{\sf (C.3)}
holds:
{\sf (C.1)}
Player $i$ is pure.
{\sf (C.2)}
All players
$i' \in [n] \setminus \{ i \}$
with
$\omega \left( i, i', 1 \right)
 \neq
 0$
are pure.
{\sf (C.3)}
Player $i$ is mixed
with
${\mathsf{kM}}_{i}\left( p_{i}^{1}, {\bf p}_{-i}
                 \right)
 <
 {\mathsf{kM}}_{i}\left( p_{i}^{2}, {\bf p}_{-i}
                 \right)$.
\end{lemma}

\noindent
For the proof of Lemma~\ref{even moments},
we shall employ a combinatorial
{\it Embracing Lemma}
(Lemma~\ref{embracing helps})
to establish that
the $k$-moment valuation
increases strictly monotone
for $k \in \{ 6, 8 \}$;
it seems that a different technique
is needed to extend
Lemma~\ref{even moments}
beyond $k=8$.

\begin{lemma}[Embracing Lemma]
\label{embracing helps}
Fix a pair of odd integers
$r \geq 3$ and $s \geq 3$.
Fix probabilities
$p$ and $q$
with
$0 < p < \frac{\textstyle 1}
              {\textstyle 2}$
and
$\frac{\textstyle 1}
      {\textstyle 2}
 <
 q
 <
 1$.
Fix a weight $w$
and a pair of numbers
$\alpha, \beta \in {\mathbb{R}}^{+}$
with
$\alpha
 \cdot
 \beta
 \geq
 \frac{\textstyle 1}
      {\textstyle 2}$.
Then,
the function
{
\small
\begin{eqnarray*}
       {\mathsf{F}}(y)
& := & \alpha
       \cdot
       \frac{\textstyle 1}
            {\textstyle (r-1)!}
       \cdot
       \frac{\textstyle 1}
            {\textstyle (s+1)!}\,
       {\mathsf{f}}(p, r-1)
       \cdot
       {\mathsf{f}}(q, s+1)\,
       w^{s+1}
       y^{r-1}                                                            \\
&    & +
       \frac{\textstyle 1}
            {\textstyle r!}
       \cdot
       \frac{\textstyle 1}
            {\textstyle s!}\,
       {\mathsf{f}}(p, r)
       \cdot
       {\mathsf{f}}(q, s)\,
       w^{s}
       y^{r}                                                             \\
&    & +
       \beta
       \cdot
       \frac{\textstyle 1}
            {\textstyle (r+1)!}
       \cdot
       \frac{\textstyle 1}
            {\textstyle (s-1)!}\,
       {\mathsf{f}}(p, r+1)
       \cdot
       {\mathsf{f}}(q, s-1)\,
       w^{s-1}
       y^{r+1}
\end{eqnarray*}
}
increases strictly monotone in $y$.
\end{lemma}

\noindent
The {\it Embracing Lemma}
establishes that a triple of ``adjacent''
partitions
in the formula
from Proposition~\ref{formula for k moment}
increases strictly monotone;
intuitively,
the terms corresponding to two of the partitions
in the triple
are positive
and ``help out''
by embracing
the third negative term
to counterbalance
its negative effect to increasing monotonicity.

\begin{proof}
The proof of Lemma~\ref{embracing helps}
uses an elementary observation:

\begin{observation}
\label{geometric mean}
Consider an odd integer $r \geq 3$.
Then,
for all $x \in [0, 1]$,
\begin{eqnarray*}
      {\mathsf{f}}(x, r-1)
      \cdot
      {\mathsf{f}}(x, r+1)
& > & \left( {\mathsf{f}}(x, r)
      \right)^{2}\, .
\end{eqnarray*}
\end{observation}

\begin{proof}
The claim amounts to
\begin{eqnarray*}
      x^{2}\,
      (1-x)^{2}
      \left( x^{r-2} + (1-x)^{r-2}
      \right)
      \cdot
      \left( x^{r} + (1-x)^{r}
      \right)
& > & x^{2}
      (1-x)^{2}
      \left( (1-x)^{r} - x^{r}
      \right)\, ,
\end{eqnarray*}
which holds trivially.
\end{proof}

\noindent
We continue with the proof
of the
{\it Embracing Lemma}.
Write
\begin{eqnarray*}
      {\mathsf{F}}(y)
& := & a
      \cdot
      y^{r-1}
      -
      b
      \cdot
      y^{r}
      +
      c
      \cdot
      y^{r+1}\, ,
\end{eqnarray*}
with
\begin{eqnarray*}
       a
& := &  \alpha
       \cdot
       \frac{\textstyle 1}
            {\textstyle (r-1)!}
       \cdot
       \frac{\textstyle 1}
            {\textstyle (s+1)!}\,
       {\mathsf{f}}(p, r-1)
       \cdot
       {\mathsf{f}}(q, s+1)\,
       w^{s+1}\, ,                                              \\
       b
& :=  & -
       \frac{\textstyle 1}
            {\textstyle r!}
       \cdot
       \frac{\textstyle 1}
            {\textstyle s!}\,
       {\mathsf{f}}(p, r)
       \cdot
       {\mathsf{f}}(q, s)\,
       w^{s}\, ,                                               \\
       c
& :=  & \beta
       \cdot
       \frac{\textstyle 1}
            {\textstyle (r+1)!}
       \cdot
       \frac{\textstyle 1}
            {\textstyle (s-1)!}\,
       {\mathsf{f}}(p, r+1)
       \cdot
       {\mathsf{f}}(q, s-1)\,
       w^{s-1}\, .
\end{eqnarray*}
Note that
$a, b, c > 0$.
Clearly,
\begin{eqnarray*}
      {\mathsf{F}}^{\prime}(y)
& = & (r-1)\, a\, y^{r-2}
      -
      r\, b\, y^{r-1}
      +
      (r+1)\, c\, y^{r}                             \\
& = & y^{r-2}
      \cdot
      \underbrace{\left( (r-1)\, a
                         - r\,b\,y
                         + (r+1)\,c\,y^{2}
                  \right)}_{:= {\mathsf{h}}(y)}\, .
\end{eqnarray*}
We shall prove that
${\mathsf{h}}^{\prime}(y) > 0$
for all
$y \geq 0$.
Clearly,
\begin{eqnarray*}
      {\mathsf{h}}^{\prime}(y)
& = & -r\, b
      +
      2\,(r+1)\,c\,y\, ,
\end{eqnarray*}
and
\begin{eqnarray*}
      {\mathsf{h}}^{\prime\prime}(y)
& = & 2\, (r+1) y\ \
      >\ \
      0\, .
\end{eqnarray*}
Thus,
${\mathsf{h}}(y)$
takes its minimum value for
$y_{0}
 =
 \frac{\textstyle rb}
      {\textstyle 2(r+1)c}$.
Hence,
it suffices to prove that
${\mathsf{h}}(y_{0})
 >
 0$.
Clearly,
\begin{eqnarray*}
      {\mathsf{h}}(y_{0})
& = & (r-1)\, a
      -
      r\, b\,
      \frac{\textstyle rb}
           {\textstyle 2(r+1)c}
      +
      (r+1)\, c\,
      \left( \frac{\textstyle rb}
                  {\textstyle 2(r+1)c}
      \right)^{2}                                                          \\
& = & (r-1)\, a
      -
      \frac{\textstyle r^{2}}
           {\textstyle 4(r+1)}
      \cdot
      \frac{\textstyle b^{2}}
           {\textstyle c}\, ;
\end{eqnarray*}
thus,
${\mathsf{h}}(y_{0}) > 0$
if and only if
\begin{eqnarray*}
      \frac{\textstyle 4\, ac}
           {\textstyle b^{2}}
& > & \frac{\textstyle r^{2}}
           {\textstyle r^{2} - 1}\, .
\end{eqnarray*}
We now verify the latter condition.
Note that
{
\small
\begin{eqnarray*}
      \frac{\textstyle 4\, ac}
           {\textstyle b^{2}}
& = & 4\,
      \alpha
      \beta
      \frac{\textstyle (r!)^{2}}
           {\textstyle (r-1)!
                       (r+1)!}
      \frac{\textstyle (s!)^{2}}
           {\textstyle (s-1)!
                       (s+1)!}
      \frac{\textstyle {\mathsf{f}}(p, r-1)
                       {\mathsf{f}}(p, r+1)}
           {\textstyle ({\mathsf{f}}(p, r))^{2}}
      \frac{\textstyle {\mathsf{f}}(q, s-1)
                       {\mathsf{f}}(q, s+1)}
           {\textstyle ({\mathsf{f}}(q, s))^{2}}\, .
\end{eqnarray*}
}
By Observation~\ref{geometric mean},
it follows that
\begin{eqnarray*}
      \frac{\textstyle 4\, ac}
           {\textstyle b^{2}}
& >  & 2
       \cdot
       \frac{\textstyle r}
            {\textstyle r+1}
       \cdot
       \frac{\textstyle s}
            {\textstyle s+1}                                             \\
&\geq& \frac{\textstyle r}
            {\textstyle r+1}
       \cdot
       \frac{\textstyle 3}
            {\textstyle 2}                                               \\
&\geq& \frac{\textstyle r^{2}}
            {\textstyle r^{2} - 1}\, ,
\end{eqnarray*}
since $s \geq 3$
and
$r \geq 3$.
\end{proof}

\noindent
We now continue with the proof of Lemma~\ref{even moments}.

\begin{proof}
If Condition {\sf (C.1)} holds,
then we are done.
So,
assume that
Condition {\sf (C.1)} does {\em not} hold,
so that player $i$
is mixed.
If Condition {\sf (C.3)} holds,
then we are done.
So,
assume that
Condition {\sf (C.3)}
does {\em not} hold.
Since player $i$ is mixed,
this implies that
${\mathsf{kM}}_{i}\left( p_{i}^{1},
                         {\bf p}_{-i}
                 \right)
 \geq
 {\mathsf{kM}}_{i}\left( p_{i}^{2},
                         {\bf p}_{-i}
                 \right)$.
We shall establish Condition {\sf (C.2)}.
We proceed by case analysis.

\noindent
\underline{{\sf The cases $k=2$ and $k=4$:}}
By Proposition~\ref{formula for k moment},
for each link $\ell \in [2]$,
\begin{eqnarray*}
      {\mathsf{2M}}_{i}\left( p_{i}^{\ell},
                              {\bf p}_{-i}
                       \right)
& = & \sum_{j \in [n] \setminus \{ i \} \mid \omega(i, j, \ell) \neq 0}
        {\mathsf{f}}\left( p_{j}(\ell),
                           2
                    \right)
        \left( \omega (i, j, \ell)
        \right)^{2}
\end{eqnarray*}
and
\begin{eqnarray*}
&   & {\mathsf{4M}}_{i}\left( p_{i}^{\ell},
                              {\bf p}_{-i}
                       \right)                                                                    \\
& = &  \sum_{j \in [n] \setminus \{ i \} \mid \omega (i, j, \ell) \neq 0}                         
        {\mathsf{f}}\left( p_{j}(\ell),
                           4
                    \right)
        \left( \omega (i, j, \ell)
        \right)^{4}                                                                               \\
&   &  +                                                                                          \\
&   &  \sum_{j, k \in [n] \setminus \{ i \}, j \neq k
             \mid
             \omega (i, j, \ell) \neq 0,
             \omega (i, k, \ell) \neq 0
            }
        {\mathsf{f}}\left( p_{j}(\ell),
                           2
                    \right)
        \cdot
        {\mathsf{f}}\left( p_{k}(\ell),
                           2
                    \right)\,
       \left( \omega (i, j, \ell)
              \omega (i, k, \ell)
       \right)^{2}\, .
\end{eqnarray*}
Since
{\it (i)}
$p_{j}(\overline{\ell})
 =
 1 - p_{j}(\ell)$,
{\it (ii)}
${\mathsf{f}}(x, 2)
 =
 {\mathsf{f}}(1-x, 2)$
for $x \in [0, 1]$,
\& 
{\it (iii)}
for each player $j \in [n] \setminus \{ i \}$,
$\omega (i, j, 1)
 \neq
 0$
if and only if
$\omega (i, j, 2)
 \neq
 0$,
it follows that
{
\small
\begin{eqnarray*}
&   & \underbrace{{\mathsf{2M}}_{i}\left( p_{i}^{1},
                                         {\bf p}_{-i}
                                  \right)
                 -
                 {\mathsf{2M}}_{i}\left( p_{i}^{2},
                                         {\bf p}_{-i}
                                  \right)
                }_{\geq 0}                                                  \\
& = & \sum_{j \in [n] \setminus \{ i \}
            \mid
            \omega (i, j, 1) \neq 0,
            \omega (i, j, 2) \neq 0
           }
        {\mathsf{f}}\left( p_{j}(\ell),
                           2
                    \right)
        \cdot
        \underbrace{\left( \left( \omega (i, j, 1)
                           \right)^{2}
                           -
                           \left( \omega (i, j, 2)
                           \right)^{2}
                    \right)}_{< 0}\, ,
\end{eqnarray*}
}
and
{\small
\begin{eqnarray*}
&   & \underbrace{{\mathsf{4M}}_{i}\left( p_{i}^{1},
                                          {\bf p}_{-i}
                                   \right)
                   -
                   {\mathsf{4M}}_{i}\left( p_{i}^{2},
                                          {\bf p}_{-i}
                                   \right)}_{\geq 0}                         \\
& = & \sum_{j \in [n] \setminus \{ i \}
            \mid
            \omega (i, j, 1) \neq 0,
            \omega (i, j, 2) \neq 0}
        {\mathsf{f}}\left( p_{j}(1),
                           4
                    \right)
        \cdot
        \underbrace{\left( \left( \omega (i, j, 1)
                           \right)^{4}
                           -
                           \left( \omega (i, j, 2)
                           \right)^{4}
                    \right)}_{< 0}                                          \\
&   & +
      \sum_{j, k \in [n] \setminus \{ i \}, j \neq k
              \mid
              \omega (i, j, 1)
              \cdot
              \omega (i, j, 2)
              \cdot
              \omega (i, k, 1)
              \cdot
              \omega (i, k, 2)
              \neq 0
           }
        {\mathsf{f}}\left( p_{j}(1),
                           2
                    \right)
        \cdot
        {\mathsf{f}}\left( p_{k}(1),
                           2
                    \right)\,                                               \\
&   & \cdot
      \underbrace{\left( \left( \omega (i, j, 1)
                           \right)^{2}\,
                           \left( \omega (i, k, 1)
                           \right)^{2}
                           -
                           \left( \omega (i, j, 2)
                           \right)^{2}\,
                           \left( \omega (i, k, 2)
                           \right)^{2}
                    \right)}_{< 0}\, .
\end{eqnarray*}
}
Since the two links are ordered,
it follows that
for each player
$j \in [n] \setminus \{ i \}$,
${\mathsf{f}}\left( p_{j}(1),
                    2
             \right)
 =
 {\mathsf{f}}\left( p_{j}(1),
                    4
             \right)
 =
 0$,
which implies that
$p_{j}(1) = 1$
and player $j$ is pure,
as needed for
Condition {\sf (C.2)}.

For the remaining cases
$k \in \{ 6, 8 \}$,
we shall be establishing
that
${\mathsf{kM}}_{i}$
increases strictly monotone
in the weights.
Since the two links are ordered,
the inequality
${\mathsf{kM}}_{i}\left( p_{i}^{1},
                         {\bf p}_{-i}
                 \right)
 \geq
 {\mathsf{kM}}_{i}\left( p_{i}^{2},
                         {\bf p}_{-i}
                 \right)$
implies
that for each player $j \in [n] \setminus \{ i \}$
with
$\omega (i, j, 1)
 \neq
 0$
and
$\omega (i, j, 2)
 \neq
 0$,
${\mathsf{f}}\left( p_{j}(1),
                    r_{j}
             \right)
 =
 {\mathsf{f}}\left( p_{j}(1),
                    r_{j}
             \right)
 =
 0$,
which implies that
$p_{j}(1) = 1$
and player $j$ is pure,
as needed for
Condition {\sf (C.2)}.
We shall be referring again
to the formula
from Proposition~\ref{formula for k moment}.

\noindent
\underline{The case $k=6$:}
Note that $6$
can be partitioned
as
$6$,
$(4,2)$,
$(3,3)$,
$(2,4)$
and
$(2,2,2)$.
The terms of the formula
corresponding to
the partitions
$6$
and
$(2,2,2)$
have strictly positive coefficients;
so,
they are increasing strictly monotone
in the weights.
Now
group together
the terms of the formula
corresponding to
the partitions
with
$(r_{j}, r_{k})
 \in
 \left\{ (4,2),
         (3,3),
         (2,4)
 \right\}$
in a single sum.
Lemma~\ref{embracing helps}
(with $\alpha = \beta = 1$)
implies that the sum
increases strictly monotone
in the weights.

\noindent
\underline{The case $k=8$:}
The only partitions of $8$
that use odd numbers
are
{\it (i)}
$(5,3)$ and $(3,5)$,
and
{\it (ii)}
$(3,3,2)$,
$(3,2,3)$
and
$(2,3,3)$.
Partitions in {\it (i)}
involve {\em two} strategies,
while
partitions in {\it (ii)}
involve {\em three} strategies.

\noindent
\underline{Case {\it (i)}:}
Let $j$ and $k$
be the two strategies
with
$0 < p_{j}(\ell) < \frac{1}
                        {2}$
and
$\frac{\textstyle 1}
      {\textstyle 2}
 <
 p_{k}(\ell)
 <
 1$
for a link $\ell \in [2]$;
so,
${\mathsf{f}}\left( p_{j}(\ell),
                    3
             \right)
 \cdot
 {\mathsf{f}}\left( p_{k}(\ell),
                    5
             \right)
 <
 0$
and
${\mathsf{f}}\left( p_{j}(\ell),
                    5
             \right)
 \cdot
 {\mathsf{f}}\left( p_{k}(\ell),
                    3
             \right)
 <
 0$.

Consider now
the terms of the formula
corresponding to the partitions
with
\begin{eqnarray*}
        (r_{j}, r_{k})
& \in & \left\{ (6,2),
                (5,3),
                (4,4),
                (3,5),
                (2,6)
        \right\}\, .
\end{eqnarray*}
First
group together
the terms
corresponding to
the partitions
with
\begin{eqnarray*}
        (r_{j}, r_{k})
& \in & \left\{ (6,2),
                (5,3),
                (4,4)
        \right\}
\end{eqnarray*}
in a single sum,
and invoke Lemma~\ref{embracing helps}
with $\alpha = 1$ and $\beta = \frac{\textstyle 1}
                                    {\textstyle 2}$;
it follows that
the sum
increases strictly monotone
in the weights.
Then,
group together
the terms
corresponding to
the partitions
with
\begin{eqnarray*}
        (r_{j}, r_{k})
& \in & \left\{ (2,6),
                (3,5),
                (4,4)
        \right\}
\end{eqnarray*}
in a single sum,
and invoke Lemma~\ref{embracing helps}
with $\alpha = 1$ and $\beta = \frac{\textstyle 1}
                                    {\textstyle 2}$;
it follows that
the sum
increases strictly monotone
in the weights.

\noindent
\underline{Case {\it (ii)}:}
Let $j$, $k$ and $t$
be the three strategies
with
$0 < p_{j}(1) < \frac{\textstyle 1}
                     {\textstyle 2}$
and
$\frac{\textstyle 1}
      {\textstyle 2}
 <
 p_{k}(1)
 <
 1$
for a link $\ell \in [2]$;
assume,
without loss of generality,
that
$\frac{\textstyle 1}
      {\textstyle 2}
 <
 p_{k}(\ell)
 <
 1$;
so,
${\mathsf{f}}\left( p_{j}(\ell),
                    3
             \right)
 \cdot
 {\mathsf{f}}\left( p_{k}(\ell),
                    3
             \right)
 <
 0$
and
${\mathsf{f}}\left( p_{j}(\ell),
                    3
             \right)
 \cdot
 {\mathsf{f}}\left( p_{t}(\ell),
                    3
             \right)
 <
 0$,
while
${\mathsf{f}}\left( p_{j}(\ell),
                    3
             \right)
 \cdot
 {\mathsf{f}}\left( p_{t}(\ell),
                    3
             \right)
 >
 0$.
Consider now
the terms of the formula
corresponding to the partitions
with
\begin{eqnarray*}
        (r_{j}, r_{k}, r_{t})
& \in & \left\{ (4,2,2),
                (3,3,2),
                (2,4,2),
                (3,2,3),
                (2,2,4)
        \right\}\, .
\end{eqnarray*}
In the same way as for Case {\it (i)},
Lemma~\ref{embracing helps}
implies that
the sum of the terms
corresponding to these partitions
increases strictly monotone
in the weights.
\end{proof}

\noindent
We continue with the proof of Proposition~\ref{the real almost pure property}.
Consider the $k$-moment valuation
${\mathsf{kM}}$
with $k \in \{ 2, 4, 6, 8 \}$.
So,
Lemma~\ref{even moments} applies.
Assume,
by way of contradiction,
that Condition {\sf (C.3)} holds;
then,
${\mathsf{kM}}_{i}\left( p_{i}^{1},
                         {\bf p}_{-i}
                 \right)
 <
 {\mathsf{kM}}_{i}\left( p_{i}^{2},
                        {\bf p}_{-i}
                 \right)$,
which implies,
by definition of ${\mathsf{R}}$,
that
${\mathsf{R}}_{i}\left( p_{i}^{1},
                        {\bf p}_{-i}
                 \right)
 <
 {\mathsf{R}}_{i}\left( p_{i}^{2},
                       {\bf p}_{-i}
                 \right)$.
By the {\it Weak-Equilibrium-for-Expectation} property,
${\mathsf{E}}_{i}\left( p_{i}^{1},
                        {\bf p}_{-i}
                 \right)
 =
 {\mathsf{E}}_{i}\left( p_{i}^{2},
                       {\bf p}_{-i}
                 \right)$.
Hence,
${\mathsf{V}}_{i}\left( p_{i}^{1},
                        {\bf p}_{-i}
                 \right)
 <
 {\mathsf{V}}_{i}\left( p_{i}^{2},
                       {\bf p}_{-i}
                 \right)$.
Since ${\mathsf{V}}$ is concave,
the {\it Optimal-Value} property
(Proposition~\ref{constant value})
implies that
${\mathsf{V}}_{i}\left( p_{i}^{1},
                        {\bf p}_{-i}
                 \right)
 =
 {\mathsf{V}}_{i}\left( p_{i}^{2},
                       {\bf p}_{-i}
                 \right)$.
A contradiction.
Hence,
by Lemma~\ref{even moments},
either
{\sf (C.1)}
player $i$ is pure,
or
{\sf (C.2)}
all players
$i' \in [n] \setminus \{ i \}$
with $\omega (i, i', 1) \neq 0$
are pure.
\end{proof}

\subsection{${\mathcal{NP}}$-Hardness Result}
\label{two strategies strong np hardness}

\noindent
We give an example
of a player-specific scheduling game ${\mathsf{G}}$
with $n = 3$ players $0$, $1$ and $2$
on two ordered links $1$ and $2$,
with no ${\mathsf{V}}$-equilibrium
for an
$({\mathsf{E}}
  +
  {\mathsf{R}})$-valuation 
${\mathsf{V}}$,
where ${\mathsf{R}}$
is 
{\sf (1)}
${\mathsf{Var}}$,
or
{\sf (2)}
${\mathsf{SD}}$,
or
{\sf (3)}
a concave linear sum
$\sum_{k \in \{ 2, 4, 6, 8 \}}
   \alpha_{k}
   \cdot
   {\mathsf{kM}}$,
with $\alpha_{k} \geq 0$
for $k \in \{ 2, 4, 6, 8 \}$.
For $i \in \{ 0, 1, 2 \}$
and $\ell \in [2]$,
set
\begin{eqnarray*}
       \omega (i, i, \ell)
& := & 0\, ;                                                     \\
       \omega (i, (i+1)\ \mbox{mod}\ 3, \ell)
& := & \delta_{\ell\, 2}\, ;                                      \\
       \omega (i, (i+2)\ \mbox{mod}\ 3, \ell)
& := & 2 + \delta_{\ell\, 2}\, .
\end{eqnarray*}
($\delta$ is the {\it Kronecker delta}:
$\delta_{\ell 2} = 1$ if $\ell = 2$
and $0$ otherwise.)
Assume,
by way of contradiction,
that ${\mathsf{G}}$ has a pure equilibrium.
If the three players are on the same link,
then each player has cost greater than $0$
and can reduce her cost
by switching to the other link.
If two of the players are
on the same link $\ell \in [2]$
while the third player
is on link $\overline{\ell}$.
($\overline{\ell}$ denotes the link
other than $\ell$.)
Then,
there is a player $i \in \{ 0, 1, 2 \}$
such that player $i$ is on link $\ell$
while player $(i+1)\ \mbox{mod}\ 3$
is on link $\overline{\ell}$;
her cost on link $\ell$
is $2 + \delta_{\ell\, 2}$,
and she can reduce her cost to
$\delta_{\overline{\ell}\, 2}$
by switching to link $\overline{\ell}$.
A contradiction in both cases.
Finally,
assume that there is a ${\mathsf{V}}$-equilibium ${\bf p}$
where some player $i \in \{ 0, 1, 2 \}$ is mixed.
By the {\it Mixed-Player-Has-Pure-Neighbors} property
(Proposition~\ref{the real almost pure property}),
players $(i+1)\ \mbox{mod}\ 3$
and
$(i+2)\ \mbox{mod}\ 3$
are pure.
By the {\it Weak-Equilibrium-for-Expectation} property
(Corollary~\ref{brand new}
and Proposition~\ref{concavity implies weak equilibrium}),
${\mathsf{E}}_{i}\left( p_{i}^{1},
                       {\bf p}_{-i}
                \right)
 =
 {\mathsf{E}}_{i}\left( p_{i}^{2},
                       {\bf p}_{-i}
                \right)$
or
$\sum_{i' \neq i \mid p_{i'} = p_{i'}^{1}}
   \omega (i, i', 1)
 = 
 \sum_{i' \neq i \mid p_{i'} = p_{i'}^{2}}
   \omega (i, i', 2)$. 
By the definition of weights,
a contradiction follows.
Thus,
{\sf $\exists {\mathsf{V}}$-EQUILIBRIUM}
is non-trivial
in the considered setting.
We show:

\begin{theorem}
\label{concrete two ordered links}
Fix an
$({\mathsf{E}}
  +
  {\mathsf{R}})$-valuation
${\mathsf{V}}$,
where
${\mathsf{R}}$
is 
{\sf (1)}
${\mathsf{Var}}$,
or
{\sf (2)}
${\mathsf{SD}}$,
or
{\sf (3)}
a concave linear sum
$\sum_{k \in \{ 2, 4, 6, 8 \}}
   \alpha_{k}
   \cdot
   {\mathsf{kM}}$,
with $\alpha_{k} \geq 0$
for $k \in \{ 2, 4, 6, 8 \}$.
Then,
{\sf $\exists {\mathsf{V}}$-EQUILIBRIUM}
is strongly ${\mathcal{NP}}$-hard
for player-specific scheduling games
on two ordered links.
\end{theorem}

\noindent
The proof will use a reduction
from {\sf MULTIBALANCED PARTITION},
a problem we introduce and show
strongly ${\mathcal{NP}}$-complete:

\noindent
\begin{tabular}{lp{14cm}l}
\hline
{\sc I.:} & $\langle n, m,
                           {\mathsf{A}}
                   \rangle$,
                  with integers $n$, $m$
                  and
                  a set
                  ${\mathsf{A}}
                   =
                   \left\{ a_{ij}
                           \mid
                           i \in [n], j \in [m]
                   \right\}$.
\\
{\sc Q.:} & Is there a subset
                  ${\mathsf{I}}
                   \subset
                   [n]$
                  such that
                  for each $j \in [m]$,
                  $\sum_{i \in {\mathsf{I}}}
                     a_{ij}
                  =
                  3
                  +
                  2\, \sum_{i \not\in {\mathsf{I}}}
                        a_{ij}$?
\\
\hline
\end{tabular}

\noindent
We show:

\begin{proposition}
\label{little shantry}
{\sf MULTIBALANCED PARTITION}
is strongly ${\mathcal{NP}}$-complete.
\end{proposition}

\noindent
The proof of Proposition~\ref{little shantry}
employs a reduction from
{\sf 3-DIMENSIONAL MATCHING}~\cite[SP1]{GJ79}:

\noindent
\begin{tabular}{lp{14cm}l}
\hline
{\sc I.:} & $\langle {\mathsf{W}},
                           {\mathsf{X}},
                           {\mathsf{Y}},
                           {\mathsf{M}}
                   \rangle$,
                  where ${\mathsf{M}}$
                  is a set with
                  ${\mathsf{M}}
                   \subseteq
                   {\mathsf{W}}
                   \times
                   {\mathsf{X}}
                   \times
                   {\mathsf{Y}}$
                  and
                  ${\mathsf{W}}$,
                  ${\mathsf{X}}$
                  and
                  ${\mathsf{Y}}$
                  are disjoint                                                           \\
          &      sets with
                  $|{\mathsf{W}}|
                   =
                   |{\mathsf{X}}|
                   =
                   |{\mathsf{Y}}|
                   =
                   q$.
\\
{\sc Q.:} & Does ${\mathsf{M}}$
                  contain a {\it matching,}
                  i.e.,
                  a subset
                  ${\mathsf{M}'}
                   \subset
                   {\mathsf{M}}$
                  with
                  $|{\mathsf{M}'}|
                   =
                   q$
                  such that no two                                                     \\
          &     elements of
                  ${\mathsf{M}'}$
                  agree in any coordinate?
\\
\hline
\end{tabular}


\begin{proof}
Clearly,
{\sf MULTIBALANCED PARTITION}
belongs to
${\mathcal{NP}}$.
For the ${\mathcal{NP}}$-hardness
we shall employ a
reduction from
{\sf 3-DIMENSIONAL MATCHING}.\footnote{The reduction is very similar
                                       to the one used
                                       in the proof of~\cite[Theorem 3.5]{GJ79}.}
Given an instance
$\langle {\mathsf{W}},
         {\mathsf{X}},
         {\mathsf{Y}},
         {\mathsf{M}}
 \rangle$
of {\sf 3-DIMENSIONAL MATCHING},
with
${\mathsf{W}}
 =
 \left\{ w_{1}, \ldots, w_{q}
 \right\}$,
${\mathsf{X}}
 =
 \left\{ x_{1}, \ldots, x_{q}
 \right\}$,
${\mathsf{Y}}
 =
 \left\{ y_{1}, \ldots, y_{q}
 \right\}$
and
${\mathsf{M}}
 =
 \left\{ m_{1}, \ldots, m_{k}
 \right\}$,
where for each $i \in [k]$,
$m_{i} = (w_{{\mathsf{f}}(i)}, x_{{\mathsf{g}}(i)}, y_{{\mathsf{h}}(i)})$
with functions
${\mathsf{f}},
 {\mathsf{g}},
 {\mathsf{h}}:
 [k]
 \rightarrow
 [q]$
giving the first, second and third coordinate,
respectively,
of each element
$m_{i} \in {\mathsf{M}}$,
we construct
an instance
$\langle n, m, {\mathsf{A}}
 \rangle$
of {\sf MULTIBALANCED PARTITION}
as follows:
\begin{quote}
\begin{itemize}

\item
$n := k+1$
and
$m := 3q$.

\item
For $1 \leq i \leq k$
and
$1 \leq j \leq m$,
{
\small
\begin{eqnarray*}
       a_{ij}
& := & \left\{ \begin{array}{ll}
                1\, , & \mbox{if ($1 \leq j \leq q$
                                  and
                                  $j = {\mathsf{f}}(i)$)}          \\
                      & \mbox{or ($q+1 \leq j \leq 2q$
                                  and
                                  $j-q = {\mathsf{g}}(i)$)}        \\
                      & \mbox{or ($2q+1 \leq j \leq 3q$
                                  and
                                  $j-2q = {\mathsf{h}}(i)$)}         \\
               0\, , & \mbox{otherwise}
              \end{array}
       \right.\, ,
\end{eqnarray*}
}
and for $1 \leq j \leq m$,
{
\small
\begin{eqnarray*}
       a_{k+1,j}
& := & 2\, b_{j}\, ,
\end{eqnarray*}
}
where
{
\small
\begin{eqnarray*}
       b_{j} 
& := & \sum_{i \in [k]}
            a_{ij}\, .
\end{eqnarray*}            
}

\end{itemize}
\end{quote}
We prove:

\begin{lemma}
\label{hotel arosa}
$\langle {\mathsf{W}},
         {\mathsf{X}},
         {\mathsf{Y}},
         {\mathsf{M}}
 \rangle$
has a solution
if and only if
$\langle n, m, {\mathsf{A}}
 \rangle$
has a solution.
\end{lemma}

\begin{proof}
``$\Rightarrow$'':
Assume first that
$\langle {\mathsf{W}},
         {\mathsf{X}},
         {\mathsf{Y}},
         {\mathsf{M}}
 \rangle$
has a solution
${\mathsf{M}}'$.
Set
\begin{eqnarray*}
      {\mathsf{I}}'
& := &
 \left\{ i \in [k]
         \mid
         m_{i} \in {\mathsf{M}}'
 \right\}\, .
\end{eqnarray*} 
Then,
clearly,
for each $j \in [m]$,
$\sum_{i \in {\mathsf{I}}'}
   a_{ij}
 =
 1$.
Hence,
for each $j \in [m]$,
\begin{eqnarray*}
      \sum_{i \in [k] \setminus {\mathsf{I}}'}
        a_{ij}
& = & \sum_{i \in [k]}
        a_{ij}
      -
      \sum_{i \in {\mathsf{I}}'}
        a_{ij}\ \
  =\ \ b_{j} - 1\, .
\end{eqnarray*}
Set now
${\mathsf{I}}
 :=
 {\mathsf{I}}'
 \cup
 \{ k+1 \}$.
Then,
for each $j \in [m]$,
\begin{eqnarray*}
      \sum_{i \in {\mathsf{I}}}
        a_{ij}
& = & \sum_{i \in {\mathsf{I}}'}
        a_{ij}
      +
      a_{k+1,j}                                       \\
& = & 1 + 2\, b_{j}                                   \\
& = & 3 + 2\, (b_{j} - 1)                             \\
& = & 3 + 2\, \sum_{i \in [k] \setminus {\mathsf{I}}'}
                a_{ij}                                \\
& = & 3 + 2\, \sum_{i \not\in {\mathsf{I}}}
                a_{ij}\, ,
\end{eqnarray*}
so that
${\mathsf{I}}$
is a solution of
$\langle n, m, {\mathsf{A}}
 \rangle$.

``$\Leftarrow$'':
Assume now that
$\langle n, m, {\mathsf{A}}
 \rangle$
has a solution
${\mathsf{I}}$.
We claim that
$k+1 \in {\mathsf{I}}$:
Assume,
by way of contradiction,
that
$k+1 \not \in {\mathsf{I}}$.
Then,
for each $j \in [m]$,
\begin{eqnarray*}
       \lefteqn{3 + 2\, \sum_{i \not\in {\mathsf{I}}}
                          a_{ij}}                                                  \\
\geq & 3 + 2\, a_{k+1,j}
     & \mbox{(since $k+1 \not\in {\mathsf{I}}$)}                                  \\
=    & 3 + 4 b_{j}
     & \mbox{(by definition of $a_{k+1,j}$)}                                       \\
>    & b_{j}
     &                                                                            \\
=    & \sum_{i \in [k]}
         a_{ij}
     & \mbox{(by definition of $b_{j}$)}                                          \\
\geq & \sum_{i \in {\mathsf{I}}}
         a_{ij}\, ,
     &
\end{eqnarray*}
a contradiction to the assumption that
${\mathsf{I}}$
is a solution of
$\langle n, m, {\mathsf{A}}
 \rangle$.
So,
$k+1 \in {\mathsf{I}}$.

Fix now an arbitrary $j \in [m]$.
Set
${\mathsf{\Delta}}_{j}
 :=
 \sum_{i \in {\mathsf{I}} \setminus \{ k+1 \}}
   a_{ij}$.
Since ${\mathsf{I}}$ is a solution of
$\langle n, m, {\mathsf{A}}
 \rangle$,
{
\small
\begin{eqnarray*}
      \sum_{i \in {\mathsf{I}}}
        a_{ij}
& = & 3 +
      2\, \sum_{i \not\in {\mathsf{I}}}
            a_{ij}\, .
\end{eqnarray*}
}
Since $k+1 \in {\mathsf{I}}$,
this implies that
{
\small
\begin{eqnarray*}
      a_{k+1,j}
      +
      \sum_{i \in {\mathsf{I}} \setminus \{ k+1 \}}
        a_{ij}
& = & 3 +
      2\, \left( \sum_{i \in [k]}
                   a_{ij}
                 -
                 \sum_{i \in {\mathsf{I}} \setminus \{ k+1 \}}
                   a_{ij}
          \right)\, .
\end{eqnarray*}
}
Hence,
{
\small
\begin{eqnarray*}
      2\, b_{j} +
      {\mathsf{\Delta}}_{j}
& = & 3 +
      2\,
      \left( b_{j} - {\mathsf{\Delta}}_{j}
      \right)\, .
\end{eqnarray*}
}
It follows that
${\mathsf{\Delta}}_{j}
 =
 1$,
so that
${\mathsf{I}} \setminus \{ k+1 \}$
is a solution of
$\langle {\mathsf{W}},
         {\mathsf{X}},
         {\mathsf{Y}},
         {\mathsf{M}}
 \rangle$.
\end{proof}

\noindent
Lemma~\ref{hotel arosa}
establishes the reduction
for the ${\mathcal{NP}}$-hardness;
since the number involved in the reduction
are polynomially bounded,
strong ${\mathcal{NP}}$-hardness follows.
\end{proof}

\noindent
For the proof 
of the reduction
for Theorem~\ref{concrete two ordered links},
we shall use the {\it Mixed-Player-Has-Pure-Neighbors} property
(Proposition~\ref{the real almost pure property})
to identify the mixed players;
in turn,
we shall apply the {\it Weak-Equilibrium-for-Expectation} property
to the mixed players 
to obtain a solution to {\sf MULTIBALANCED PARTITION}
or a ${\mathsf{V}}$-equilibrium.
We continue with the proof
of Theorem~\ref{concrete two ordered links}.

\begin{proof}
We shall employ
a reduction from
{\sf MULTIBALANCED PARTITION}.
Given an instance
$\langle n, m, {\mathsf{A}}
 \rangle$
of {\sf MULTIBALANCED PARTITION},
we construct an instance
$\langle {\mathsf{G}}
 \rangle$
of
{\sf $\exists {\mathsf{V}}$-EQUILIBRIUM}
as follows.
${\mathsf{G}}$
is a player-specific scheduling game
on two ordered links
with $n+5m$ players
in the player set
${\mathsf{\Pi}}
 :=
 [n]
 \cup
 \{ [k, j]
    \mid
    k \in [m],
    j \in \{ 0, \ldots, 4
          \}
 \}$. 
Set $M := \max_{k} \sum_{j \in [n]} a_{kj}$.
We assume that $M \geq 4$.
We now define
the weights
$\omega (\pi_{1}, \pi_{2}, \ell)$
where
$\pi_{1}, \pi_{2}
 \in
 {\mathsf{\Pi}}$
and
$\ell \in [2]$,
where $\delta$
is the {\it Kronecker delta}:
$\delta_{\ell\ell'}
 =
  1$
if $\ell = \ell'$
and $0$
otherwise.

\begin{enumerate}

\item
\underline{$\pi_{1} = [k, j]$
           with
           $k \in [m]$ and
           $ j \in \{ 0, \ldots, 3
                  \}$:}
Then,
{
\small
\begin{eqnarray*}
            \omega ([k, j], \pi_{2}, \ell)  
& := & \left\{ \begin{array}{ll}
                 M + \delta_{\ell 2}\, ,     & \mbox{if $\pi_{2} = [k, i]$
                                                     with $i = 4$
                                                     or $i \neq (j+1)\ \mbox{mod}\ 4$}          \\
                 M - 4 + \delta_{\ell 2}\, , & \mbox{if $\pi_{2} = [k, (j+1)\ \mbox{mod}\ 4]$}  \\
                 0\, ,                       & \mbox{otherwise}
               \end{array}
       \right.\, .
\end{eqnarray*}
}

\item
\underline{$\pi_{1} = [k, 4]$
           with
           $k \in [m]$:}
Then,
{
\small
\begin{eqnarray*}
       \omega ([k, 4], \pi_{2}, \ell)
& := & \left\{ \begin{array}{ll}
                 M + \delta_{\ell 2}\, ,     & \mbox{if $\pi_{2} = [k, i]$
                                                     with $i \in [4]$}                          \\
                 a_{kj}\, ,                  & \mbox{if $\pi_{2} = j \in [n]$
                                                     and $\ell = 1$}                            \\
                 2\, a_{kj}\, ,              & \mbox{if $\pi_{2} = j \in [n]$
                                                     and $\ell = 2$}                            \\
                 0\, ,                       & \mbox{otherwise}
               \end{array}
       \right.\, .
\end{eqnarray*}
}

\item
\underline{$\pi_{1} = i \in [n]$:}
Then,
$\omega (i, \pi_{2}, \ell) = 0$
for all $\pi_{2} \in {\mathsf{\Pi}}$
and $\ell \in [2]$.

\end{enumerate}
Note that
${\mathsf{G}}$
is a player-specific scheduling game
on two ordered links.
We now prove:

\begin{lemma}
\label{hilton fontana}
In a ${\mathsf{V}}$-equilibrium ${\bf p}$,
for each $k \in [m]$,
there is an index
$j \in \{ 0, \ldots, 4 \}$
such that
player $[k, j]$
is non-pure.
\end{lemma}

\begin{proof}
Assume,
by way of contradiction,
that there is an index
$\widehat{k} \in [m]$
such that
all players
$[\widehat{k}, j]$
with $j \in \{ 0, \ldots, 4 \}$
are pure;
so,
$p_{[\widehat{k}, j]}[1]
 \in
 \{ 0, 1 \}$
for all
$j \in \{ 0, \ldots, 4 \}$.
If there are at least four players
$[\widehat{k}, j]$
with $j \in \{ 0, \ldots, 4 \}$
choosing the same link,
then
at least one player
$[\widehat{k}, j]$
with $j \in \{ 0, \ldots, 3 \}$
can improve her cost
by switching to the other link.
So,
assume that there are three players
$[\widehat{k}, j]$
with $j \in \{ 0, \ldots, 4 \}$
choosing some link $\ell \in [2]$,
while the remaining two players
choose the other link.
Then,
there is an index
$j_{0} \in \{ 0, \ldots, 3 \}$
such that
player $[\widehat{k}, j_{0}]$
chooses link $\ell$
while player
$[\widehat{k}, (j_{0} +1)\ \mbox{mod}\ 4]$
chooses link $\overline{\ell}$.
So,
${\mathsf{\mu}}_{[\widehat{k}, j_{0}]}({\bf p})
 =
 3M + 3 \delta_{\ell 2}$,
and player $[\widehat{k}, j_{0}]$
can improve
by switching to link $\overline{\ell}$
where her cost becomes
$3M - 4 + 3 \delta_{\overline{\ell} 2}
 <
 3M + 3 \delta_{\ell 2}$.
A contradiction
to the assumption that
${\bf p}$
is a ${\mathsf{V}}$-equilibrium.
\end{proof}

\noindent
We next prove:

\begin{lemma}
\label{hilton percola}
In a ${\mathsf{V}}$-equilibrium,
each player
$[k, j]$,
with
$k \in [m]$ and $j \in \{ 0, \ldots, 3 \}$,
is pure.
\end{lemma}

\begin{proof}
Assume,
by way of contradiction,
that there is a non-pure player
$\pi = [\widehat{k}, \widehat{j}]$
with $\widehat{k} \in [m]$
and
$\widehat{j} \in \{ 0, \ldots, 3 \}$.
Then,
by the {\it Mixed-Player-Has-Pure-Neighbors} property
(Proposition~\ref{the real almost pure property}),
all players
$[\widehat{k}, j]$
with
$j \in \{ 0, \ldots, 4 \} \setminus \{ \widehat{j} \}$
are pure.
By the {\it Weak-Equilibrium-for-Expectation} property
for player $\pi$,
${\mathsf{E}}_{\pi}\left( p_{\pi}^{1},
                          {\bf p}_{-\pi}
                   \right)
 =
 {\mathsf{E}}_{\pi}\left( p_{\pi}^{2},
                          {\bf p}_{-\pi}
                   \right)$.
For each link $\ell \in \{ 1, 2 \}$,
denote as
${\mathsf{\Phi}}_{\ell}$
the set of players
$[\widehat{k},j]$
with $j \in \{ 0, \ldots, 4 \}
            \setminus
            \{ \widehat{j} \}$
choosing link $\ell$;
set
$y_{\ell} := 1$
if
$[\widehat{k}, (\widehat{j}+1)\ \mbox{mod}\ 4]
 \in
 {\mathsf{\Phi}}_{\ell}$
and $0$ otherwise.
Clearly,
$|y_{1} - y_{2}|
 =
 1$.
Then,
\begin{eqnarray*}
      {\mathsf{E}}_{\pi}\left( p_{\pi}^{1},
                              {\bf p}_{-\pi}
                        \right)
& = & |{\mathsf{\Phi}}_{1}|
      \cdot
      M
      -
      4\, y_{1}\, ,
\end{eqnarray*}
and
\begin{eqnarray*}
      {\mathsf{E}}_{\pi}\left( p_{\pi}^{2},
                              {\bf p}_{-\pi}
                        \right)
& = & |{\mathsf{\Phi}}_{2}|
      \cdot
      (M+1)
      -
      4\, y_{2}\, .
\end{eqnarray*}
We proceed by case analysis.
\begin{enumerate}

\item
Assume first that
$|{\mathsf{\Phi}}_{1}|
 >
 |{\mathsf{\Phi}}_{2}|$.
So,
$|{\mathsf{\Phi}}_{1}|
 \geq
 3$
and
$|{\mathsf{\Phi}}_{2}|
 \leq
 1$.
Then,
\begin{eqnarray*}
         {\mathsf{E}}_{\pi}\left( p_{\pi}^{1},
                                 {\bf p}_{-\pi}
                           \right)
& \geq & 3\, M - 4\, y_{1}\ \
         >\ \
         M+1 - 4\, y_{2}\ \
         \geq\ \
         {\mathsf{E}}_{\pi}\left( p_{\pi}^{2},
                                 {\bf p}_{-\pi}
                           \right)\, ,
\end{eqnarray*}
since
$3 M > M+5 \geq M+1 + 4\, (y_{1} - y_{2})$.
A contradiction.

\item
Assume now that
$|{\mathsf{\Phi}}_{1}|
 <
 |{\mathsf{\Phi}}_{2}|$.
So,
$|{\mathsf{\Phi}}_{2}|
 \geq
 3$
and
$|{\mathsf{\Phi}}_{1}|
 \leq
 1$.
Then,
\begin{eqnarray*}
         {\mathsf{E}}_{\pi}\left( p_{\pi}^{2},
                                 {\bf p}_{-\pi}
                           \right)
& \geq & 3\, M + 3 - 4\, y_{2}\ \
         >\ \
         M - 4\, y_{1}\ \
         \geq\ \
         {\mathsf{E}}_{\pi}\left( p_{\pi}^{1},
                                 {\bf p}_{-\pi}
                           \right)\, ,
\end{eqnarray*}
since
$3 M + 3 > M+4 \geq M + 4\, (y_{2} - y_{1})$.
A contradiction.

\item
Assume finally that
$|{\mathsf{\Phi}}_{1}|
 =
 |{\mathsf{\Phi}}_{2}|$.
Then,
\begin{eqnarray*}
         {\mathsf{E}}_{\pi}\left( p_{\pi}^{1},
                                 {\bf p}_{-\pi}
                           \right)
& = & 2\, M - 4\, y_{1}\, ,
\end{eqnarray*}
and
\begin{eqnarray*}
         {\mathsf{E}}_{\pi}\left( p_{\pi}^{2},
                                 {\bf p}_{-\pi}
                           \right)
& = & 2\, M + 2 - 4\, y_{2}\, .
\end{eqnarray*}
Since
${\mathsf{E}}_{\pi}\left( p_{\pi}^{2},
                         {\bf p}_{-\pi}
                   \right)
 =
 {\mathsf{E}}_{\pi}\left( p_{\pi}^{2},
                         {\bf p}_{-\pi}
                  \right)$,
it follows that
$4 (y_{2} - y_{1})
 =
 2$.
A contradiction.

\end{enumerate}
The claim follows.
\end{proof}

\noindent
We are now ready to prove:

\begin{lemma}
\label{one direction}
If ${\mathsf{G}}$ has a
${\mathsf{V}}$-equilibrium,
then $\langle n, m, {\mathsf{A}}
      \rangle$
has a solution.
\end{lemma}

\begin{proof}
Consider a ${\mathsf{V}}$-equilibrium
${\bf p}$.
By Lemmas~\ref{hilton fontana}
and~\ref{hilton percola},
all players $[k,4]$
with $k \in [m]$
are non-pure.
Hence,
the {\it Mixed-Player-Has-Pure-Neighbors} property
(Proposition~\ref{the real almost pure property})
and the definition of the weights
imply together that all players
$i \in [n]$
and
$[k, j]$
with $k \in [m]$
and $j \in \{ 0, \ldots, 3 \}$
are pure.
Fix now an arbitrarily chosen player
${\mathsf{\pi}} = [\widehat{k}, 4]$.
For each link $\ell \in [2]$,
denote
${\mathsf{I}}_{\ell}
 :=
 \{ i \in [n]
    \mid
    p_{i}(\ell) = 1
 \}$
and
${\mathsf{\Phi}}_{\ell}
 :=
 \{ [\widehat{k}, j]
    \mid
    j \in \{ 0, \ldots, 3 \}\
    \mbox{and}\
    p_{[\widehat{k}, j]}(\ell) = 1
 \}$.
Clearly,
\begin{eqnarray*}
      {\mathsf{E}}_{{\mathsf{\pi}}}\left( p_{{\mathsf{\pi}}}^{1},
                                        {\bf p}_{-{\mathsf{\pi}}}
                                 \right)
& = & (1 + \left|{\mathsf{\Phi}}_{1}
           \right|)
      \cdot
      M
      +
      \sum_{i \in {\mathsf{I}}_{1}}
        a_{i}\, ,
\end{eqnarray*}
and
\begin{eqnarray*}
      {\mathsf{E}}_{{\mathsf{\pi}}}\left( p_{{\mathsf{\pi}}}^{2},
                                        {\bf p}_{-{\mathsf{\pi}}}
                                 \right)
& = & (1 + \left| {\mathsf{\Phi}}_{2}
           \right|)
      \cdot
      (M+1)
      +
      2\, \sum_{i \in {\mathsf{I}}_{2}}
            a_{i}\, .
\end{eqnarray*}
The {\it Weak-Equilibrium-for-Expectation} property
(Corollary~\ref{brand new}
and Proposition~\ref{concavity implies weak equilibrium})
implies that
${\mathsf{E}}_{{\mathsf{\pi}}}\left( p_{{\mathsf{\pi}}}^{1},
                                   {\bf p}_{-{\mathsf{\pi}}}
                            \right)
 =
 {\mathsf{E}}_{{\mathsf{\pi}}}\left( p_{{\mathsf{\pi}}}^{2},
                                   {\bf p}_{-{\mathsf{\pi}}}
                            \right)$.
By the choice of $M$,
this implies that
$|{\mathsf{\Phi}}_{1}|
 =
 |{\mathsf{\Phi}}_{2}|
 =
 2$,
so that
{
\small
\begin{eqnarray*}
      \sum_{i \in {\mathsf{I}}_{1}}
        a_{i}
& = & 3 +
      2\, \sum_{i \in {\mathsf{I}}_{2}}
            a_{i}\, .
\end{eqnarray*}
}
Hence,
${\mathsf{I}}$
is a solution of $\langle n, m, {\mathsf{A}} \rangle$.
\end{proof}

\noindent
We finally prove:

\begin{lemma}
\label{the other direction}
If $\langle n, m, {\mathsf{A}}
      \rangle$
has a solution,
then ${\mathsf{G}}$ has a
${\mathsf{V}}$-equilibrium.
\end{lemma}

\noindent
For the proof of Lemma~\ref{the other direction},
we shall use
a particular monotonicity property 
of the risk valuation
${\mathsf{R}}
 =
 {\mathsf{kM}}$
with an even integer $k \geq 2$.
Note that
{
\small
\begin{eqnarray*}
      \widehat{{\mathsf{kM}}}(a,b,q)
& = & (1-q)
      \cdot
      (a - ((1-q) \cdot a + q \cdot b))^{k}
      +
      q
      \cdot
      (b - ((1-q) \cdot a + q \cdot b))^{k}                                        \\
& = & (1-q)
      \cdot
      q^{k}
      \cdot
      (b-a)^{k}
      +
      q
      \cdot
      (1-q)^{k}
      \cdot
      (b-a)^{k}                                                                   \\
& = & q
      \cdot
      (1-q)
      \cdot
      (b-a)^{k}
      \cdot
      \left( q^{k-1} + (1-q)^{k-1}
      \right)\, .
\end{eqnarray*}
}
This implies that
all valuations ${\mathsf{R}}$
addressed in Theorem~\ref{concrete two ordered links}
incur $\widehat{{\mathsf{R}}}(a,b,q)$
which increases monotonically
in $b-a$
for a fixed probability $q \in (0, 1)$.
In fact,
the function
$\widehat{{\mathsf{R}}}(a, b, q)$
with $a \leq b$
is a function
$\widetilde{{\mathsf{R}}}(b-a, q)
 :=
 \widehat{{\mathsf{R}}}(a, b, q)$
in $b-a$;
the function
$\widetilde{{\mathsf{R}}}(b-a, q)$
increases monotonically in $b-a$
for a fixed $q \in (0, 1)$
and satisfies
$\widetilde{{\mathsf{R}}}(b-a, q)
 =
 \widetilde{{\mathsf{R}}}(b-a, 1-q)$.
For the proof of Lemma~\ref{the other direction},
we shall refer to these
two properties together
as the 
{\it Two-Values Risk-Monotonicity} property.
We continue with the proof of
Lemma~\ref{the other direction}.

\begin{proof}
Consider a solution
${\mathsf{I}}
 \subset
 [n]$
of $\langle n, m, {\mathsf{A}}
 \rangle$.
Define a mixed profile
${\bf p}$
as follows:
\begin{itemize}

\item
For each $i \in [n]$,
$p_{i}(1)
 := 1$
if $i \in {\mathsf{I}}$,
and $0$ otherwise.

\item
For $1 \leq k \leq m$:
$p_{[k, 0]}(1), p_{[k, 2]}(1)
 := 1$;
$p_{[k, 1]}(1), p_{[k,3]}(1)
 := 0$;
$p_{[k, 4]}(2)
 := x \in (0, 1)$.

\end{itemize}
We now prove that
$x$ can be chosen so that
${\bf p}$ is a ${\mathsf{V}}$-equilibrium.
We proceed by case analysis.
\begin{enumerate}

\item
Players $i \in [n]$
cannot improve since
$\omega (i, {\mathsf{\pi}}, \ell)
 =
 0$
for all ${\mathsf{\pi}}
         \in
         {\mathsf{\Pi}}$
and $\ell \in [2]$.

\item
Consider now players
$[k, r]$
with $k \in [m]$
and $r \in \{ 0, 1, 2, 3 \}$.
Note that for fixed $k$ and $r$,
${\mathsf{\mu}}_{[k,r]}
         ({\bf s})$
takes only two values over all profiles
${\bf s} \in {\cal S}$
with
${\bf p}({\bf s}) > 0$.
(Observe for this that
player ${\mathsf{\pi}} = [k, 4]$
is the only player
with $0 < p_{{\mathsf{\pi}}}(1) < 1$
and
$\omega \left( [k,r], {\mathsf{\pi}}, \ell
        \right)
 \neq
 0$
for $r \in \{ 0, 1, 2, 3 \}$
and
$\ell \in [2]$.)
Denote as
${\mathsf{A}}_{kr}$
and
${\mathsf{B}}_{kr}$
the two values taken by
${\mathsf{\mu}}_{[k,r]}
         ({\bf s})$
over all profiles
${\bf s} \in {\cal S}$
with
${\bf p}({\bf s}) > 0$,
with
${\mathsf{A}}_{kr}
 >
 {\mathsf{B}}_{kr}$.
By the cost functions
and the definition of ${\bf p}$,
we get:
\begin{itemize}

\item
\underline{For $r \in \{ 0, 2 \}$:}
${\mathsf{A}}_{kr}
 =
 3M$
and
${\mathsf{B}}_{kr}
 =
 2M$,
so that
{
\small
\begin{eqnarray*}
      {\mathsf{E}}_{[k,r]}({\bf p})
& = & (1-x)
      \cdot
      {\mathsf{A}}_{kr}
      +
      x
      \cdot
      {\mathsf{B}}_{kr}                  \\
& = & (3-x)
      \cdot
      M\, ,
\end{eqnarray*}
}
and,
using 
the {\it Two-Values Risk-Monotonicity} property,
{
\small
\begin{eqnarray*}
      {\mathsf{R}}_{[k,r]}({\bf p})
& = & \widehat{{\mathsf{R}}}_{[kr]}\left( {\mathsf{A}}_{kr},
                                         {\mathsf{B}}_{kr},
                                         1-x
                                  \right)                                                     \\
& = & \widehat{{\mathsf{R}}}_{[kr]}\left( 3M,
                                         2M,
                                         1-x
                                  \right)                                                     \\
& = & \widetilde{{\mathsf{R}}}_{[kr]}\left( M,
                                           x
                                  \right)\, .
\end{eqnarray*}
}

\item
\underline{For $r \in \{ 1, 3 \}$:}
${\mathsf{A}}_{kr}
 =
 3\, (M+1)$
and
${\mathsf{B}}_{kr}
 =
 2\, (M+1)$,
so that
{
\small
\begin{eqnarray*}
      {\mathsf{E}}_{[k,r]}({\bf p})
& = & x
      \cdot
      {\mathsf{A}}_{kr}
      +
      (1-x)
      \cdot
      {\mathsf{B}}_{kr}                  \\
& = & (2+x)
      \cdot
      (M+1)\, ,
\end{eqnarray*}
}
and,
using the {\it Two-Values Risk-Monotonicity} property,
{
\small
\begin{eqnarray*}
      {\mathsf{R}}_{[k,r]}({\bf p})
& = & \widehat{{\mathsf{R}}}_{[kr]}\left( {\mathsf{A}}_{kr},
                                         {\mathsf{B}}_{kr},
                                         x
                                  \right)                                              \\
& = & \widehat{{\mathsf{R}}}_{[kr]}\left( 3\, (M+1),
                                         2\, (M+1),
                                         x
                                  \right)                                              \\
& = & \widetilde{{\mathsf{R}}}_{[kr]}\left( M+1,
                                         x
                                  \right)\, .
\end{eqnarray*}
}

\end{itemize}
Consider now the cost of player $[k, r]$
with $r \in \{ 0, 1, 2, 3 \}$
when she switches to the other link.
Denote as
$\widehat{{\bf p}}$
the corresponding mixed profile.
As in the case of the mixed profile ${\bf p}$,
${\mathsf{\mu}}_{[k,r]}({\bf s})$
takes only two values
over all profiles
${\bf s} \in {\cal S}$
with
$\widehat{{\bf p}}({\bf s})
 >
 0$.
Denote as
${\mathsf{C}}_{kr}$
and
${\mathsf{D}}_{kr}$
the two values taken by
${\mathsf{\mu}}_{[k,r]}
         ({\bf s})$
over all profiles
${\bf s} \in {\cal S}$
with
$\widehat{{\bf p}}({\bf s}) > 0$,
with
${\mathsf{C}}_{kr}
 >
 {\mathsf{D}}_{kr}$.
By the cost functions
and the definition of $\widehat{{\bf p}}$,
we get:
\begin{itemize}

\item
\underline{For $r \in \{ 0, 2 \}$:}
${\mathsf{C}}_{kr}
 =
 4\, (M+1) - 4 = 4M$
and
${\mathsf{D}}_{kr}
 =
 3\, (M+1) - 4 = 3M-1$,
so that
{
\small
\begin{eqnarray*}
      {\mathsf{E}}_{[k,r]}(\widehat{{\bf p}})
& = & x
      \cdot
      {\mathsf{C}}_{kr}
      +
      (1-x)
      \cdot
      {\mathsf{D}}_{kr}                  \\
& = & 3M-1
      +
      x
      \cdot
      (M+1)\, ,
\end{eqnarray*}
}
and,
using the {\it Two-Values Risk-Monotonicity} property,
{
\small
\begin{eqnarray*}
      {\mathsf{R}}_{[k,r]}(\widehat{{\bf p}})
& = & \widehat{{\mathsf{R}}}_{[kr]}\left( {\mathsf{C}}_{kr},
                                         {\mathsf{D}}_{kr},
                                         1-x
                                  \right)                                                     \\
& = & \widehat{{\mathsf{R}}}_{[kr]}\left( 4M,
                                         3M-1,
                                         1-x
                                  \right)                                                     \\
& = & \widetilde{{\mathsf{R}}}_{[kr]}\left( M+1,
                                           x
                                  \right)\, .
\end{eqnarray*}
}

\item
\underline{For $r \in \{ 1, 3 \}$:}
${\mathsf{C}}_{kr}
 =
 4M-4$
and
${\mathsf{D}}_{kr}
 =
 3M-4$,
so that
{
\small
\begin{eqnarray*}
      {\mathsf{E}}_{[k,r]}(\widehat{{\bf p}})
& = & x
      \cdot
      {\mathsf{D}}_{kr}
      +
      (1-x)
      \cdot
      {\mathsf{C}}_{kr}                  \\
& = & 4M-4
      -
      x
      \cdot
      M\, ,
\end{eqnarray*}
}
and,
using the {\it Two-Values Risk-Monotonicity} property,
{
\small
\begin{eqnarray*}
      {\mathsf{R}}_{[k,r]}(\widehat{{\bf p}})
& = & \widehat{{\mathsf{R}}}_{[kr]}\left( {\mathsf{C}}_{kr},
                                         {\mathsf{D}}_{kr},
                                         x
                                  \right)                                              \\
& = & \widehat{{\mathsf{R}}}_{[kr]}\left( 4M - 4,
                                         3M - 4,
                                         x
                                  \right)                                              \\
& = & \widetilde{{\mathsf{R}}}_{[kr]}\left( M,
                                           x
                                    \right)\, .
\end{eqnarray*}
}

\end{itemize}
We now determine a probability $x \in (0, 1)$
so that for all players
$[k, r]$
with $r \in \{ 0, 1, 2, 3 \}$,
${\mathsf{V}}_{[k,r]}({\bf p})
 \leq
 {\mathsf{V}}_{[k,r]}(\widehat{{\bf p}})$.
We proceed by case analysis.
\begin{itemize}

\item
\underline{For $r \in \{ 0, 2 \}$:}
Then,
${\mathsf{V}}_{[k,r]}({\bf p})
 \leq
 {\mathsf{V}}_{[k,r]}(\widehat{{\bf p}})$
if and only if
{
\small
\begin{eqnarray*}
         (3-x) \cdot M
         +
         \widetilde{{\mathsf{R}}}_{[k,r]}(M, x)
& \leq & 3M-1
         +
         x
         \cdot
         (M+1)
         +
         \widetilde{{\mathsf{R}}}_{[k,r]}(M+1, x)
\end{eqnarray*}
}
if and only if
{
\small
\begin{eqnarray*}
         1
& \leq & x
         \cdot
         (2M+1)
         +
         \widetilde{{\mathsf{R}}}_{[k,r]}(M+1, x)
         -
         \widetilde{{\mathsf{R}}}_{[k,r]}(M, x)\, .
\end{eqnarray*}
}

\item
\underline{For $r \in \{ 1, 3 \}$:}
Then,
${\mathsf{V}}_{[k,r]}({\bf p})
 \leq
 {\mathsf{V}}_{[k,r]}(\widehat{{\bf p}})$
if and only if
{
\small
\begin{eqnarray*}
         (2+x) \cdot (M+1)
         +
         \widetilde{{\mathsf{R}}}_{[k,r]}(M+1, x)
& \leq & 4M-4
         -
         x
         \cdot
         M
         +
         \widetilde{{\mathsf{R}}}_{[k,r]}(M, x)
\end{eqnarray*}
}
if and only if
{
\small
\begin{eqnarray*}
         x
         \cdot
         (2M+1)
         +
         \widetilde{{\mathsf{R}}}_{[k,r]}(M+1, x)
         -
         \widetilde{{\mathsf{R}}}_{[k,r]}(M, x)
& \leq & 2M-6\, .
\end{eqnarray*}
}

\end{itemize}
Set
{
\small
\begin{eqnarray*}
         {\mathsf{h}}(x)
 & := &  x
         \cdot
         (2M+1)
         +
         \widetilde{{\mathsf{R}}}_{[k,r]}(M+1, x)
         -
         \widetilde{{\mathsf{R}}}_{[k,r]}(M, x)\, .
\end{eqnarray*}
}
We shall determine a probability $x$
so that
$1 \leq {\mathsf{h}}(x) \leq 2M-6$.
(Since $M \geq 4$,
$2M-6 \geq 1$.)
Observe that ${\mathsf{h}}(0) = 0$.
Set
{
\small
\begin{eqnarray*}
       \widehat{x}
& := & \frac{\textstyle 1}
            {\textstyle 2M+1}\, .
\end{eqnarray*}
}
By the {\it Two-Values Risk-Monotonicity} property,
$\widetilde{{\mathsf{R}}}_{[k,r]}(M+1, x)
 >
 \widetilde{{\mathsf{R}}}_{[k,r]}(M, x)$.
This implies that
$ {\mathsf{h}}(\widehat{x}) > 1$.
If
${\mathsf{h}}(\widehat{x})
 \leq
 2M-6$,
then set
$x := \widehat{x}$
and we are done.
If
${\mathsf{h}}(\widehat{x})
 >
 2M-6$,
then
the continuity of ${\mathsf{R}}$
implies that
there is a probability
$\widetilde{x}
 \in
 (0, \widehat{x})$
such that
$1 \leq {\mathsf{h}}(\widetilde{x}) \leq 2M-6$,
and we are done.

\item
Finally consider
players $[k, 4]$
with $k \in [m]$.
By the definition of ${\bf p}$,
{
\small
\begin{eqnarray*}
    {\mathsf{E}}_{[k,4]}\left( p_{[k, 4]}^{1},
                              {\bf p}_{-[k,4]}
                       \right)
& = 3\, M + \sum_{i \in {\mathsf{I}}}
              a_{ki}\, ,
\end{eqnarray*}
}
and
{
\small
\begin{eqnarray*}
    {\mathsf{E}}_{[k,4]}\left( p_{[k, 4]}^{2},
                              {\bf p}_{-[k,4]}
                       \right)
& = 3\, (M+1) + 2
                \sum_{i \not\in {\mathsf{I}}}
                  a_{ki}\, .
\end{eqnarray*}
}
Since ${\mathsf{I}}$ is a solution of
$\langle n, m, {\mathsf{A}}
 \rangle$,
we get that
{
\small
\begin{eqnarray*}
{\mathsf{E}}_{[k,4]}\left( p_{[k, 4]}^{1},
                           {\bf p}_{-[k,4]}
                    \right)
& = &
 {\mathsf{E}}_{[k,4]}\left( p_{[k, 4]}^{1},
                           {\bf p}_{-[k,4]}
                    \right)\, .
\end{eqnarray*}
}
By the {\it Risk-Positivity} property,
this implies that
${\mathsf{R}}_{[k,4]}({\bf p})
 =
 0$
so that
${\mathsf{V}}_{[k,4]}({\bf p})
 =
 {\mathsf{E}}_{[k,4]}({\bf p})$.
Since ${\bf p}$ is a fully mixed profile
(since $0 < x < 1$)
and ${\mathsf{V}}$ has the
{\it Weak-Equilibrium-for-Expectation} property
(Corollary~\ref{brand new}
and Proposition~\ref{concavity implies weak equilibrium}),
${\mathsf{E}}_{[k,4]}({\bf p})$
cannot decrease;
hence,
neither ${\mathsf{V}}_{[k,4]}({\bf p})$ can.

\end{enumerate}
It follows from the case analysis that
${\bf p}$ is a ${\mathsf{V}}$-equilibrium.
\end{proof}

\noindent
Lemmas~\ref{one direction}
and~\ref{the other direction}
establish the reduction
for the ${\mathcal{NP}}$-hardness;
since the numbers involved
in the reduction
are polynomially bounded,
strong ${\mathcal{NP}}$-hardness follows.
\end{proof}

\section{Two Players}
\label{two players}

\noindent
Consider a concave valuation ${\mathsf{V}}^{{\mathsf{\nu}}}$,
for an increasing and strictly convex function ${\mathsf{\nu}}$.
Since ${\mathsf{\nu}}^{-1}$ is also increasing,
a mixed profile ${\bf p}$
is a ${\mathsf{V}}^{{\mathsf{\nu}}}$-equilibrium
for a game ${\mathsf{G}}$
if and only if
${\bf p}$
is an ${\mathsf{E}}$-equilibrium
for the game ${\mathsf{G}}^{{\mathsf{\nu}}}$
constructed from ${\mathsf{G}}$
by setting
for each player $i \in [n]$
and profile
${\bf s} \in {\cal S}$,
${\mathsf{\mu}}_{i}^{{\mathsf{\nu}}}({\bf s})
 :=
 {\mathsf{\nu}}({\mathsf{\mu}}_{i}({\bf s}))$.
Since every game has an ${\mathsf{E}}$-equilibrium~\cite{N50,N51},
this implies that there is a
${\mathsf{V}}^{{\mathsf{\nu}}}$-equilibrium 
for ${\mathsf{G}}$,
and the associated search problem
for a ${\mathsf{V}}^{{\mathsf{\nu}}}$-equilibrium
is total;
it is in ${\cal PPAD}$~\cite{P94}
for 2-players games.
Nevertheless,
we shall show that there are other (concave) valuations ${\mathsf{V}}$
for which deciding the existence
of a ${\mathsf{V}}$-equilibrium
is strongly ${\mathcal{NP}}$-hard
for 2-players games.

\subsection{General ${\mathcal{NP}}$-Hardness Result}
\label{general np hardness}

\noindent
We show:

\begin{theorem}
\label{two players complexity}
Fix an
$({\mathsf{E}}
  +
  {\mathsf{R}})$-valuation
${\mathsf{V}}$
such that:
\begin{enumerate}


\item[{\sf (1)}]
${\mathsf{V}}$
has the {\it Weak-Equilibrium-for-Expectation}
property.

\item[{\sf (2)}]
There is a polynomial time computable
${\mathsf{\delta}}$
with
$0
 <
 {\mathsf{\delta}}
 \leq
 \frac{\textstyle 1}
        {\textstyle 4}$
such that:
\begin{enumerate}

\item[{\sf (2/a)}]
$\widehat{{\mathsf{R}}}(1, 1 + 2 {\mathsf{\delta}}, q)
 <
 \frac{\textstyle 1}
      {\textstyle 2}$
for each probability
$q \in [0, 1]$.

\item[{\sf (2/b)}]
$\widehat{{\mathsf{V}}}(1, 1 + 2 {\mathsf{\delta}}, r)
 <
 \widehat{{\mathsf{V}}}(1, 2, q)$
for all
$0 \leq r \leq q \leq 1$.

\item[{\sf (2/c)}]
The Crawford game
${\mathsf{G}}_{C}({\mathsf{\delta}})$
with bimatrix
$\left( \begin{array}{ll}
          \langle 1 + {\mathsf{\delta}},
                  1 + {\mathsf{\delta}}
          \rangle                            & \langle 1, 1 + 2{\mathsf{\delta}}
                                               \rangle                                        \\
          \langle 1, 1 + 2{\mathsf{\delta}}
          \rangle                            & \langle 1 + 2{\mathsf{\delta}},
                                                       1
                                               \rangle
        \end{array}
 \right)$
has no ${\mathsf{V}}$-equilibrium.

\end{enumerate}

\end{enumerate}
Then,
{\sf $\exists {\mathsf{V}}$-EQUILIBRIUM}
is strongly ${\mathcal{NP}}$-hard
for 2-players games.
\end{theorem}

\noindent
We present a general proof
with a reduction
involving the parameter 
${\mathsf{\delta}}$
from Condition {\sf (2)},
required to be polynomial time computable.
The reduction uses 
the {\it Crawford game}
${\mathsf{G}}_{C}(\delta)$
as a ``gadget'';
for any ${\mathsf{\delta}}$
with $0 < {\mathsf{\delta}} < 1$,
${\mathsf{G}}_{C}(\delta)$
is an adapted generalization
of a bimatrix game
from~\cite[Section 4]{C90}.
The parameter ${\mathsf{\delta}}$
enters the reduction
through ${\mathsf{G}}_{C}(\delta)$.

\noindent
Specifically,
the proof of Theorem~\ref{two players complexity}
employs a reduction from ${\sf SAT}$~\cite[L01]{GJ79}.
An instance of {\sf SAT}
is a propositional formula ${\mathsf{\phi}}$
in the form of
a conjunction of {\it clauses}
${\mathsf{C}} = \{ {\mathsf{c}}_{1},
                   \ldots,
                   {\mathsf{c}}_{k}
                \}$
over a set of {\it variables}
${\cal V} = \{ {\mathsf{v}}_{1},
               \ldots,
               {\mathsf{v}}_{m}
            \}$.
Denote as
${\mathsf{L}}
 =
 \{ \ell_{1}, \overline{\ell}_{1},
    \ldots,
    \ell_{m}, \overline{\ell}_{m}
 \}$
the set of {\it literals}
corresponding to the variables in ${\cal V}$.
We shall use lower-case letters
${\mathsf{c}},
 {\mathsf{c}}_{1},
 {\mathsf{c}}_{2}, \ldots$,
${\mathsf{v}},
 {\mathsf{v}}_{1},
 {\mathsf{v}}_{2}, \ldots$,
and
$\ell, \ell_{1}, \ell_{2}, \ldots$
to denote
clauses from ${\mathsf{C}}$,
variables from ${\cal V}$
and
 literals from ${\mathsf{L}}$,
respectively.
Denote
${\mathsf{\Lambda}}
 := {\mathsf{C}}
    \cup
    {\cal V}
    \cup
    {\mathsf{L}}$.
We shall use
the {\it Crawford set}
${\cal F}
 =
 \{ {\mathsf{f}}_{1}, {\mathsf{f}}_{2}
 \}$
with two strategies
${\mathsf{f}}_{1}$
and
${\mathsf{f}}_{2}$;
${\mathsf{f}}$ denotes
either
${\mathsf{f}}_{1}$ or ${\mathsf{f}}_{2}$.
The cost values
are chosen judiciously
so as to carefully
assure or exclude the existence
of a ${\mathsf{V}}$-equilibrium.
We continue with the proof
of Theorem~\ref{two players complexity}.

\begin{proof}
\noindent
Given an instance
${\mathsf{\phi}}$
of {\sf SAT},
construct a game
${\mathsf{G}}
 =
 {\mathsf{G}}({\mathsf{\phi}})
 =
 \langle [2],
         \{ S_{i}
         \}_{i \in [2]},
         \{ {\mathsf{\mu}}_{i}
         \}_{i \in [2]}
 \rangle$
as follows.
For each player
$i \in [2]$,
$S_{i}
 :=
 {\mathsf{\Lambda}}
 \cup
 {\cal F}$.
The cost functions
$\left\{ {\mathsf{\mu}}_{i}
 \right\}_{i \in [2]}$
are given
in Figure~\ref{ridiculous}.

\begin{figure}[ht]
\begin{center}
\begin{tabular}{||l|l|l||}
\hline
\hline
\hline
{\sf Profile} ${\bf s}$& {\sf Condition} on ${\bf s}$      & $\left\langle {\mathsf{\mu}}_{1}({\bf s}),
                                                                                                                {\mathsf{\mu}}_{2}({\bf s})
                                                                                             \right\rangle$                                                                  \\
\hline
\hline
$\langle \ell_{1},
         \ell_{2}
 \rangle$              & $\ell_{1}
                          \neq
                          \overline{\ell}_{2}$             & $\langle 1, 1
                                                              \rangle$                                                         \\
\hline
$\langle \ell_{1},
         \ell_{2}
 \rangle$              & $\ell_{1}
                          =
                          \overline{\ell}_{2}$             & $\langle 2,
                                                                      2
                                                              \rangle$                                                        \\
\hline
\hline
$\langle \ell,
         {\mathsf{v}}
 \rangle$              & $\ell$ is a literal for ${\mathsf{v}}$ & $\langle 2,
                                                                           m
                                                                   \rangle$                                                        \\
\hline
$\langle \ell,
         {\mathsf{v}}
 \rangle$             & $\ell$ is not a literal for ${\mathsf{v}}$ & $\langle 2,
                                                                              0
                                                                      \rangle$                                                         \\
\hline
$\langle \ell,
         {\mathsf{c}}
 \rangle$             & $\ell \not\in {\mathsf{c}}$        & $\langle 2, 0
                                                              \rangle$                                                         \\
\hline
$\langle \ell,
         {\mathsf{c}}
 \rangle$             & $\ell \in {\mathsf{c}}$            & $\langle 2, m
                                                              \rangle$                                                         \\
\hline
$\langle \ell, {\mathsf{f}}
 \rangle$             & ---                                & $\langle 2, 1
                                                              \rangle$                                                         \\
\hline
\hline
$\langle {\mathsf{v}}_{1},
         {\mathsf{v}}_{2}
 \rangle$ or
$\langle {\mathsf{c}}_{1},
         {\mathsf{c}}_{2}
 \rangle$ or
$\langle {\mathsf{v}},
         {\mathsf{c}}
 \rangle$

                      & ---                                & $\langle 2, 2
                                                              \rangle$                                                         \\
\hline
$\langle {\mathsf{v}},
         {\mathsf{f}}
 \rangle$ or
$\langle {\mathsf{c}},
         {\mathsf{f}}
 \rangle$             & ---                                & $\langle 2, 1
                                                              \rangle$                                                         \\
\hline
\hline
$\langle {\mathsf{f}}_{1},
         {\mathsf{f}}_{1}
 \rangle$             & ---                                & $\langle 1 + {\mathsf{\delta}},
                                                                      1 + {\mathsf{\delta}}
                                                              \rangle$                                                          \\
\hline
$\langle {\mathsf{f}}_{1},
         {\mathsf{f}}_{2}
 \rangle$ or
$\langle {\mathsf{f}}_{2},
         {\mathsf{f}}_{1}
 \rangle$             & ---                                & $\langle 1, 1 + 2 {\mathsf{\delta}}
                                                              \rangle$                                                          \\
\hline
$\langle {\mathsf{f}}_{2},
         {\mathsf{f}}_{2}
 \rangle$             & ---                                & $\langle 1 + 2 {\mathsf{\delta}},
                                                                      1

                                                              \rangle$                                                          \\
\hline
\hline
\hline
\end{tabular}
\caption{The cost functions for the game ${\mathsf{G}}({\mathsf{\phi}})$.
For a profile
$\langle s_{1}, s_{2}
 \rangle$
not in the table,
set
${\mathsf{\mu}}_{i}(s_{1}, s_{2})
 :=
 {\mathsf{\mu}}_{\overline{i}}(s_{2}, s_{1})$,
with $i \in [2]$
and 
$\overline{i}
     \neq
     i$;
so,
$\overline{i}$
is the player other than $i$.}
\label{ridiculous}
\end{center}
\end{figure}

\noindent
For a player $i \in [2]$,
denote
$p_{i}({\cal F})
 :=  \sum_{{\mathsf{f}} \in {\cal F}}
         p_{i}({\mathsf{f}})$,
$p_{i}({\mathsf{L}})
 :=  \sum_{\ell \in {\mathsf{L}}}
         p_{i}(\ell)$
and
$p_{i}({\mathsf{\Lambda}})
 :=  \sum_{{\mathsf{\lambda}} \in {\mathsf{\Lambda}}}
         p_{i}({\mathsf{\lambda}})$;
note that
$p_{i}({\cal F})
 +
 p_{i}({\mathsf{\Lambda}})
 =
 1$.
We prove a sequence of technical claims:

\begin{lemma}
\label{both players play the formula}
In a ${\mathsf{V}}$-equilibrium
$\langle p_{1},
         p_{2}
 \rangle$
for ${\mathsf{G}}$,
$p_{1} ({\mathsf{\Lambda}})
 \cdot
 p_{2} ({\mathsf{\Lambda}})
 >
 0$.
\end{lemma}

\begin{proof}
Assume,
by way of contradiction,
that
$p_{i}({\mathsf{\Lambda}})
 =
 0$
for some player
$i \in [2]$;
so,
$p_{i}({\cal F})
 =
 1$.
For easier notation,
fix
$i := 1$.
(The proof is the same
for $i := 2$.)
By Condition {\sf (2/c)},
the Crawford game ${\mathsf{G}}_{C}({\mathsf{\delta}})$
has no ${\mathsf{V}}$-equilibrium.
Hence,
it follows that
$p_{2}({\mathsf{\Lambda}})
 >
 0$.
We proceed by case analysis
on
$p_{2}({\cal F})$.
\begin{enumerate}

\item
Assume first that
$p_{2}({\cal F})
 >
 0$;
so,
$p_{2}({\mathsf{\Lambda}})
 <
 1$.
Fix a strategy
${\mathsf{f}} \in {\cal F}$
with
$p_{2}({\mathsf{f}}) > 0$.
By the cost functions,
${\mathsf{\mu}}_{2}
          ({\mathsf{f}}', {\mathsf{f}})
 \leq
 1 + 2\, {\mathsf{\delta}}
 <
 2$
for each strategy
${\mathsf{f}}' \in {\cal F}$.
Hence,
${\mathsf{E}}_{2}
          (p_{1}, p_{2}^{{\mathsf{f}}})
 <
 2$
since
${\mathsf{\delta}}
 <
 \frac{\textstyle 1}
      {\textstyle 2}$. 
Fix now a strategy
${\mathsf{\lambda}}
 \in
 {\mathsf{\Lambda}}$
with
$p_{2}({\mathsf{\lambda}})
 >
  0$.
By the cost functions,
${\mathsf{\mu}}_{2}
          ({\mathsf{f}}_{1}, {\mathsf{\lambda}})
 =
 2$
for each strategy
${\mathsf{f}}_{1} \in {\cal F}$.
Hence,
${\mathsf{E}}_{2}
          (p_{1}, p_{2}^{{\mathsf{\lambda}}})
 =
 2$.
By the
{\it Weak-Equilibrium-for-Expectation} property
for player $2$,
${\mathsf{E}}_{2}
          (p_{1}, p_{2}^{{\mathsf{f}}})
 =
 {\mathsf{E}}_{2}
          (p_{1}, p_{2}^{{\mathsf{\lambda}}})$.
A contradiction.

\item
Assume now that
$p_{2}({\cal F})
 =
 0$;
so,
$p_{2}({\mathsf{\Lambda}})
 =
 1$.
By the cost functions,
${\mathsf{\mu}}_{2}\left( {\mathsf{f}}, {\mathsf{\lambda}}
                   \right)
 =
 2$
for each pair
of strategies
${\mathsf{f}} \in {\cal F}$
and
${\mathsf{\lambda}}
 \in
 {\mathsf{\Lambda}}$.
Hence,
${\mathsf{E}}_{2}\left( p_{1}, p_{2}
                 \right)
 =
 2$.
It follows
by the {\it Risk-Positivity} property
that
${\mathsf{R}}_{2}\left( p_{1}, p_{2}
                 \right)
 =
 0$,
so that
${\mathsf{V}}_{2}\left( p_{1}, p_{2}
                 \right)
 =
 2$.
Consider now player $2$
switching to
the strategy
${\mathsf{f}}_{2} \in {\cal F}$.
By Condition {\sf (2/a)},
{
\small
\begin{eqnarray*}
      {\mathsf{R}}_{2}\left( p_{1},
                             p_{2}^{{\mathsf{f}}_{2}}
                      \right)
& = & \widehat{{\mathsf{R}}}_{2}\left( 1,
                                       1 + 2\, {\mathsf{\delta}},
                                       p_{1}({\mathsf{f}}_{1})
                                \right)\
  \
  <\
  \
  \frac{\textstyle 1}
       {\textstyle 2}\, .
\end{eqnarray*}
}
Thus,
{
\small
\begin{eqnarray*}
      {\mathsf{E}}_{2}\left( p_{1}, p_{2}^{{\mathsf{f}}_{2}}
                      \right)
& = & \underbrace{{\mathsf{\mu}}_{2}({\mathsf{f}}_{2}, {\mathsf{f}}_{2})}_{= 1}
      \cdot\,
      p_{1}({\mathsf{f}}_{2})
      +
      \underbrace{{\mathsf{\mu}}_{2}({\mathsf{f}}_{1}, {\mathsf{f}}_{2})}_{= 1 + 2 {\mathsf{\delta}}}
      \cdot
      p_{1}({\mathsf{f}}_{1})                                                                  \\
& = & 1
      \cdot
      p_{1}({\mathsf{f}}_{2})
      +
      (1 + 2 {\mathsf{\delta}})
      \cdot\,
      p_{1}({\mathsf{f}}_{1})                                                                  \\
&\leq& 1 + 2 {\mathsf{\delta}}\, ,
\end{eqnarray*}
}
so that
{
\small
\begin{eqnarray*}
      \lefteqn{{\mathsf{V}}_{2}\left( p_{1},
                                      p_{2}^{{\mathsf{f}}_{2}}
                               \right)}                                                      \\
= & {\mathsf{E}}_{2}\left( p_{1},
                           p_{2}^{{\mathsf{f}}_{2}}
                    \right)
    +
    {\mathsf{R}}_{2}\left( p_{1},
                           p_{2}^{{\mathsf{f}}_{2}}
                      \right)                                                                
  &                                                                                          \\
< & 1 + 2 {\mathsf{\delta}}
    +
    \frac{\textstyle 1}
           {\textstyle 2}
    &                                                                                         \\
\leq& 2
    & \mbox{(since ${\mathsf{\delta}}
                    \leq
                    \frac{\textstyle 1}
                         {\textstyle 4}$)}\, .
\end{eqnarray*}
}
A contradiction to the assumption
that
$\langle p_{1}, p_{2}
 \rangle$
is a ${\mathsf{V}}$-equilibrium.
\end{enumerate}
The claim follows.
\end{proof}

\begin{lemma}
\label{if a player plays only the formula then the other also does}
In a ${\mathsf{V}}$-equilibrium
$\langle p_{1},
         p_{2}
 \rangle$
for ${\mathsf{G}}$,
if
$p_{i}({\mathsf{\Lambda}})
 =
 1$
then
 $p_{\overline{i}}({\mathsf{\Lambda}})
    =
    1$.
\end{lemma}

\begin{proof}
Set $i := 1$
so that $p_{1}({\mathsf{\Lambda}})
         =
         1$.
By Lemma~\ref{both players play the formula},
$p_{2}({\mathsf{\Lambda}})
 >
 0$.
If
$p_{2}({\mathsf{\Lambda}})
 =
 1$,
then we are done.
So assume
$p_{2}({\mathsf{\Lambda}})
 <
 1$.
This implies that
$p_{2}({\cal F})
 >
 0$.

By the cost functions,
for each strategy
${\mathsf{f}} \in {\cal F}$,
${\mathsf{\mu}}_{2}\left( {\mathsf{\lambda}},
                          {\mathsf{f}}
                  \right)
 =
 1$
for each strategy
${\mathsf{\lambda}}
 \in
 {\mathsf{\Lambda}}$.
Hence,
for each strategy
${\mathsf{f}} \in {\cal F}$
with
$p_{2}({\mathsf{f}})
 >
 0$,
{\it (i)}
${\mathsf{E}}_{2}\left( p_{1},
                        p_{2}^{f}
                 \right)
 =
 1$,
so that
the {\it Weak-Equilibrium-for-Expectation} property
for player $2$
implies that
${\mathsf{E}}_{2}\left( p_{1},
                        p_{2}^{{\mathsf{\lambda}}}
                 \right)
 =
 1$
for each strategy
${\mathsf{\lambda}}
 \in
 {\mathsf{\Lambda}}$
with
$p_{2}({\mathsf{\lambda}})
 >
 0$,
and
{\it (ii)}
${\mathsf{R}}_{2}(p_{1}, p_{2}^{{\mathsf{f}}})
 =
 0$,
by the {\it Risk-Positivity} property.
Hence,
${\mathsf{V}}_{2}\left( p_{1},
                        p_{2}^{{\mathsf{f}}}
                 \right)
 =
 1$.

Assume,
by way of contradiction,
that there is a strategy
${\mathsf{\lambda}}'
 \in
 {\mathsf{\Lambda}}$
with
$p_{2}({\mathsf{\lambda}}')
 >
 0$
such that
${\mathsf{\mu}}_{2}({\mathsf{\lambda}},
                    {\mathsf{\lambda}}^{\prime})
 \neq
 1$
for some
${\mathsf{\lambda}}
 \in
 {\mathsf{\Lambda}}$
with
$p_{1}({\mathsf{\lambda}})
 >
 0$.
It follows
by the {\it Risk-Positivity} property
that
${\mathsf{R}}_{2}(p_{1}, p_{2})
 >
 0$.
Then,
{
\small
\begin{eqnarray*}
      {\mathsf{V}}_{2}(p_{1}, p_{2})
& = & {\mathsf{E}}_{2}(p_{1}, p_{2})
      +
      {\mathsf{R}}_{2}(p_{1}, p_{2})                    \\
& = & \sum_{{\mathsf{\lambda}}
            \in
            {\mathsf{\Lambda}}
           }
        \underbrace{{\mathsf{E}}_{2}
                    (p_{1}, p_{2}^{{\mathsf{\lambda}}})}_{= 1}
        \cdot
        p_{2}({\mathsf{\lambda}})
      +
      \sum_{{\mathsf{f}} \in {\cal F}}
        \underbrace{{\mathsf{E}}_{2}
                    (p_{1}, p_{2}^{{\mathsf{f}}})}_{= 1}
        \cdot
        p_{2}({\mathsf{f}})
      +
      {\mathsf{R}}_{2}(p_{1}, p_{2})                    \\
& = & 1 + {\mathsf{R}}_{2}(p_{1}, p_{2})\, .
\end{eqnarray*}
}
Since
${\mathsf{V}}_{2}\left( p_{1},
                        p_{2}^{{\mathsf{f}}}
                 \right)
 =
 1$,
player $2$
improves her cost by
switching to
a strategy ${\mathsf{f}} \in {\cal F}$.
A contradiction to
the assumption that
$\langle p_{1}, p_{2}
 \rangle$
is a ${\mathsf{V}}$-equilibrium.

Hence,
it follows that
${\mathsf{\mu}}_{2}({\mathsf{\lambda}},
                    {\mathsf{\lambda}}')
 =
 1$
for each pair of strategies
${\mathsf{\lambda}},
 {\mathsf{\lambda}}'
 \in
 {\mathsf{\Lambda}}$
with
$p_{1}({\mathsf{\lambda}})
 >
 0$
and
$p_{2}({\mathsf{\lambda}}')
 >
 0$.
By the cost functions,
it follows that
for each pair of distinct strategies
${\mathsf{\lambda}},
 {\mathsf{\lambda}}'
 \in
 {\mathsf{\Lambda}}$
with
$p_{1}({\mathsf{\lambda}})
 >
 0$
and
$p_{2}({\mathsf{\lambda}}')
 >
 0$,
${\mathsf{\lambda}},
 {\mathsf{\lambda}}'
 \in
 {\mathsf{L}}$
with
${\mathsf{\lambda}}
 \neq
 \overline{{\mathsf{\lambda}}'}$,
${\mathsf{\mu}}_{1}({\mathsf{\lambda}},
                    {\mathsf{\lambda}}')
 =
 1$.

Consider now player $1$.
Note that
for each strategy
${\mathsf{\lambda}}
 \in
 {\mathsf{\Lambda}}$
with $p_{1}({\mathsf{\lambda}})
      >
      0$,
{\it (i)}
${\mathsf{\mu}}_{1}({\mathsf{\lambda}},
                    {\mathsf{\lambda}}'
                   )
 =
 1$
for each strategy
 ${\mathsf{\lambda}}'
 \in
 {\mathsf{\Lambda}}$
with $p_{2}({\mathsf{\lambda}}')
      >
      0$,
and
{\it (ii)}
${\mathsf{\mu}}_{1}({\mathsf{\lambda}},
                    {\mathsf{f}}
                   )
 =
 2$
for each strategy
 ${\mathsf{f}}
 \in
 {\cal F}$
with $p_{2}({\mathsf{f}})
      >
      0$.
Hence,
$\left\{ {\mathsf{\mu}}_{1}({\bf s})
         \mid
         {\bf p}({\bf s}) > 0
 \right\}
 =
 \{ 1, 2 \}$,
so that
${\mathsf{V}}_{1}(p_{1}, p_{2})
 =
 \widehat{{\mathsf{V}}}_{1}(1, 2, q)$,
where
$q = p_{2}({\cal{F}})$.

Consider player $1$
switching to the pure strategy
$p_{1}^{{\mathsf{f}}_{1}}$.
By the cost functions,
${\mathsf{\mu}}_{1}({\mathsf{f}}_{1},
                    {\mathsf{\lambda}})
 =
 1$
for each strategy
${\mathsf{\lambda}}
 \in
 {\mathsf{\Lambda}}$,
${\mathsf{\mu}}_{1}({\mathsf{f}}_{1},
                    {\mathsf{f}}_{2})
 =
 1$
and
${\mathsf{\mu}}_{1}({\mathsf{f}}_{1},
                    {\mathsf{f}}_{1})
 =
 1 + {\mathsf{\delta}}$.                     
Hence,
for
$\widehat{{\bf p}}
 =
 \langle p_{1}^{{\mathsf{f}}_{1}},
         p_{2}
 \rangle$,
$\{ {\mathsf{\mu}}_{1}({\bf s})
    \mid
    \widehat{{\bf p}}({\bf s})
    >
    0
 \}
 =
 \{ 1, 1 + {\mathsf{\delta}} \}$,
so that
${\mathsf{V}}_{1}\left( p_{1}^{{\mathsf{f}}_{1}},
                        p_{2}
                 \right)       
 =
 \widehat{{\mathsf{V}}}_{1}\left( 1, 1 + {\mathsf{\delta}}, r
                           \right)$,
with $r = p_{2}({\mathsf{f}}_{1})$.                           
So,
{
\small
\begin{eqnarray*}
      \lefteqn{\widehat{{\mathsf{V}}}_{1}\left( 1, 1 + {\mathsf{\delta}}, r
                                         \right)}                                      \\
<   & \widehat{{\mathsf{V}}}_{1}\left( 1, 2, q
                                         \right)
    & \mbox{(by Condition {\sf (2/b)}, 
             since $r \leq q$ and $1 + {\mathsf{\delta}} < 2$)}                        \\
=   & {\mathsf{V}}_{1}(p_{1}, p_{2})\, .
    &             
\end{eqnarray*}
}
So,
player $1$ improves her cost by switching to
the pure strategy $p_{1}^{{\mathsf{f}}_{1}}$.
A contradiction to the assumption that
$\langle p_{1},
         p_{2}
 \rangle$
is a ${\mathsf{V}}$-equilibrium.
The proof for $i := 2$ is identical
except that in the last stage
it uses the inequality
$\widehat{{\mathsf{V}}}_{2}\left( 1, 1 + 2\, {\mathsf{\delta}}, r
                           \right)
 <
 \widehat{{\mathsf{V}}}_{2}\left( 1, 2, q
                           \right)$,
holding by Condition {\sf (2/b)}.
\end{proof}

\begin{lemma}
\label{both players play only the formula}
In a ${\mathsf{V}}$-equilibrium
$\langle p_{1},
         p_{2}
 \rangle$
for ${\mathsf{G}}$,
$p_{1}({\mathsf{\Lambda}})
 =
 p_{2}({\mathsf{\Lambda}})
 =
 1$.
\end{lemma}

\begin{proof}
By Lemma~\ref{both players play the formula},
$p_{1}({\mathsf{\Lambda}})
 >
 0$
and
$p_{2}({\mathsf{\Lambda}})
 >
 0$.
If
$p_{i}({\mathsf{\Lambda}})
 =
 1$
for some player $i \in [2]$,
then,
by Lemma~\ref{if a player plays only the formula then the other also does},
$p_{\overline{i}}({\mathsf{\Lambda}})
 =
 1$,
and we are done.
So assume that
$p_{1}({\mathsf{\Lambda}})
 <
 1$
and
$p_{2}({\mathsf{\Lambda}})
 <
 1$;
this implies that
$p_{1}({\cal F})
 >
 0$
and
$p_{2}({\cal F})
 >
 0$.
By the cost functions,
for each pair
${\mathsf{\lambda}},
 {\mathsf{\lambda}}'
 \in
 {\mathsf{\Lambda}}$,
\begin{eqnarray*}
         {\mathsf{\mu}}_{1}({\mathsf{\lambda}},
                            {\mathsf{\lambda}}')
         +
         {\mathsf{\mu}}_{2}({\mathsf{\lambda}},
                            {\mathsf{\lambda}}')
& \geq & 2\, .
\end{eqnarray*}
So,
{
\small
\begin{eqnarray*}
         \frac{\textstyle 1}
              {\textstyle p_{1}\left( {\mathsf{\Lambda}}
                               \right)\,
                          p_{2}\left( {\mathsf{\Lambda}}
                               \right)}
         \sum_{{\mathsf{\lambda}},
               {\mathsf{\lambda}}'
               \in
               {\mathsf{\Lambda}}}
           \left( {\mathsf{\mu}}_{1}({\mathsf{\lambda}},
                                     {\mathsf{\lambda}}'
                                    )
                  +
                  {\mathsf{\mu}}_{2}({\mathsf{\lambda}},
                                     {\mathsf{\lambda}}'
                                    )
           \right)\,
           p_{1} ({\mathsf{\lambda}})\,
           p_{2} ({\mathsf{\lambda}}')
& \geq & 2 \, .
\end{eqnarray*}
}
So,
there is a player $i \in [2]$
with
{
\small
\begin{eqnarray*}
       b
& := & \frac{\textstyle 1}
           {\textstyle p_{1}\left( {\mathsf{\Lambda}}
                            \right)\,
                       p_{2}\left( {\mathsf{\Lambda}}
                            \right)}\,
         \sum_{{\mathsf{\lambda}},
               {\mathsf{\lambda}}'
               \in
               {\mathsf{\Lambda}}}
           {\mathsf{\mu}}_{i}({\mathsf{\lambda}},
                              {\mathsf{\lambda}}'
                             )\,
           p_{1} ({\mathsf{\lambda}})\,
           p_{2} ({\mathsf{\lambda}}')\ \
\geq\ \ 1\, .
\end{eqnarray*}
}
Without loss of generality,
set $i := 1$.
By the
{\it Weak-Equilibrium-for-Expectation} property
for player $1$,
${\mathsf{E}}_{1}\left( p_{1}^{{\mathsf{\lambda}}_{1}},
                             p_{2}
                     \right)
 =  {\mathsf{E}}_{1}\left( p_{1}^{{\mathsf{\lambda}}_{2}},
                             p_{2}
                     \right)$
for all ${\mathsf{\lambda}}_{1},
         {\mathsf{\lambda}}_{2}
         \in
         {\mathsf{\Lambda}}$
with
$p_{1} \left( {\mathsf{\lambda}}_{1}
       \right)
 >
 0$
and
$p_{1} \left( {\mathsf{\lambda}}_{2}
       \right)
 >
 0$.
For ${\mathsf{\lambda}}
     \in
     {\mathsf{\Lambda}}$,
it holds that
{
\small
\begin{eqnarray*}
      {\mathsf{E}}_{1}\left( p_{1}^{{\mathsf{\lambda}}},
                             p_{2}
                     \right)
& = & 2\, p_{2}\left( {\cal F}
               \right)
      +
      \sum_{{\mathsf{\lambda}}'
            \in
            {\mathsf{\Lambda}}}
        {\mathsf{\mu}}_{1}({\mathsf{\lambda}},
                           {\mathsf{\lambda}}'
                          )\,
        p_{2} \left( {\mathsf{\lambda}}'
              \right)\, .
\end{eqnarray*}
}
Hence,
{
\small
\begin{eqnarray*}
      \sum_{{\mathsf{\lambda}}'
            \in
            {\mathsf{\Lambda}}}
        {\mathsf{\mu}}_{1}({\mathsf{\lambda}}_{1},
                           {\mathsf{\lambda}}'
                          )\,
        p_{2} \left( {\mathsf{\lambda}}'
              \right)
& = & \sum_{{\mathsf{\lambda}}'
            \in
            {\mathsf{\Lambda}}}
        {\mathsf{\mu}}_{1}({\mathsf{\lambda}}_{2},
                           {\mathsf{\lambda}}'
                          )\,
        p_{2} \left( {\mathsf{\lambda}}'
              \right)
\end{eqnarray*}
}
for all ${\mathsf{\lambda}}_{1},
         {\mathsf{\lambda}}_{2}
         \in
         {\mathsf{\Lambda}}$
with
$p_{1} \left( {\mathsf{\lambda}}_{1}
       \right)
 >
 0$
and
$p_{1} \left( {\mathsf{\lambda}}_{2}
       \right)
 >
 0$.
Hence,
it follows that
for each ${\mathsf{\lambda}}
          \in
          {\mathsf{\Lambda}}$
with
$p_{1}({\mathsf{\lambda}})
 >
 0$,
{
\small
\begin{eqnarray*}
      \sum_{{\mathsf{\lambda}}'
            \in
            {\mathsf{\Lambda}}}
        {\mathsf{\mu}}_{1}({\mathsf{\lambda}},
                           {\mathsf{\lambda}}'
                          )\,
        p_{2} \left( {\mathsf{\lambda}}'
              \right)
& = & b
      \cdot
      p_{2} \left( {\mathsf{\Lambda}}
            \right)\, .
\end{eqnarray*}
}
So,
{
\small
\begin{eqnarray*}
      {\mathsf{E}}_{1}\left( p_{1}^{{\mathsf{\lambda}}},
                             p_{2}
                     \right)
& = & 2\, p_{2}\left( {\cal F}
              \right)
      +
      \underbrace{b}_{\geq 1}
      \cdot
      p_{2} \left( {\mathsf{\Lambda}}
            \right)\, .
\end{eqnarray*}
}
But for each strategy
${\mathsf{f}} \in {\cal F}$
with $p_{1}({\mathsf{f}}) > 0$,
\begin{eqnarray*}
{\mathsf{E}}_{1}\left( p_{1}^{{\mathsf{f}}},
                       p_{2}
                \right)
& = &
  a\,
  \cdot
  p_{2}\left( {\cal F}
       \right)
  +
  p_{2} \left( {\mathsf{\Lambda}}
        \right)\, ,
\end{eqnarray*}
for some $a$
with $|a| < 2$.
Hence,
the {\it Weak-Equilibrium-for-Expectation} property
for player $1$
implies that
${\mathsf{E}}_{1}\left( p_{1}^{{\mathsf{\lambda}}},
                        p_{2}
                 \right)
 =
 {\mathsf{E}}_{1}\left( p_{1}^{{\mathsf{f}}},
                        p_{2}
                \right)$,
or
\begin{eqnarray*}
      \underbrace{(2 - a)}_{> 0}
      \cdot
      p_{2}\left( {\cal F}
           \right)
& = & \underbrace{(1 - b)}_{\leq 0}
      \cdot
      p_{2} \left( {\mathsf{\Lambda}}
            \right)\, .
\end{eqnarray*}
A contradiction.
\end{proof}

\noindent
We now establish
some stronger properties
for a ${\mathsf{V}}$-equilibrium.

\begin{lemma}
\label{both players play only literals}
Consider a ${\mathsf{V}}$-equilibrium
$\langle p_{1},
         p_{2}
 \rangle$
for the game ${\mathsf{G}}$.
Then,
$p_{1}({\mathsf{L}})
 =
 p_{2}({\mathsf{L}})
 =
 1$,
and for every player $i \in [2]$,
for every literal
${\mathsf{\lambda}}
 \in
 {\mathsf{L}}$,
$p_{i}({\mathsf{\lambda}})
 \cdot
 p_{\overline{i}}(\overline{{\mathsf{\lambda}}})
 =
 0$.
Moreover,
${\mathsf{E}}_{1}(p_{1}, p_{2})
 =
 {\mathsf{E}}_{2}(p_{1}, p_{2})
 =
 1$ 
and
${\mathsf{R}}_{1}(p_{1}, p_{2})
 =
 {\mathsf{R}}_{2}(p_{1}, p_{2})
 =
 0$.
\end{lemma}

\begin{proof}
By Lemma~\ref{both players play only the formula},
$p_{1}({\mathsf{\Lambda}})
 =
 p_{2}({\mathsf{\Lambda}})
 =
 1$,
so that
$p_{1}({\cal F})
 =
 p_{2}({\cal F})
 =
 0$.
By the cost functions,
for each pair
${\mathsf{\lambda}},
 {\mathsf{\lambda}}
 \in
 {\mathsf{\Lambda}}$,
${\mathsf{\mu}}_{1}({\mathsf{\lambda}},
                    {\mathsf{\lambda}}')
 +
 {\mathsf{\mu}}_{2}({\mathsf{\lambda}},
                    {\mathsf{\lambda}}')
 \geq
 2$.
So,
${\mathsf{E}}_{1}(p_{1}, p_{2})
 +
 {\mathsf{E}}_{2}(p_{1}, p_{2})
 \geq
 2$.
It follows that
there is a player
$i \in [2]$
such that
${\mathsf{E}}_{i}({\mathsf{\Lambda}},
                  {\mathsf{\Lambda}})
 \geq
 1$.
For easier notation,
take $i := 1$.

By the cost functions,
for each
${\mathsf{\lambda}}
 \in
 {\mathsf{\Lambda}}$,
${\mathsf{\mu}}_{1}({\mathsf{f}}_{1},
                    {\mathsf{\lambda}})
 =
 1$.
If player $1$ switches to
the pure strategy $p_{1}^{{\mathsf{f}}_{1}}$,
then
{\it (i)}
${\mathsf{E}}_{1}(p_{1}^{{\mathsf{f}}_{1}},
                  p_{2})
 =
 1
 \leq
 {\mathsf{E}}_{1}(p_{1},
                  p_{2})$,
and
{\it (ii)}
${\mathsf{R}}_{1}(p_{1}^{{\mathsf{f}}_{1}},
                  p_{2})
 =
 0$,
by the {\it Risk-Positivity} property.
Since
${\mathsf{R}}_{1}(p_{1}, p_{2})
 \geq
 0$
(by the {\it Risk-Positivity} property),
it follows that
if either
${\mathsf{E}}_{1}(p_{1}, p_{2})
 >
 1$
or
${\mathsf{R}}_{1}(p_{1}, p_{2})
 >
 0$,
then player $1$ improves her cost
by switching to
$p_{1}^{{\mathsf{f}}_{1}}$.
Since
$\langle p_{1}, p_{2}
 \rangle$
is a ${\mathsf{V}}$-equilibrium,
it follows that both
${\mathsf{E}}_{1}(p_{1}, p_{2})
 =
 1$
and
${\mathsf{R}}_{1}(p_{1}, p_{2})
 =
 0$.

Now,
${\mathsf{E}}_{1}(p_{1}, p_{2})
 +
 {\mathsf{E}}_{2}(p_{1}, p_{2})
 \geq
 2$
and
${\mathsf{E}}_{1}(p_{1}, p_{2})
 =
 1$
together imply that
${\mathsf{E}}_{2}(p_{1}, p_{2})
 \geq
 1$.
In the same way as above,
${\mathsf{E}}_{2}(p_{1}, p_{2})
 =
 1$
and
${\mathsf{R}}_{2}(p_{1}, p_{2})
 =
 0$
follow.

By the cost functions,
for each pair
${\mathsf{\lambda}},
 {\mathsf{\lambda}}'
 \in
 {\mathsf{\Lambda}}$,
${\mathsf{\mu}}_{1}({\mathsf{\lambda}},
                    {\mathsf{\lambda}}')
 +
 {\mathsf{\mu}}_{1}({\mathsf{\lambda}},
                    {\mathsf{\lambda}}')
 \geq
 2$.
We have just shown that
${\mathsf{E}}_{1}(p_{1}, p_{2})
 +
 {\mathsf{E}}_{2}(p_{1}, p_{2})
 =
 2$.
This implies that
${\mathsf{\mu}}_{1}({\mathsf{\lambda}},
                    {\mathsf{\lambda}}')
 =
 {\mathsf{\mu}}_{2}({\mathsf{\lambda}},
                    {\mathsf{\lambda}}')
 =
 1$
for each pair
${\mathsf{\lambda}},
 {\mathsf{\lambda}}'
 \in
 {\mathsf{\Lambda}}$
with
$p_{1}({\mathsf{\lambda}})
 >
 0$
and
$p_{2}({\mathsf{\lambda}}')
 >
 0$.
Thus,
by the cost functions,
${\mathsf{\lambda}},
 {\mathsf{\lambda}}'
 \in
 {\mathsf{L}}$,
so that
$p_{1}({\mathsf{L}})
 =
 p_{2}({\mathsf{L}})
 =
 1$,
and  
$p_{i}({\mathsf{\lambda}})
 >
 0$
implies that
$p_{\overline{i}}(\overline{{\mathsf{\lambda}}})
 =
 0$
for all pairs of a player
$i \in [2]$
and a literal
${\mathsf{\lambda}} \in {\mathsf{L}}$.
\end{proof}

\noindent
We continue to prove:

\begin{lemma} 
\label{literals are played uniformly}
In a ${\mathsf{V}}$-equilibrium
$\langle p_{1},
         p_{2}
 \rangle$
for ${\mathsf{G}}$,
for each player $i \in [2]$
and
for each literal
$\ell \in {\mathsf{L}}$,
$p_{i}(\ell)
 +
 p_{i}(\overline{\ell})
 >
 0$.
\end{lemma}

\begin{proof}
Assume,
by way of contradiction,
that there is
a player $i \in [2]$
and a literal
$\ell \in {\mathsf{L}}$
such that
$p_{i}(\ell)
 =
 p_{i}(\overline{\ell})
 =
 0$.
Without loss of generality,
set $i := 2$.
Recall that,
by Lemma~\ref{both players play only literals},
${\mathsf{V}}_{1}(p_{1}, p_{2})
 =
 1$
and
$p_{2}({\mathsf{L}})
 =
 1$.

Consider player $1$
switching to the pure strategy
$p_{1}^{{\mathsf{v}}}$
for the variable ${\mathsf{v}}$
such that
$\ell$ and $\overline{\ell}$
are literals for ${\mathsf{v}}$.
By the cost functions,
${\mathsf{\mu}}_{1}\left( {\mathsf{v}},
                          {\mathsf{\lambda}}
                   \right)
 =
 m$
if ${\mathsf{\lambda}}
    \in
    \{ \ell, \overline{\ell}
    \}$,
and $0$
if ${\mathsf{\lambda}}
    \in
    {\mathsf{L}}
    \setminus
    \{ \ell, \overline{\ell}
    \}$.
Since
$p_{2}({\mathsf{L}})
 =
 1$
and
$p_{2}(\ell)
 =
 p_{2}(\overline{\ell})
 =
 0$,
it follows that
${\mathsf{\mu}}_{1}\left( {\mathsf{v}},
                          {\mathsf{\lambda}}
                   \right)
 =
 0$ 
for all strategies ${\mathsf{\lambda}}$
with
$p_{2}({\mathsf{\lambda}})
 >
 0$.
Thus,
{\it (i)}
${\mathsf{E}}_{1}\left( p_{1}^{{\mathsf{v}}},
                        p_{2}
                 \right)       
 =
 0$,
and
{\it (ii)}
${\mathsf{R}}_{1}\left( p_{1}^{{\mathsf{v}}},
                        p_{2}
                 \right)       
 =
 0$
(by the {\it Risk-Positivity} property). 
Hence,
${\mathsf{V}}_{1}\left( p_{1}^{{\mathsf{v}}},
                        p_{2}
                 \right)       
 =
 0$,
so that player $1$ improves her cost
by switching to
the pure strategy
$p_{1}^{{\mathsf{v}}}$.
A contradiction to the assumption that
$\langle p_{1},
         p_{2}
 \rangle$ 
is a ${\mathsf{V}}$-equilibrium.
\end{proof}

\noindent
Finally, we prove:

\begin{lemma}
\label{induce truth assignment}
A ${\mathsf{V}}$-equilibrium
$\langle p_{1}, p_{2}
 \rangle$
for ${\mathsf{G}}$
induces a unique truth assignment
for ${\mathsf{\phi}}$.
\end{lemma}

\begin{proof}
For each pair of a player
$i \in [2]$
and a literal $\ell \in {\mathsf{L}}$,
it holds that
{\it (i)}
$p_{1}(\ell) +
 p_{1}(\overline{\ell})
 >
 0$
(by Lemma~\ref{literals are played uniformly}),
and
{\it (ii)}
if $p_{i}(\ell) > 0$,
then
$p_{\overline{i}}(\overline{\ell}) = 0$
(by Lemma~\ref{both players play only literals}).
Thus,
for each variable
${\mathsf{v}}$,
there is a literal
$\ell$ for ${\mathsf{v}}$
such that
$p_{1}(\ell),
 p_{2}(\ell)
 >
 0$
and
$p_{1}(\overline{\ell})
 =
 p_{2}(\overline{\ell})
 =
 0$.
\end{proof}

\noindent
We are now ready to prove:

\begin{lemma}
\label{unsat implies no equilibrium}
${\mathsf{\phi}}$
is satisfiable
if and only if
${\mathsf{G}}({\mathsf{\phi}})$
has a
${\mathsf{V}}$-equilibrium.
\end{lemma}

\begin{proof}
``$\Leftarrow$:''
Assume first that
${\mathsf{\phi}}$
is {\em not} satisfiable.
Assume,
by way of contradiction,
that
${\mathsf{G}}({\mathsf{\phi}})$
has a
${\mathsf{V}}$-equilibrium
$\langle p_{1}, p_{2}
 \rangle$.
By Lemma~\ref{induce truth assignment},
$\langle p_{1}, p_{2}
 \rangle$
induces a unique truth assignment
$\gamma$
for ${\mathsf{\phi}}$.
Since
${\mathsf{\phi}}$
is not satisfiable,
there is a clause ${\mathsf{c}}$
such that
for each literal $\ell$
with $p_{1}(\ell), p_{2}(\ell) > 0$,
$\ell \not\in {\mathsf{c}}$.
Consider now player $1$
switching to
the pure strategy $p_{1}^{{\mathsf{c}}}$.
Then,
{
\small
\begin{eqnarray*}
      {\mathsf{E}}_{1}\left( p_{1}^{{\mathsf{c}}},
                             p_{2}
                      \right)
& = & \sum_{\ell \in {\mathsf{L}}
            \mid
            p_{2}(\ell) > 0}
        \underbrace{{\mathsf{\mu}}_{1}({\mathsf{c}},
                                       \ell)}_{= 0}\,
        \cdot
        p_{2}(\ell)\
       \
      =\
      \
      0\, ,
\end{eqnarray*}
}
and
${\mathsf{R}}_{1}\left( p_{1}^{{\mathsf{c}}},
                        p_{2}
                 \right)
 =
 0$
(by the {\it Risk-Positivity} property),
so that
${\mathsf{V}}_{1}\left( p_{1}^{{\mathsf{c}}},
                        p_{2}
                 \right)
 =
 0$.                 
By Lemma~\ref{both players play only literals},
${\mathsf{V}}_{1}(p_{1}, p_{2})
 =
 1$.
So,
player $1$ improves her cost
by switching to
the pure strategy
$p_{1}^{{\mathsf{c}}}$.
A contradiction to the assumption that
$\langle p_{1}, p_{2}
 \rangle$
is a
${\mathsf{V}}$-equilibrium.

``$\Rightarrow$:''
Assume now that
${\mathsf{\phi}}$
is satisfiable.
For a satisfying assignment
$\gamma$ of ${\mathsf{\phi}}$,
set
$p_{i}(\ell)
 :=
 \frac{\textstyle 1}
      {\textstyle m}$
for each literal $\ell \in {\mathsf{L}}$
with
$\gamma (\ell)
 =
 1$.
We shall prove that
$\langle p_{1}, p_{2}
 \rangle$
is a
${\mathsf{V}}$-equilibrium.
Fix a player $i \in [2]$.
By the cost functions,
${\mathsf{\mu}}_{i}(\ell_{j},
                    \ell_{k})
 =
 1$
for each pair of literals
$\ell_{j},
 \ell_{k}$
with
$p_{1}(\ell_{j})
 \cdot
 p_{2}(\ell_{k})
 >
 0$.
Hence, 
the {\it Risk-Positivity} property
implies that
${\mathsf{R}}_{i}(p_{1}, p_{2})
 =
 0$.
Furthermore,
{
\small
\begin{eqnarray*}
      {\mathsf{E}}_{i}(p_{1}, p_{2})
& = &
\sum_{\langle \ell_{j},
                    \ell_{k}
            \rangle
            \in
            {\mathsf{L}}
            \mid
            \gamma(\ell_{j})
            =
            \gamma(\ell_{k})
            =
            1}
        {\mathsf{\mu}}_{i}(\ell_{j},
                           \ell_{k})\,
                   p_{1}(\ell_{j})
                   \cdot
                   p_{2}(\ell_{k})\
      \
      =\
      \
      m^{2}
      \cdot
      1
      \cdot
      \frac{\textstyle 1}
           {\textstyle m}
      \cdot
      \frac{\textstyle 1}
           {\textstyle m}\
      \
      =\
      \
      1\, .
\end{eqnarray*}
}
So,
player $i$
may decrease
${\mathsf{V}}_{i}$
only if she decreases
${\mathsf{E}}_{i}$.
It suffices to prove that
player $i$ cannot decrease ${\mathsf{E}}_{i}$
by switching to a pure strategy.
We proceed by case analysis.
\begin{enumerate}

\item
Consider player $i$ switching
to the pure strategy $p_{1}^{{\mathsf{v}}}$
for a variable
${\mathsf{v}} \in {\cal V}$
with literals
$\ell, \overline{\ell}$.
By the cost functions,
for each literal
${\mathsf{\lambda}}
 \in
 {\mathsf{L}}$,
${\mathsf{\mu}}_{i}\left( {\mathsf{v}},
                          {\mathsf{\lambda}}
                   \right)
 =
 m$
if ${\mathsf{\lambda}}
    \in
    \{ \ell, \overline{\ell}
    \}$
and $0$ otherwise.
By construction,
$p_{i}(\ell)
 +
 p_{i}(\overline{\ell})
 =
 \frac{\textstyle 1}
      {\textstyle m}$.
So,
{
\small
\begin{eqnarray*}
      {\mathsf{E}}_{i}\left( p_{i}^{{\mathsf{v}}},
                             p_{\overline{i}}
                      \right)
& = & 0
      \cdot
      \left( 1 - \frac{\textstyle 1}
                      {\textstyle m}
      \right)
      +
      m
      \cdot
      \frac{\textstyle 1}
           {\textstyle m}\
     \
    =\
     \
    1\, .
\end{eqnarray*}
}

\item
Consider player $i$ switching
to the pure strategy
$p_{i}^{{\mathsf{c}}}$
for a clause
${\mathsf{c}} \in {\cal C}$.
By the cost functions,
for each literal
${\mathsf{\lambda}}
 \in
 {\mathsf{L}}$,
${\mathsf{\mu}}_{i}\left( {\mathsf{c}},
                          {\mathsf{\lambda}}
                   \right)
 =
 m$
if ${\mathsf{\lambda}}
    \in
    {\mathsf{c}}$
and $0$ otherwise.
Since ${\mathsf{\phi}}$
is satisfiable,
there is at least one literal
$\ell
 \in
 {\mathsf{c}}$
with ${\mathsf{\gamma}}(\ell)
      =
      1$;
hence,
by construction of $p_{\overline{i}}$,
there is at least one literal
$\ell
 \in
 {\mathsf{c}}$
with
$p_{\overline{i}}(\ell)
 =
 \frac{\textstyle 1}
      {\textstyle m}$.
Thus,
{
\small
\begin{eqnarray*}
      {\mathsf{E}}_{i}\left( p_{i}^{{\mathsf{c}}},
                             p_{\overline{i}}
                      \right)
& = & \sum_{\ell \in {\mathsf{c}}
            \mid
            p_{\overline{i}}(\ell) > 0}
        \underbrace{{\mathsf{\mu}}_{i}\left( {\mathsf{c}},
                                             \ell
                                      \right)}_{=m}
        \cdot
        p_{\overline{i}}(\ell)\
        \
       \geq\
       \
       m
       \cdot
       \frac{\textstyle 1}
            {\textstyle m}\
       \
      =\
      \
      1\, .
\end{eqnarray*}
}

\item
Consider player $i$ switching
to the pure strategy
$p_{i}^{{\mathsf{f}}}$
for some ${\mathsf{f}} \in
          {\cal F}$.
By construction
of the cost functions,
for each literal
${\mathsf{\lambda}}
 \in
 {\mathsf{L}}$,
${\mathsf{\mu}}_{i}\left( {\mathsf{f}},
                          \ell
                   \right)
 =
 1$.
It follows that
${\mathsf{E}}_{i}\left( p_{i}^{{\mathsf{f}}},
                        p_{\overline{i}}
                 \right)
 =
 1$.

\item
Finally,
consider player $i$ switching
to the pure strategy
$p_{i}^{\ell}$
for some literal
$\ell \in {\mathsf{L}}$.
Assume first that
${\mathsf{\gamma}}(\ell) = 1$.
Then,
$p_{\overline{i}}(\overline{\ell})
 =
 0$.
Hence, 
by the cost functions,
${\mathsf{\mu}}_{i}(\ell, \ell^{\prime})
 =
 1$
for each literal
$\ell^{\prime} \in
 {\mathsf{L}}$
with
$p_{\overline{i}}(\overline{\ell})
 >
 0$.
It follows that
${\mathsf{E}}_{i}\left( p_{i}^{\ell},
                        p_{\overline{i}}
                 \right)
 =
 1$.
Assume now that
${\mathsf{\gamma}}(\ell) = 0$.
Then,
$p_{\overline{i}}(\overline{\ell})
 =
 \frac{\textstyle 1}
      {\textstyle m}$.
By the cost functions,
${\mathsf{\mu}}_{i}(\ell, \overline{\ell})
 =
 2$
and
${\mathsf{\mu}}_{i}(\ell, \ell^{\prime})
 =
 1$  
for
$\ell^{\prime} \in {\mathsf{L}} \setminus \{ \overline{\ell}
                                          \}$
with
$p_{\overline{i}}(\ell^{\prime})
 >
 0$.
It follows that
{
\small
\begin{eqnarray*} 
      {\mathsf{E}}_{i}\left( p_{i}^{\ell},
                             p_{\overline{i}}
                      \right)                           
& = & 2 \cdot \frac{\textstyle 1}
                   {\textstyle m}
      +
      1 \cdot \left( 1 - \frac{\textstyle 1}
                              {\textstyle m}
              \right)\
      \
     =\
     \
     1 
      +
      \frac{\textstyle 1}
           {\textstyle m}\, .                                     
\end{eqnarray*}
}

\end{enumerate}
The claim now follows.
\end{proof}

\noindent
Lemma~\ref{unsat implies no equilibrium}
establishes the reduction
for the ${\mathcal{NP}}$-hardness;
since the numbers involved in the reduction
are polynomially bounded,
strong ${\mathcal{NP}}$-hardness follows.
\end{proof}

\subsection{Concrete ${\mathcal{NP}}$-Hardness Result}
\label{concrete np hardness}

We remark that the proof of
the reduction for Theorem~\ref{two players complexity}
is modular
in treating ${\mathsf{V}}$ in an abstract way
through using {\it Risk-Positivity,}
{\it Weak-Equilibrium-for-Expectation,}
and the properties
in Condition {\sf (2)}.
This modularity yields an extension
of Theorem~\ref{two players complexity}
to concrete $({\mathsf{E}} +
              {\mathsf{R}})$-valuations
enjoying these properties.   
We shall verify Conditions
{\sf (1)} and {\sf (2)}
from Theorem~\ref{two players complexity}
for the
$({\mathsf{E}}
 +
 {\mathsf{R}})$-valuations
${\mathsf{V}}$,
where
{\sf (1)}
${\mathsf{R}}
 =
 {\mathsf{\gamma}}
 \cdot
 {\mathsf{Var}}$,
or
{\sf (2)}
${\mathsf{R}}
 =
 {\mathsf{\gamma}}
 \cdot
 {\mathsf{SD}}$,
or a valuation
{\sf (3)}
${\mathsf{V}}
 =
 \lambda\,
 ({\mathsf{E}} + 
 {\mathsf{\gamma}}
 \cdot
 {\mathsf{Var}})
 +
 (1-\lambda)\,
 {\mathsf{V}}^{{\mathsf{\nu}}}$,
with $0 < \lambda \leq 1$,
where
${\mathsf{\nu}}(x)
 =
 x^{r}$,
with $r \geq 2$,
and with ${\mathsf{\gamma}} > 0$.
The {\it Weak-Equilibrium-for-Expectation} property
in Condition {\sf (1)}
follows from Corollary~\ref{brand new}.
For Condition {\sf (2)},
we shall prove three technical claims
associated with
Conditions {\sf (2/a)},
{\sf (2/b)} and
{\sf (2/c)},
respectively.
We start with Condition {\sf (2/a)}.
We remark that the
existence of a ${\mathsf{\delta}}$
such that
$\widehat{{\mathsf{R}}}(1, 1 + 2 {\mathsf{\delta}}, q)
 <
 \frac{\textstyle 1}
      {\textstyle 2}$
for each probability
$q \in [0, 1]$
follows already from the continuity
of ${\mathsf{R}}$;
but its polynomial time computation
is bound to depend
on each particular
$\widehat{{\mathsf{R}}}$.
We prove:

\begin{lemma}
\label{little tiny}
Fix an
$({\mathsf{E}}
 +
 {\mathsf{R}})$-valuation
${\mathsf{V}}$,
where
{\sf (1)}
${\mathsf{R}}
 =
 {\mathsf{\gamma}}
 \cdot
 {\mathsf{Var}}$,
or
{\sf (2)}
${\mathsf{R}}
 =
 {\mathsf{\gamma}}
 \cdot
 {\mathsf{SD}}$,
or
{\sf (3)}
${\mathsf{V}}
 =
 \lambda\,
 ({\mathsf{E}} + {\mathsf{\gamma}} \cdot {\mathsf{Var}})
 +
 (1-\lambda)\,
 {\mathsf{V}}^{{\mathsf{\nu}}}$,
with $0 < \lambda \leq 1$,
where
${\mathsf{\nu}}(x)
 =
 x^{r}$
with $r \geq 2$,
and with ${\mathsf{\gamma}} > 0$.
Then,
there is a polynomial time computable
${\mathsf{\Delta}}$
with
$0 < {\mathsf{\Delta}}
   \leq \frac{\textstyle 1}
                  {\textstyle 4}$
such that
$\widehat{{\mathsf{R}}}(1, 1 + 2 {\mathsf{\delta}}, q)
 <
 \frac{\textstyle 1}
      {\textstyle 2}$
for all
$q \in [0, 1]$
and
$0 \leq {\mathsf{\delta}} < {\mathsf{\Delta}}$.   
\end{lemma}

\begin{proof}
For each valuation
${\mathsf{V}}$,
we shall choose a suitable ${\mathsf{\Delta}}$.

\begin{enumerate}

\item[{\sf (1)}]
\underline{${\mathsf{R}} = {\mathsf{\gamma}} \cdot {\mathsf{Var}}$:}
Then,
${\mathsf{\gamma}}
 \cdot
 \widehat{{\mathsf{R}}}\left( 1, 1 + 2 {\mathsf{\delta}}, q
                       \right)
 =
 {\mathsf{\gamma}}
 \cdot
 q (1-q)
 \cdot
 4 {\mathsf{\delta}}^{2}
 \leq
 {\mathsf{\gamma}}
 \cdot
 \frac{\textstyle 1}
      {\textstyle 4}
 \cdot
 4 {\mathsf{\delta}}^{2}
 =
 {\mathsf{\gamma}}
 \cdot
 {\mathsf{\delta}}^{2}$.
Choose ${\mathsf{\Delta}}$
as a rational number 
no larger than
$\frac{\textstyle 1}
      {\textstyle 4}
 \cdot
 \min \left\{ \frac{\textstyle 1}
                   {\textstyle \sqrt{{\mathsf{\gamma}}}},
              1
      \right\}$.                   
This choice
satisfies that
for each 
${\mathsf{\delta}} < {\mathsf{\Delta}}$,
${\mathsf{\gamma}}
 \cdot
 {\mathsf{\delta}}^{2}
 <
 \frac{\textstyle 1}
      {\textstyle 2}$.

\item[{\sf (2)}]
\underline{${\mathsf{R}} = {\mathsf{\gamma}} \cdot {\mathsf{SD}}$:}
Then,
${\mathsf{\gamma}}
 \cdot
 \widehat{{\mathsf{R}}}\left( 1, 1 + 2 {\mathsf{\delta}}, q
                       \right)
 =
 {\mathsf{\gamma}}
 \cdot
 \sqrt{q (1-q)}
 \cdot
 2 {\mathsf{\delta}}
 \leq
 {\mathsf{\gamma}}
 \cdot
 \frac{\textstyle 1}
      {\textstyle 2}
 \cdot
 2 {\mathsf{\delta}}
 =
 {\mathsf{\gamma}}
 \cdot
 {\mathsf{\delta}}$.
Choose
${\mathsf{\Delta}}
 :=
 \frac{\textstyle 1}
      {\textstyle 4}
 \cdot
 \min \left\{ \frac{\textstyle 1}
                   {\textstyle {\mathsf{\gamma}}},
              1
      \right\}$. 
This choice satisfies
that for each
${\mathsf{\delta}}
 <
 {\mathsf{\Delta}}$,
${\mathsf{\gamma}}
 \cdot
 {\mathsf{\delta}}
 <
 \frac{\textstyle 1}
      {\textstyle 2}$.

\item[{\sf (3)}]
\underline{${\mathsf{V}} =
            {\mathsf{\lambda}}
            ({\mathsf{E}} + {\mathsf{\gamma}} \cdot {\mathsf{Var}})
            +
            (1 - {\mathsf{\lambda}})
            {\mathsf{V}}^{{\mathsf{\nu}}}$
            where ${\mathsf{\nu}}(x) = x^{r}$ with $r \geq 2$:}\\
Consider first the risk valuation
${\mathsf{R}} = {\mathsf{V}}^{{\mathsf{\nu}}} - {\mathsf{E}}$,
where ${\mathsf{\nu}}(x) = x^{r}$ with $r \geq 2$.
Note that
{
\small
\begin{eqnarray*}
      \widehat{{\mathsf{R}}}\left( 1, 1 + 2 {\mathsf{\delta}}, q
                       \right)
& = & \sqrt[{\textstyle r}]{q \cdot \left( 1 + 2 {\mathsf{\delta}})
                                    \right)^{r}
                            +
                            (1-q) 
              }
      -
      1
      -
      q \cdot (2 {\mathsf{\delta}})\, .
\end{eqnarray*}
}
Set
${\mathsf{\Delta}}
 :=
 \frac{\textstyle 1}
      {\textstyle 4}$.
Then,
for every
${\mathsf{\delta}} < {\mathsf{\Delta}}$,
the fact that
$\widehat{{\mathsf{R}}}\left( 1, 1 + 2 {\mathsf{\delta}}, q
                       \right)$
increases monotonically in ${\mathsf{\delta}}$
implies that                       
{
\small
\begin{eqnarray*}
      \widehat{{\mathsf{R}}}\left( 1, 1 + 2 {\mathsf{\delta}}, q
                            \right)
&\leq& \sqrt[{\textstyle r}]{1 + q
                   \cdot
                   \left( \left( \frac{\textstyle 3}
                                      {\textstyle 2}
                          \right)^{r}
                          -
                          1
                   \right)
              }
      -
      1
      -
      q
      \cdot
      \frac{\textstyle 1}
           {\textstyle 2}\, ;
\end{eqnarray*}
}
for $q = 0$,
$\sqrt[{\textstyle r}]{1 + 0
                   \cdot
                   \left( \left( \frac{\textstyle 3}
                                      {\textstyle 2}
                          \right)^{r}
                          -
                          1
                   \right)
              }
      -
      1
      -
      0
      \cdot
      \frac{\textstyle 1}
           {\textstyle 2}
 =
 0
 <
 \frac{\textstyle 1}
      {\textstyle 2}$.
So assume that
$q > 0$.
Then,
by the fact that
$\sqrt[{\textstyle r}]{1 + q
                   \cdot
                   \left( \left( \frac{\textstyle 3}
                                      {\textstyle 2}
                          \right)^{r}
                          -
                          1
                   \right)
              }$
increases monotonically in $q$
and $0 < q \leq 1$,
we get that           
$\widehat{{\mathsf{R}}}\left( 1, 1 + 2 {\mathsf{\delta}}, q
                       \right)
 \leq
 \frac{\textstyle 3}
      {\textstyle 2}
 -
 1
 -
 q \cdot
 \frac{\textstyle 1}
      {\textstyle 2}
 =
 \frac{\textstyle 1}
      {\textstyle 2}
 -
 q
 \cdot
 \frac{\textstyle 1}
      {\textstyle 2}
 <
 \frac{\textstyle 1}
      {\textstyle 2}$
since
$q > 0$.

\noindent
Combined with the choice of ${\mathsf{\Delta}}$
for {\sf (1)},
the required property
for the convex combination
${\mathsf{V}} =
 {\mathsf{\lambda}}
 ({\mathsf{E}} + {\mathsf{\gamma}} \cdot {\mathsf{Var}})
 +
 (1 - {\mathsf{\lambda}})
 {\mathsf{V}}^{{\mathsf{\nu}}}$,
where ${\mathsf{\nu}}(x) = x^{r}$ with $r \geq 2$,
holds by choosing
${\mathsf{\Delta}}$
as a rational number
no larger than
$\min \left\{ \frac{\textstyle 1}
                   {\textstyle 4}
              \cdot
              \min \left\{ \frac{\textstyle 1}
                                {\textstyle \sqrt{\gamma}},
                           1
                   \right\},
              \frac{\textstyle 1}
                   {\textstyle 4}
      \right\}             
 =                  
 \frac{\textstyle 1}
      {\textstyle 4}
 \cdot
 \min \left\{ \frac{\textstyle 1}
                   {\textstyle \sqrt{\gamma}},
              1
      \right\}$.

\end{enumerate}
\end{proof}

\noindent
We continue with Condition {\sf (2/b)}.
We prove:

\begin{lemma}
\label{the monotonicity lemma}
Fix an
$({\mathsf{E}}
  +
  {\mathsf{R}})$-valuation
${\mathsf{V}}$,  
where
{\sf (1)}
${\mathsf{R}}
 =
 {\mathsf{\gamma}}
 \cdot
 {\mathsf{Var}}$,
or
{\sf (2)}
${\mathsf{R}}
 =
 {\mathsf{\gamma}}
 \cdot
 {\mathsf{SD}}$,
or
{\sf (3)}
${\mathsf{V}}
 =
 \lambda\,
 ({\mathsf{E}} + {\mathsf{\gamma}} \cdot {\mathsf{Var}})
 +
 (1-\lambda)\,
 {\mathsf{V}}^{{\mathsf{\nu}}}$,
with $0 < \lambda \leq 1$,
where
${\mathsf{\nu}}$
is increasing and strictly convex,
and with ${\mathsf{\gamma}} > 0$.
Then,
there is a polynomial time computable
${\mathsf{\Delta}}$
with
$0 < {\mathsf{\Delta}}
   \leq
   \frac{\textstyle 1}
           {\textstyle 4}$
such that
$\widehat{{\mathsf{V}}}(1, 1 + 2 {\mathsf{\delta}}, r)
 <
 \widehat{{\mathsf{V}}}(1, 2, q)$
for all 
$q \in (0, 1)$,
with $0 \leq r \leq q$,
and
$0 \leq {\mathsf{\delta}} 
   < {\mathsf{\Delta}}$.
\end{lemma}

\begin{proof}
For each valuation
${\mathsf{V}}$,
we shall choose a suitable ${\mathsf{\Delta}}$.

\noindent
\underline{Case {\sf (1):} 
           ${\mathsf{R}}
            =
            {\mathsf{\gamma}}
            \cdot
            {\mathsf{Var}}$:}
Note that
{
\small
\begin{eqnarray*}
      \widehat{{\mathsf{V}}}(1, 
                             1+2 {\mathsf{\delta}},
                             r)
& = & (1+2 {\mathsf{\delta}})
      \cdot
      r
      +
      1
      \cdot
      (1-r)
      +
      {\mathsf{\gamma}}
      \cdot
      r
      (1-r)
      \cdot
      4 {\mathsf{\delta}}^{2}                              \\
& = & 2 {\mathsf{\delta}}
      \cdot
      r    
      + 1
      +
      {\mathsf{\gamma}}
      \cdot
      r
      (1-r)
      \cdot
      4 {\mathsf{\delta}}^{2}\, ,   
\end{eqnarray*}
}
and
{
\small
\begin{eqnarray*}
      \widehat{{\mathsf{V}}}(1, 
                             2,
                             q)
& = & 2 \cdot q +
      1 \cdot (1-q)
      +
      {\mathsf{\gamma}}
      \cdot
      q (1-q)                                             \\
& = & q + 1 + 
      {\mathsf{\gamma}}
      \cdot
      q (1-q)\, .
\end{eqnarray*}
}
Note also that
{
\small
\begin{eqnarray*}
      \max_{0 \leq r \leq q}
        r (1-r)
& = & \left\{ \begin{array}{ll}
                \frac{\textstyle 1}
                     {\textstyle 4}\, ,        & \mbox{if $q \geq \frac{\textstyle 1}
                                                                       {\textstyle 2}$}        \\
                q(1-q)\, ,                     & \mbox{otherwise}
              \end{array}
      \right.\, .          
\end{eqnarray*}
}
\noindent
If
$q \leq \frac{\textstyle 1}
                   {\textstyle 2}$,
then
$r (1-r)
 \leq
 q (1-q)$.
Choosing
${\mathsf{\Delta}}
 :=
 \frac{\textstyle 1}
         {\textstyle 4}$,
this implies immediately that
$\widehat{{\mathsf{V}}}(1, 
                                         1+2 {\mathsf{\delta}},
                                         r)
 <
 \widehat{{\mathsf{V}}}(1, 
                                         2,
                                         q)$
for
$0 \leq {\mathsf{\delta}} < {\mathsf{\Delta}}$.
Consider now
$q > \frac{\textstyle 1}
          {\textstyle 2}$
and set 
${\mathsf{\Delta}}
 :=
 \min
 \left\{ \frac{\textstyle 1}
                    {\textstyle 4},
            \frac{\textstyle 1}
                    {\textstyle 2 (1+{\mathsf{\gamma}})}
 \right\}$.
Denote
$A
 := 1 + {\mathsf{\delta}}
             (2q + {\mathsf{\gamma}})$.
Then,
for ${\mathsf{\delta}} < {\mathsf{\Delta}}$,
{
\small
\begin{eqnarray*}
         \widehat{{\mathsf{V}}}(1, 
                                1+2 {\mathsf{\delta}},
                                r)
& \leq & 2 {\mathsf{\delta}}
              \cdot
              q    
             + 1
             +
             {\mathsf{\gamma}}
             \cdot
             {\mathsf{\delta}}^{2}\
             \
            <\
            \
            A\, ,
\end{eqnarray*}
}
while
{
\small
\begin{eqnarray*}
       \lefteqn{A}                                                                                                                                                \\
< & 1 + {\mathsf{\Delta}}
             \cdot
             (2q + {\mathsf{\gamma}})
    & \mbox{(since ${\mathsf{\delta}}
                               <
                               {\mathsf{\Delta}}$)}                                                                                                        \\   
= & 1 + \frac{\textstyle 2q + {\mathsf{\gamma}}}
                     {\textstyle 2 (1 + {\mathsf{\gamma}})}
    &                                                                                                                                                                   \\
< & 1 + \frac{\textstyle 2q (1 + {\mathsf{\gamma}})}
                     {\textstyle 2 (1 + {\mathsf{\gamma}})}
    & \mbox{(since $q > \frac{\textstyle 1}
                                              {\textstyle 2}$)}                                                                                                \\
= & 1 + q
    &                                                                                                                                                                   \\
< &  \widehat{{\mathsf{V}}}(1, 
                        2,
                        q)\, .
\end{eqnarray*}
}

\noindent
\underline{Case {\sf (2):} 
           ${\mathsf{R}}
            =
            {\mathsf{\gamma}}
            \cdot
            {\mathsf{SD}}$:}
Then,
{
\small
\begin{eqnarray*}
      \widehat{{\mathsf{V}}}(1, 
                             1+2 {\mathsf{\delta}},
                             r)
& = & 2 {\mathsf{\delta}}
      \cdot
      r
      +
      1
      +
      {\mathsf{\gamma}}
      \cdot
      \sqrt{r
            (1-r)}
      \cdot
      2 {\mathsf{\delta}}\, ,
\end{eqnarray*}
}
and
{
\small
\begin{eqnarray*}
      \widehat{{\mathsf{V}}}(1, 
                             2,
                             q)
& = & q +
      1
      +
      {\mathsf{\gamma}}
      \cdot
      \sqrt{q (1-q)}\, .
\end{eqnarray*}
}
For
$q \leq \frac{\textstyle 1}
                    {\textstyle 2}$,
the argument is identical
to the one for Case {\sf (1)}.
Consider now
$q > \frac{\textstyle 1}
          {\textstyle 2}$,
and set again
${\mathsf{\Delta}}
 :=
 \min
 \left\{ \frac{\textstyle 1}
                   {\textstyle 4},
            \frac{\textstyle 1}
                    {\textstyle 2 (1+{\mathsf{\gamma}})}
 \right\}$.
Denote again
$      A
 := 1 + {\mathsf{\delta}}
             (2q +  {\mathsf{\gamma}})$.
By arguments identical to those
for Case {\sf (1)},
we derive that
$\widehat{{\mathsf{V}}}(1, 
                                         1+2 {\mathsf{\delta}},
                                          r)
  <
  A$
and
$A
  \leq
  q+1
  <
  \widehat{{\mathsf{V}}}(1, 
                                          2,
                                          q)$.

\noindent
\underline{Case {\sf (3):} ${\mathsf{V}}
                            =
                            \lambda
                            \cdot
                            ({\mathsf{E}} +
                             {\mathsf{\gamma}}
                             \cdot
                             {\mathsf{Var}})
                            +
                            (1 - \lambda)
                            \cdot
                            {\mathsf{V}}^{{\mathsf{\nu}}}$,
                            ${\mathsf{\nu}}$ increasing
                            and strictly convex.}\\  
Consider first
the valuation
${\mathsf{V}}
 =
 {\mathsf{V}}^{{\mathsf{\nu}}}$.
Note that
{
\small
\begin{eqnarray*}
      \widehat{{\mathsf{V}}}^{{\mathsf{\nu}}}(1, 1+2 {\mathsf{\delta}}, r)
& = & {\mathsf{\nu}}^{-1}
      \left( {\mathsf{\nu}}(1) \cdot (1-r)
             +
             {\mathsf{\nu}}(1+ 2 {\mathsf{\delta}}) \cdot r
      \right)\, ,        
\end{eqnarray*}
}
and
{
\small
\begin{eqnarray*}
      \widehat{{\mathsf{V}}}^{{\mathsf{\nu}}}(1, 2, q)
& = & {\mathsf{\nu}}^{-1}
      \left( {\mathsf{\nu}}(1) \cdot (1-q)
             +
             {\mathsf{\nu}}(2) \cdot q
      \right)\, .
\end{eqnarray*}
}
Thus,
{
\small
\begin{eqnarray*}
      \widehat{{\mathsf{V}}}^{{\mathsf{\nu}}}(1, 1+2 {\mathsf{\delta}}, r)
& < & \widehat{{\mathsf{V}}}^{{\mathsf{\nu}}}(1, 2, q) 
\end{eqnarray*}
}
if and only if
{
\small
\begin{eqnarray*}
      r 
      \cdot
      \left( {\mathsf{\nu}}(1+ 2 {\mathsf{\delta}})
             -
             {\mathsf{\nu}}(1)
      \right)       
& < & q
      \cdot
      \left( {\mathsf{\nu}}(2) 
             -
             {\mathsf{\nu}}(1)
      \right)\, .       
\end{eqnarray*}
}
Set
${\mathsf{\Delta}}
 :=
 \frac{\textstyle 1}
      {\textstyle 4}$.
Since 
$r \leq q$,
${\mathsf{\nu}}$
is strictly increasing
and
$1 + 2 {\mathsf{\delta}} < 2$
for 
${\mathsf{\delta}}
 <
 {\mathsf{\Delta}}$,
the last inequality holds,
and we are done.

Combined with the choice of ${\mathsf{\Delta}}$
for {\sf (1)},
the required property
for the convex combination
${\mathsf{V}}
 =
 \lambda
 \cdot
 ({\mathsf{E}} +
  {\mathsf{\gamma}}
  \cdot
  {\mathsf{Var}})
 +
 (1 - \lambda)
 \cdot
 {\mathsf{V}}^{{\mathsf{\nu}}}$,
where ${\mathsf{\nu}}$
is increasing and strictly convex,
holds by choosing
${\mathsf{\Delta}}
 :=
 \min \left\{ \min   \left\{ \frac{\textstyle 1} 
                                  {\textstyle 4},
                             \frac{\textstyle 1}
                                  {\textstyle {\mathsf{\gamma}}}
                     \right\},
              \frac{\textstyle 1}
                      {\textstyle 4}
      \right\}
 =
 \min \left\{ \frac{\textstyle 1} 
                           {\textstyle 4},
                    \frac{\textstyle 1}
                            {\textstyle {\mathsf{\gamma}}}
      \right\}$.                                                         
\end{proof}

\noindent
Last,
for Condition {\sf (2/c)},
we use the
{\it Weak-Equilibrium-for-Expectation} property
to prove:

\begin{lemma} 
\label{no v equilibrium for crawford}
Fix an
$({\mathsf{E}}
  +
  {\mathsf{R}})$-valuation
${\mathsf{V}}$,  
where
{\sf (1)}
${\mathsf{R}}
 =
 {\mathsf{\gamma}}
 \cdot
 {\mathsf{Var}}$,
or
{\sf (2)}
${\mathsf{R}}
 =
 {\mathsf{\gamma}}
 \cdot
 {\mathsf{SD}}$,
or
{\sf (3)}
${\mathsf{V}}
 =
 \lambda\,
 ({\mathsf{E}} + {\mathsf{\gamma}} \cdot {\mathsf{Var}})
 +
 (1-\lambda)\,
 {\mathsf{V}}^{{\mathsf{\nu}}}$,
with $0 < {\mathsf{\lambda}} \leq 1$,
where ${\mathsf{\nu}}$
is increasing and strictly convex,
and with
${\mathsf{\gamma}} > 0$.
Then,
for any ${\mathsf{\delta}}$,
$0 < {\mathsf{\delta}} < 1$,
the Crawford game
${\mathsf{G}}_{C}(\delta)$
has no
${\mathsf{V}}$-equilibrium.
\end{lemma}

\begin{proof}
${\mathsf{G}}_{C}({\mathsf{\delta}})$
has no pure equilibrium
due to the following cycle of {\it improvement steps}:

\begin{tabular}{lllllllll}
$({\mathsf{f}}_{1}, {\mathsf{f}}_{1})$ & $\stackrel{1}{\longrightarrow}$ & 
$({\mathsf{f}}_{2}, {\mathsf{f}}_{1})$ & $\stackrel{2}{\longrightarrow}$ & 
$({\mathsf{f}}_{2}, {\mathsf{f}}_{2})$ & $\stackrel{1}{\longrightarrow}$ & 
$({\mathsf{f}}_{1}, {\mathsf{f}}_{2})$ & $\stackrel{2}{\longrightarrow}$ & 
$({\mathsf{f}}_{1}, {\mathsf{f}}_{1})$                                         \\
$\langle 1+{\mathsf{\delta}}, 1 + {\mathsf{\delta}} \rangle$ &      &
$\langle 1, 1 + 2\, {\mathsf{\delta}} \rangle$ &                    &
$\langle 1 + 2\, {\mathsf{\delta}}, 1 \rangle$ &                    &
$\langle 1+{\mathsf{\delta}}, 1 + {\mathsf{\delta}} \rangle$ &      &
                                                                               \\
\end{tabular}

\noindent
Assume,
by way of contradiction,
that
${\mathsf{G}}_{C}(\delta)$
has a mixed ${\mathsf{V}}$-equilibrium ${\bf p}$.
Denote
$p_{1}
 :=
 \langle x, 1-x
 \rangle$
and
$p_{2}
 :=
 \langle y, 1-y \rangle$.
Note that
{
\small
\begin{eqnarray*}
      {\mathsf{E}}_{1}\left( p_{1}^{{\mathsf{f}}_{1}},
                             p_{2}
                      \right)       
& = & y \cdot (1+{\mathsf{\delta}})
      +
      (1-y) \cdot 1\
      \
      =\
      \
      1 + y \cdot {\mathsf{\delta}}\, ,
\end{eqnarray*}
}
and
{
\small
\begin{eqnarray*}
      {\mathsf{E}}_{1}\left( p_{1}^{{\mathsf{f}}_{2}},
                             p_{2}
                      \right)       
& = & y \cdot 1 (1+{\mathsf{\delta}})
      +
      (1-y) \cdot (1+{\mathsf{\delta}})\
      \
      =\
      \
      1 + 2 {\mathsf{\delta}} - 2 y \cdot {\mathsf{\delta}}\, .
\end{eqnarray*}      
}
By the {\it Weak-Equilibrium-for-Expectation} property,
{
\small
\begin{eqnarray*}
{\mathsf{E}}_{1}\left( \langle p_{1}^{{\mathsf{f}}_{1}},
                               p_{2}
                        \rangle
                 \right)
& = &
 {\mathsf{E}}_{1}\left( \langle p_{1}^{{\mathsf{f}}_{2}},
                                p_{2}
                        \rangle
                 \right)\, ,
\end{eqnarray*}
}
or
$1 + y \cdot {\mathsf{\delta}}
 =
 1 + 2 {\mathsf{\delta}} - 2 y \cdot {\mathsf{\delta}}$,
yielding
$y =
 \frac{\textstyle 2}
      {\textstyle 3}$.
So,
$p_{2}
 =
 \left\langle \frac{\textstyle 2}
              {\textstyle 3},
         \frac{\textstyle 1}
              {\textstyle 3}
 \right\rangle$.

\noindent
\underline{Case {\sf (3)}:}
Clearly,
{
\small
\begin{eqnarray*}
      {\mathsf{E}}_{1}({\bf p})
& = & {\mathsf{E}}_{1}\left( \langle p_{1}^{{\mathsf{f}}_{1}},
                                     p_{2}
                             \rangle
                      \right)\
          \  
         =\
         \
         \frac{\textstyle 1}
                {\textstyle 3}
      \cdot
      (3 + 2 {\mathsf{\delta}})\, ,
\end{eqnarray*}
\begin{eqnarray*}      
      {\mathsf{Var}}_{1}({\bf p}) 
& = & x \cdot \frac{\textstyle 2}
                   {\textstyle 3} \cdot (1 + {\mathsf{\delta}})^{2}
      +
      x \cdot \frac{\textstyle 1}
                   {\textstyle 3} \cdot 1^{2}
      +
      (1-x) \cdot \frac{\textstyle 2}
                       {\textstyle 3} \cdot 1^{2}
      +
      (1-x) \cdot \frac{\textstyle 1}
                       {\textstyle 3} \cdot (1 + 2 {\mathsf{\delta}})^{2}  
      -
      \left(  \frac{\textstyle 1}
                   {\textstyle 3}
              \cdot
              (3 + 2 {\mathsf{\delta}})
      \right)^{2}                                                                                         \\
& = & \frac{\textstyle 2 {\mathsf{\delta}}^{2}}
           {\textstyle 3}
      \cdot
      \left( \frac{\textstyle 4}
                  {\textstyle 3}
             -
             x
      \right)\, ,
\end{eqnarray*}
}
and
{
\small
\begin{eqnarray*}      
       {\mathsf{V}}_{1}^{{\mathsf{\nu}}}({\bf p})
& = & {\mathsf{\nu}}^{-1}
      \left( x \cdot \frac{\textstyle 2}
                          {\textstyle 3} \cdot {\mathsf{\nu}} (1 + {\mathsf{\delta}})
             +
             (1-x) \cdot \frac{\textstyle 1}
                              {\textstyle 3} \cdot {\mathsf{\nu}} (1)
             +
             x \cdot \frac{\textstyle 2}
                          {\textstyle 3} \cdot {\mathsf{\nu}} (1) 
             +
             (1-x) \cdot \frac{\textstyle 1}
                              {\textstyle 3} \cdot {\mathsf{\nu}}(1 + 2 {\mathsf{\delta}})
      \right)\, ,
\end{eqnarray*}
}
so that
{
\small
\begin{eqnarray*}
&   & {\mathsf{V}}_{1}({\bf p})                                                          \\
& = & {\mathsf{\lambda}} \cdot \left( {\mathsf{E}}_{1}({\bf p})
                                      +
                                      {\mathsf{\gamma}}
                                      \cdot
                                      {\mathsf{Var}}_{1}({\bf p})
                    \right)
      +
      (1 - {\mathsf{\lambda}})
      \cdot
      {\mathsf{V}}_{1}^{{\mathsf{\nu}}}({\bf p})                                         \\
& = & {\mathsf{\lambda}}
      \cdot
      \left( \frac{\textstyle 1}
                  {\textstyle 3}
             \cdot
             (3 + 2 {\mathsf{\delta}})
             +
             {\mathsf{\gamma}}
             \cdot
             \frac{\textstyle 2 {\mathsf{\delta}}^{2}}
                  {\textstyle 3}
             \cdot
             \left( \frac{\textstyle 4}
                         {\textstyle 3}
                    - x
             \right)
      \right)
      +
      (1 - {\mathsf{\lambda}})
      \cdot                                                                                               \\
&   & \cdot
      {\mathsf{\nu}}^{-1}
      \left( x \cdot \frac{\textstyle 2}
                          {\textstyle 3} \cdot {\mathsf{\nu}} (1 + {\mathsf{\delta}})
             +
             (1-x) \cdot \frac{\textstyle 1}
                              {\textstyle 3} \cdot {\mathsf{\nu}} (1)
             +
             x \cdot \frac{\textstyle 2}
                          {\textstyle 3} \cdot {\mathsf{\nu}} (1) 
             +
             (1-x) \cdot \frac{\textstyle 1}
                              {\textstyle 3} \cdot {\mathsf{\nu}}(1 + 2 {\mathsf{\delta}})
      \right)\, .
\end{eqnarray*}
}
Since
{\it (i)}
$p_{1}$
is a
${\mathsf{V}}_{1}$-best-response
to
$p_{2}$
and
{\it (ii)}
${\mathsf{V}}_{1}$
is concave
in the mixed strategy $p_{1}$,
the {\it Optimal-Value} property
(Proposition~\ref{constant value})
implies that
there is a constant $A$
such that
${\mathsf{V}}_{1}({\bf p})
 =
 A$
for all
$x \in [0, 1]$.
Since ${\mathsf{\delta}} \neq 0$,
this yields a contradiction
for $\lambda = 1$.
So
assume $0 < \lambda < 1$.
Then,
for all
$x \in [0, 1]$,
{
\small
\begin{eqnarray*}
&   & A                                                                                           \\
& = & {\mathsf{\lambda}}
      \cdot
      \left( \frac{\textstyle 1}
                  {\textstyle 3}
             \cdot
             (3 + 2 {\mathsf{\delta}})
             +
             {\mathsf{\gamma}}
             \cdot
             \frac{\textstyle 2 {\mathsf{\delta}}^{2}}
                  {\textstyle 3}
             \cdot
             \left( \frac{\textstyle 4}
                         {\textstyle 3}
                    - x
             \right)
      \right)                                                                        
      +
      (1 - {\mathsf{\lambda}})
      \cdot                                                                                      \\ 
&   & \cdot
      {\mathsf{\nu}}^{-1}
      \left( x \cdot \frac{\textstyle 2}
                          {\textstyle 3} \cdot {\mathsf{\nu}} (1 + {\mathsf{\delta}})
             +
             (1-x) \cdot \frac{\textstyle 1}
                              {\textstyle 3} \cdot {\mathsf{\nu}} (1)
             +
             x \cdot \frac{\textstyle 2}
                          {\textstyle 3} \cdot {\mathsf{\nu}} (1) 
             +
             (1-x) \cdot \frac{\textstyle 1}
                              {\textstyle 3} \cdot {\mathsf{\nu}}(1 + 2 {\mathsf{\delta}})
      \right)\, .
\end{eqnarray*}
}
\noindent
Rearranging yields
${\mathsf{\nu}}(c_{1} + d_{1} x)
 =
 c_{2} + d_{2} x$,
for some constants
$c_{1}$, $d_{1}$,
$c_{2}$ and $d_{2}$
with
$d_{1}
 :=
 \frac{\textstyle {\mathsf{\lambda}}}
      {\textstyle {\mathsf{\lambda}} - 1}
 \cdot
 \frac{\textstyle 2\, \delta^{2}}
      {\textstyle 3}
 \neq
 0$,
for all
$x \in [0, 1]$.
Hence,
{
\small
\begin{eqnarray*}
{\mathsf{\nu}}(x)
& = & c_{2}
      -
      c_{1} \cdot \frac{\textstyle d_{2}}
                       {\textstyle d_{1}}
      +
      \frac{\textstyle d_{2}}
           {\textstyle d_{1}}
      \cdot
      x
\end{eqnarray*}
}
for all $x \in [0, 1]$,
so that
${\mathsf{\nu}}$
is not strictly convex.
A contradiction.

\noindent
\underline{Case {\sf (1)}:}
This is the special case
of Case {\sf (3)}
with ${\mathsf{\lambda}} = 1$.

\noindent
\underline{Case {\sf (2)}:}
Note that
{
\small
\begin{eqnarray*}
      \left( {\mathsf{E}}_{1} + {\mathsf{\gamma}}
                                \cdot {\mathsf{SD}}_{1}
      \right)
      ({\bf p})
& = & \frac{\textstyle 1}
           {\textstyle 3}
      \cdot
      (3 + 2 {\mathsf{\delta}})
      +
      {\mathsf{\gamma}}
      \cdot
      {\mathsf{\delta}}
      \cdot
      \sqrt{\frac{\textstyle 2}
                 {\textstyle 3}
            \cdot
            \left( \frac{\textstyle 4}
                        {\textstyle 3}
                   -x
            \right)}\, .
\end{eqnarray*}
}
Since
{\it (i)}
$p_{1}$
is a
$({\mathsf{E}}_{1} + {\mathsf{\gamma}} 
                     \cdot {\mathsf{SD}}_{1})$-best-response
to
$p_{2}$
and
{\it (ii)}
${\mathsf{E}}_{1} + {\mathsf{\gamma}} \cdot {\mathsf{SD}}_{1}$
is concave in the mixed strategy $p_{1}$,
the {\it Optimal-Value} property
(Proposition~\ref{constant value})
implies that
there is a constant $A$
such that
$({\mathsf{E}}_{1} + {\mathsf{\gamma}} \cdot {\mathsf{SD}}_{1})
 ({\bf p})
 =
 A$
for all
$x \in [0, 1]$.
Since ${\mathsf{\delta}} \neq 0$,
this yields a contradiction.
\end{proof}

\noindent
Now,
for Condition {\sf (2)}
in Theorem~\ref{two players complexity},
choose ${\mathsf{\delta}}$
as a rational number 
no larger than the minimum
of the ${\mathsf{\Delta}}$
from Lemma~\ref{little tiny}
and the ${\mathsf{\Delta}}$
from Lemma~\ref{the monotonicity lemma},
both of which are polynomial time computable
for the $({\mathsf{E}} +
          {\mathsf{R}})$-valuations
in Theorem~\ref{call from paderborn}.
This choice guarantees that
all Conditions
{\sf (2/a)},
{\sf (2/b)} and
{\sf (2/c)}
in Theorem~\ref{two players complexity}
are satisfied
by the chosen ${\mathsf{\delta}}$.
Hence,
it follows:

\begin{theorem}
\label{call from paderborn}
Fix an
$({\mathsf{E}}
  +
  {\mathsf{R}})$-valuation
${\mathsf{V}}$,  
where
{\sf (1)}
${\mathsf{R}}
 =
 {\mathsf{\gamma}}
 \cdot
 {\mathsf{Var}}$,
or
{\sf (2)}
${\mathsf{R}}
 =
 {\mathsf{\gamma}}
 \cdot
 {\mathsf{SD}}$,
or
{\sf (3)}
${\mathsf{V}}
 =
 \lambda\,
 ({\mathsf{E}} + 
 {\mathsf{\gamma}}
 \cdot
 {\mathsf{Var}})
 +
 (1-\lambda)\,
 {\mathsf{V}}^{{\mathsf{\nu}}}$,
with $0 < \lambda \leq 1$,
where
${\mathsf{\nu}}(x)
 =
 x^{r}$
with $r \geq 2$,
and with ${\mathsf{\gamma}} > 0$.
Then,
{\sf $\exists {\mathsf{V}}$-EQUILIBRIUM}
is strongly ${\mathcal{NP}}$-hard
for 2-players games.
\end{theorem}

\section{Epilogue}
\label{epilogue}

\noindent
We embarked on a research direction
making different-from-classical assumptions
on the behavior of the players
in order to model the real-world in a more accurate way
and gain insights that explain reality better. 
We have developed a framework for games
with players minimizing an
$({\mathsf{E}}
 +
 {\mathsf{R}})$-valuation ${\mathsf{V}}$.
Our framework enabled
proving the strong ${\mathcal{NP}}$-hardness
of {\sf $\exists {\mathsf{V}}$-EQUILIBRIUM}
in the simplest cases of games
with two strategies
or two players,
respectively,
and for many significant choices of
$({\mathsf{E}}
 +
 {\mathsf{R}})$-valuations ${\mathsf{V}}$.

Besides these central results,
our study is making a number of 
additional conceptual and technical contributions
through introducing
several new analytical and combinatorial tools and techniques,
which are
of wider applicability
and interest;
we summarize here the main ones.
\begin{itemize}

\item
Our proof techniques have relied heavily
on the
{\it Weak-Equilibrium-for-Expectation} property,
which indicates its computational power.

\item
We introduced ${\mathsf{E}}$-strict concavity
as the most general known class
of valuations
with the {\it Weak-Equilibrium-for-Expectation} property.

\item
We imported
${\mathsf{\nu}}$-valuations
from 
Actuarial Risk Theory~\cite{KGDD08}
into the realm of equilibrium computation
and revealed some of their algorithmic properties;
ditto
for higher moments,
generalizing the variance
considered before in~\cite{FP10,MM12},
which were used before
as risk-modeling valuations
in Portfolio Theory~\cite{KL76}.

\item
We established the {\it Mixed-Player-Has-Pure-Neighbors} property,
which explicitly identifies
a class of games 
and a corresponding class of valuations
where mixed equilibria get ``endangered''
(cf.~\cite{C90,FP10}).

\end{itemize}
Our work opens up
a wide avenue
for future research
towards revealing
the complexity of 
{\sf $\exists {\mathsf{V}}$-EQUILIBRIUM}
for other valuations ${\mathsf{V}}$.
Most obviously,
the modularity of the proof
of Theorem~\ref{two players complexity}
may allow,
similarly to Theorem~\ref{call from paderborn},
direct derivation of new concrete complexity results
for other valuations
that will be shown to have the assumed properties.
Enhancing the class of ${\mathsf{E}}$-strictly concave valuations
may yield such new valuations,
which may be directly accomodated
into the general framework
we developed.
But new tools and techniques
may be required
for settling the complexity of
{\sf $\exists {\mathsf{V}}$-EQUILIBRIUM}
when ${\mathsf{V}}$ is not concave
or not ${\mathsf{E}}$-strictly concave,
or when it lacks
the {\it Weak-Equilibrium-for-Expectation} property.
It is not clear
whether and how
our framework could be extended to accomodate
even {\it quasiconcave} valuations.
We conclude with such an example, 
cast into the context 
of maximization games.
The {\it Sharpe ratio} valuation,
formulated as
${\mathsf{SR}}
 =
 \frac{\textstyle {\mathsf{E}}}
      {\textstyle 1 + {\mathsf{SD}}}$~\cite{S63},
is the ratio of a convex 
over a concave function;
although it is not convex,
it is {\it quasiconvex} 
as shown in~\cite{SRF07}.
Does ${\mathsf{SR}}$ 
have the
{\it Weak-Equilibrium-for-Expectation} property?
What is the complexity of 
{\sf $\exists {\mathsf{SR}}$-EQUILIBRIUM}?

\newpage
\pagenumbering{roman}

\appendix

\section{${\mathsf{FP}}$-Strict Concavity}
\label{fp strict concavity}

\noindent
We recall the following definition from~\cite[Section 2]{FP10},
rendering their notation and vocabulary.

\begin{definition}[${\mathsf{FP}}$-Strict Convexity]
\label{fp strict convexity}
Fix a player $i \in [n]$.
The valuation ${\mathsf{V}}_{i}$
is {\it ${\mathsf{FP}}$-strictly convex}
if for any {\em fixed}
partial mixed profile
${\bf x}_{-i}$,
for any pair of payoff distributions
${\mathsf{P}}_{i}(x_{i}^{\prime}, {\bf x}_{-i})
 \neq
 {\mathsf{P}}_{i}(x_{i}^{\prime\prime}, {\bf x}_{-i})$,
where  
$x_{i}^{\prime}$
and
$x_{i}^{\prime\prime}$ are mixed strategies,
and for any $\alpha$ with
$0 < \alpha < 1$,
it holds that
\begin{eqnarray*}
      {\mathsf{V}}_{i} \left( \alpha  x_{i}^{\prime}
                              +
                              (1 - \alpha) x_{i}^{\prime\prime},
                              {\bf x}_{-i}
                       \right)
& < & \alpha
      {\mathsf{V}}_{i} \left( x_{i}^{\prime},
                             {\bf x}_{-i}
                       \right)
      +
      (1 - \alpha)
      {\mathsf{V}}_{i} \left( x_{i}^{\prime\prime},
                              {\bf x}_{-i}
                       \right)\, .
\end{eqnarray*}
\end{definition}

\noindent
We quote from~\cite[Section 2]{FP10}
that the {\it payoff distribution}
${\mathsf{P}}_{i}
 =
 {\mathsf{P}}_{i}({\bf p})
 $
is the probability distribution
induced from a mixed profile ${\bf p}$
on the range of possible payoffs
for player $i$
over all profiles;
so,
the probability that
player $i$ receives payoff $a$ is
$\sum_{{\bf s} \in S
       \mid
       {\mathsf{\mu}}_{i}({\bf s})
       =
       a}
   {\bf p}({\bf s})$.    
Recall also that
${\mathsf{V}}_{i}$
is ${\mathsf{FP}}$-strictly concave
if
$-{\mathsf{V}}_{i}$
is ${\mathsf{FP}}$-strictly convex.
We observe:

\begin{observation}
\label{fp counterexample}
The valuation
${\mathsf{V}}
 =
 {\mathsf{E}}
 -
 {\mathsf{Var}}$
is {\em not}  
${\mathsf{FP}}$-strictly convex.
\end{observation}

\begin{proof}
By counterexample.
Consider the game ${\mathsf{G}}$
with two players $1$ and $2$,
with
$S_{1} = \{ s_{1}, s_{2} \}$
and
$S_{2} = \{ t_{1}, t_{2}, t_{3}, t_{4} \}$.
The utilities for player $1$
are given by
\begin{eqnarray*}
      p_{1}(s_{1}, t_{j})
& = & \left\{ \begin{array}{ll}
                \frac{\textstyle 9}
                     {\textstyle 2}\, ,   & \mbox{if $j=1$}               \\
                                          &                               \\ 
                \frac{\textstyle 7}
                     {\textstyle 2}\, ,   & \mbox{if $j = 2$}             \\
                                          &                               \\ 
                0\, ,                     & \mbox{if $j \in \{ 3, 4 \}$} 
              \end{array}  
      \right.\, ,
\end{eqnarray*} 
and
\begin{eqnarray*}
      p_{1}(s_{2}, t_{j})   
& = & \left\{ \begin{array}{ll}
                0\, ,                     & \mbox{if $j \in \{ 1, 2 \}$}  \\
                                          &                               \\ 
                5\, ,                     & \mbox{if $j = 3$}             \\
                                          &                               \\
                \frac{\textstyle 15}
                     {\textstyle 4}\, ,   & \mbox{if $j = 4$}
              \end{array}
      \right.\, .                                         
\end{eqnarray*}

\noindent
Player $2$ has chosen the
mixed strategy
$x_{2} = \left\langle \frac{\textstyle 1}
                           {\textstyle 4},
                      \frac{\textstyle 1}
                           {\textstyle 4},
                      \frac{\textstyle 1}
                           {\textstyle 10},
                      \frac{\textstyle 2}
                           {\textstyle 5}
         \right\rangle$.                                 
For player $1$,
$x_{1}^{\prime}$ is the pure strategy $s_{1}$;
$x_{1}^{\prime\prime}$ is the pure strategy $s_{2}$. 
Then,
clearly,
the induced payoff distributions
${\mathsf{P}}_{1}(x_{1}^{\prime}, x_{2})$
and
${\mathsf{P}}_{1}(x_{1}^{\prime\prime}, x_{2})$
are different.

\noindent
Note that
for every mixed strategy $x_{1}$
of player $1$,
\begin{eqnarray*}
      {\mathsf{V}}_{1}(x_{1}, x_{2})
& = & {\mathsf{E}}_{1}(x_{1}, x_{2})
      -
      {\mathsf{Var}}_{1}(x_{1}, x_{2})                                   \\
& = & {\mathsf{E}}_{1}(x_{1}, x_{2}) 
      -
      \widehat{{\mathsf{E}}}_{1}(x_{1}, x_{2}) 
      +
      ({\mathsf{E}}_{1}(x_{1}, x_{2}))^{2}\, ,    
\end{eqnarray*}
where 
$\widehat{{\mathsf{E}}}_{1}$
denotes the expectation
for the square game
${\mathsf{G}}^{2}$
(where the utilities are
the squares of the utilities
for ${\mathsf{G}}$).
It is
\begin{eqnarray*}
      {\mathsf{E}}_{1}(x_{1}^{\prime}, x_{2})  
& = & \frac{\textstyle 1}
           {\textstyle 4}
      \cdot
      \frac{\textstyle 9}
           {\textstyle 2}     
      +
      \frac{\textstyle 1}
           {\textstyle 4}
      \cdot
      \frac{\textstyle 7}
           {\textstyle 2}\ \
      =\ \
      2\, ,                                                           \\
      {\mathsf{E}}_{1}(x_{1}^{\prime\prime}, x_{2})  
& = & \frac{\textstyle 1}
           {\textstyle 10}
      \cdot
      5
      +
      \frac{\textstyle 2}
           {\textstyle 5}
      \cdot
      \frac{\textstyle 15}
           {\textstyle 4}\ \
      =\ \
      2\, ,                                 
\end{eqnarray*}
and
\begin{eqnarray*}
      \widehat{{\mathsf{E}}}_{1}(x_{1}^{\prime}, x_{2})  
& = & \frac{\textstyle 1}
           {\textstyle 4}
      \cdot
      \left( \frac{\textstyle 9}
                  {\textstyle 2}
      \right)^{2}                 
      +
      \frac{\textstyle 1}
           {\textstyle 4}
      \cdot
      \left( \frac{\textstyle 7}
                  {\textstyle 2}
      \right)^{2}\ \
      =\ \
      8 +
      \frac{\textstyle 1}
           {\textstyle 8}\, ,                                                           \\
      \widehat{{\mathsf{E}}}_{1}(x_{1}^{\prime\prime}, x_{2})  
& = & \frac{\textstyle 1}
           {\textstyle 10}
      \cdot
      5^{2}
      +
      \frac{\textstyle 2}
           {\textstyle 5}
      \cdot
      \left( \frac{\textstyle 15}
                  {\textstyle 4}
      \right)^{2}\ \
      =\ \
      8 + \frac{\textstyle 1}
               {\textstyle 8}\, .                                 
\end{eqnarray*}
Thus,
\begin{eqnarray*}
      {\mathsf{V}}_{1}(x_{1}^{\prime}, x_{2})  
& = & {\mathsf{V}}_{1}(x_{1}^{\prime\prime}, x_{2})\, .  
\end{eqnarray*}

\noindent
Now, 
set
$x_{1}
 =
 \alpha x_{1}^{\prime}
 +
 (1 - \alpha) x_{1}^{\prime\prime}$,
where
$0 < \alpha < 1$.
By linearity of expectation,
\begin{eqnarray*}
      {\mathsf{E}}_{1}(x_{1}, x_{2}) 
& = & {\mathsf{E}}_{1}(x_{1}^{\prime}, x_{2})\ \
      =\ \
      {\mathsf{E}}_{1}(x_{1}^{\prime\prime}, x_{2})
\end{eqnarray*}
and
\begin{eqnarray*}
      \widehat{{\mathsf{E}}}_{1}(x_{1}, x_{2}) 
& = & \widehat{{\mathsf{E}}}_{1}(x_{1}^{\prime}, x_{2})\ \
      =\ \
      \widehat{{\mathsf{E}}}_{1}(x_{1}^{\prime\prime}, x_{2})\, .
\end{eqnarray*}
Hence,
\begin{eqnarray*}
      {\mathsf{V}}_{1}(x_{1}, x_{2}) 
& = & {\mathsf{V}}_{1}(x_{1}^{\prime}, x_{2})\ \
      =\ \
      {\mathsf{V}}_{1}(x_{1}^{\prime\prime}, x_{2})\, .
\end{eqnarray*}
So,
${\mathsf{V}}_{1}$
is {\em not}  
${\mathsf{FP}}$-strictly convex.
\end{proof}


\end{document}